\documentclass[10pt]{article}
\usepackage{amsfonts,color}
 \usepackage{amsmath,amssymb}
 \usepackage{epsfig}
\usepackage{lscape}
\usepackage{amssymb}
\usepackage{graphicx}
 \usepackage{footnote}
\usepackage[utf8]{inputenc}	
\usepackage{mathtools}
\usepackage{amsthm,amssymb,wasysym}
\usepackage{amsfonts}
\usepackage[english]{babel}
\usepackage{tikz}
\usepackage{empheq}
\usepackage{framed}
\usepackage{bbm}
\newcommand{\id}{{\boldsymbol{\mathbbm{1}}}}

\newcommand*\widefbox[1]{\fbox{\hspace{2em}#1\hspace{2em}}}

\tikzset{square arrow/.style={to path={-- ++(0,-.25) -| (\tikztotarget)}}}

\allowdisplaybreaks[1]

\makeindex

 \newtheorem{theorem}{Theorem}[section]
 \newtheorem{lemma}[theorem]{Lemma}
 \newtheorem{remark}[theorem]{Remark}
 \newtheorem{proposition}[theorem]{Proposition}
 \newtheorem{corollary}[theorem]{Corollary}

 \newcommand{\R}{\mathbb{R}}

\DeclareMathOperator{\sym}{sym}
\DeclareMathOperator{\skw}{skew}
\DeclareMathOperator{\tr}{tr}

\DeclareMathOperator{\axl}{axl}
\DeclareMathOperator{\anti}{anti}
\DeclareMathOperator{\dev}{dev}
\DeclareMathOperator{\sL}{\mathfrak{sl}}
\DeclareMathOperator{\so}{\mathfrak{so}}
\DeclareMathOperator{\gl}{\mathfrak{gl}}

\newcommand{\Co}{C_0^{\infty}(\Omega)}

\DeclareMathOperator{\Curl}{Curl\,}
\newcommand{\Sym}{ {\rm{Sym}} }

\def\dv{\textrm{div}}
\def\Div{\textrm{Div\,}}
\def\curl{\textrm{curl\,}}

\def\skew{\text{skew}}

\def\dv{\textrm{div}}

\def\dd{\displaystyle}

\makeatletter
\let\@fnsymbol\@arabic

\begin{document}
\title{A variant of  the linear isotropic indeterminate couple stress model with  symmetric local force-stress, symmetric nonlocal  force-stress,  symmetric couple-stresses and complete traction boundary conditions}
\author{\normalsize{Ionel-Dumitrel Ghiba\thanks{Ionel-Dumitrel Ghiba, \ \ \ \ Lehrstuhl f\"{u}r Nichtlineare Analysis und Modellierung, Fakult\"{a}t f\"{u}r Mathematik,
Universit\"{a}t Duisburg-Essen, Thea-Leymann Str. 9, 45127 Essen, Germany;  Alexandru Ioan Cuza University of Ia\c si, Department of Mathematics,  Blvd.
Carol I, no. 11, 700506 Ia\c si,
Romania;  Octav Mayer Institute of Mathematics of the
Romanian Academy, Ia\c si Branch,  700505 Ia\c si; and Institute of Solid Mechanics, Romanian Academy, 010141 Bucharest, Romania, email: dumitrel.ghiba@uni-due.de, dumitrel.ghiba@uaic.ro}  \quad
and \quad
Patrizio Neff\thanks{Corresponding author: Patrizio Neff,  \ \ Head of Lehrstuhl f\"{u}r Nichtlineare Analysis und Modellierung, Fakult\"{a}t f\"{u}r
Mathematik, Universit\"{a}t Duisburg-Essen,  Thea-Leymann Str. 9, 45127 Essen, Germany, email: patrizio.neff@uni-due.de}\quad
and \quad
Angela Madeo\footnote{Angela Madeo, \ \  Laboratoire de G\'{e}nie Civil et Ing\'{e}nierie Environnementale,
Universit\'{e} de Lyon-INSA, B\^{a}timent Coulomb, 69621 Villeurbanne
Cedex, France; and International Center M\&MOCS ``Mathematics and Mechanics
of Complex Systems", Palazzo Caetani,
Cisterna di Latina, Italy,
 email: angela.madeo@insa-lyon.fr}\quad and \quad Ingo M\"unch\thanks{Ingo M\"unch, Institute for Structural Analysis, Karlsruhe Institute of Technology, Kaiserstr. 12, 76131 Karlsruhe,
Germany, email: ingo.muench@kit.edu}}
}
\maketitle

\begin{center}
{\it Dedicated to Richard Toupin, in deep admiration of his scientific achievements. }
\end{center}

\begin{abstract}
In this paper we venture  a new look at the linear isotropic indeterminate couple stress model in the general framework of second  gradient elasticity and we propose a new alternative formulation which obeys Cauchy-Boltzmann's axiom of the symmetry of the force stress tensor.  For this model we prove the existence of  solutions for the equilibrium problem. Relations with other gradient elastic theories and the possibility to switch from a {4th order} (gradient elastic) problem to a 2nd order micromorphic model are also discussed with a view of obtaining symmetric force-stress tensors. It is shown that the indeterminate couple stress model can be written entirely with symmetric force-stress and symmetric couple-stress. The difference of the alternative models rests in specifying traction boundary conditions of either rotational type or strain type. If rotational type boundary conditions are used in the partial integration, the classical anti-symmetric nonlocal force stress tensor formulation is obtained. Otherwise, the difference in both formulations is only a divergence--free second order stress field such that the field equations are the same, but the traction boundary conditions are different. For these results we employ a novel integrability condition, connecting the infinitesimal  continuum rotation and the infinitesimal continuum strain. Moreover, we provide  the complete, consistent traction boundary conditions for both models.
\\
\vspace*{0.25cm}
\\
{\bf{Key words:}}  symmetric Cauchy stresses,  generalized continua, non-polar material, microstructure, size effects, microstrain model, non-smooth solutions, gradient elasticity, strain gradient elasticity, couple stresses, polar continua, hyperstresses, Boltzman axiom, dipolar gradient model, modified couple stress model, conformal invariance, micro-randomness, symmetry of couple stress tensor, consistent traction boundary conditions.
\\
\vspace*{0.25cm}
\\
{\bf{AMS 2010 subject classification:} } 74A30, 74A35.
\end{abstract}

\newpage

\tableofcontents

\newpage

\section{Introduction}

\subsection{General viewpoint}

The Cosserat model is an extended continuum model which features independent degrees of rotation in addition to the standard translational degrees of particles, see \cite{Eringen99,Neff_Jeong_bounded_stiffness09,NeffGhibaMicroModel,MauginVirtualPowers,maugin1980method} for a detailed exposition. The prize, which has to be paid for this extension are non-symmetric force stress tensors together with so-called couple stress tensors which then represent the response of the model due to spatially differing Cosserat rotations. The couple stress model is the Cosserat model \cite{Neff_JeongMMS08} with restricted rotations, i.e. in which the Cosserat rotations coincide with the continuum rotations. As such it belongs also to a certain subclass of gradient elasticity models\footnote{Le Roux \cite{Roux11} seems to give for the first time a second gradient theory in linear elasticity using a variational formulation \cite{maugin2010generalized,maugin2014continuum}.}, where the higher derivatives only act on the continuum rotations. This constitutes a big conceptual advantage since the interpretation of the Cosserat rotations as new physical degrees of freedom is  in general a difficult task. Such a model is also called a model with ``latent microstructure" \cite{Capriz85,Capriz89}.

Let $F=R\,U$ be the polar decomposition of the deformation gradient $F=\nabla \varphi$ into rotation $R\in {\rm SO}(3)$ and positive definite symmetric right stretch tensor $U=\sqrt{F^TF}$, where $\varphi:\Omega\subset \mathbb{R}^3\rightarrow\mathbb{R}^3$ characterizes  the deformation of the material filling the domain $\Omega\subseteq \mathbb{R}^3$. We write  $R={\rm polar}(F)$. In a variational context, the energy density $W$ to be minimized in the geometrically nonlinear constrained Cosserat model is given by
\begin{align}
W=W(\underbrace{U-\id}_{\text{strain}}, \underbrace{{\rm polar}(F)^T\nabla_{\rm x}{\rm polar}(F)}_{\text{curvature}})\,,
\label{energiepolar}
\end{align}
which reduced form follows from  left-invariance of the Lagrangian $W$ under superposed rotations. In this paper, our objectives are much more modest. We will only be concerned with  the linearized variant of \eqref{energiepolar}, which can be written as
\begin{align}
W=W(\!\!\!\!\!\!\underbrace{\sym \nabla u}_{\begin{array}{c}
\text{\footnotesize{infinitesimal}}\vspace{-1.5mm}\\
\text{\footnotesize{strain}}
\end{array}}\!\!\!\!\!\!,\,\underbrace{\nabla [\axl(\skew \nabla u)]}_{\begin{array}{c}
\text{\footnotesize{infinitesimal}}\vspace{-1.5mm}\\
\text{\footnotesize{curvature}}
\end{array}})=W_{\rm lin}(\sym \nabla u)+W_{\rm curv}(\nabla[\axl(\skew \nabla u)]),
\label{energiepolar1}
\end{align}
where   $u:\Omega\subset \mathbb{R}^3\rightarrow\mathbb{R}^3$ is the displacement and
\begin{align}
\nabla\axl(\skew \nabla u)=2\,\curl u.
 \end{align}
 The  energy density \eqref{energiepolar1} is the classical Lagrangian for the indeterminate couple stress formulation. As will be seen later, this formulation leads naturally to totally skew symmetric nonlocal force stress contributions.

Toupin already remarked on an alternative representation of the energy \eqref{energiepolar} \cite[Section 6]{Toupin62} which leads, in its linearized variant given by Mindlin \cite[eq.~(2.4)]{Mindlin62} to a dependence on
\begin{align}
W=W(\sym \nabla u, \Curl (\sym \nabla u))=W_{\rm lin}(\sym \nabla u)+W_{\rm curv}( \Curl (\sym \nabla u))\,,
\label{energiepolar2}
\end{align}
 due to the equivalence
 $$\nabla (\axl \skew \nabla u)=\frac{1}{2}\nabla \curl u=(\Curl (\sym \nabla u))^T$$
  instead of \eqref{energiepolar1}. The representation $\frac{1}{2}\nabla \curl u$ is directly derived from the original Cosserat  model \cite{Cosserat09,schaefer1967cosserat}.
Both authors, Toupin and Mindlin, noted that now, comparing \eqref{energiepolar1} and \eqref{energiepolar2} the force stress tensors and the couple stress tensors are changed while the balance of linear momentum equation remains unchanged such that these concepts are not uniquely defined (see also  Truesdell and Toupin's remark on null-tensors \cite[p.~~547]{Truesdell60}. However, they apparently did not realize that it is possible to use this ambiguity to obtain completely symmetric force stress tensors also in the couple stress model which is otherwise the paragon for a model having non-symmetric force stress tensors. We also need to remark that in a purely mechanical context, the observation of size-effects does not necessitate to introduce skew-symmetric stress-tensors \cite{AskesAifantis}.

 In this paper we do not discuss in detail the field of applications of such a special format of gradient elasticity model. Suffice it to say that much attention is directed to nano-scaled material in which size-effects may become important, which may make the present model applicable at strong stress gradients in the vicinity of cracks, or more generally, in highly heterogeneous media. We must also warn the reader: the indeterminate couple stress model is, in our view, a certain singular limit of the Cosserat model with independent displacements and micro-rotation and therefore some degenerate behaviour is to be expected throughout.

\subsection{The linear indeterminate couple stress model}
As hinted at above, the indeterminate couple stress model is a  specific gradient elastic model in which the higher order interaction is restricted to the continuum rotation $\skew\,\nabla u$ (or equivalently, $\curl u$). It is therefore traditionally interpreted  to include interactions of rotating particles and it is possible to prescribe boundary conditions of rotational type. Superficially, this is the simplest possible generalization of linear elasticity in order to include the gradient of the local continuum rotation as a source of stress and strain energy. In this paper, we limit our analysis to  linear isotropic materials and only to the second gradient\footnote{There is such a formula, which says that all second derivatives of $u$ can be obtained from linear combinations of partial derivatives of strain, i.e. $D^2 u={\rm Lin}(\nabla \sym \nabla u)$, $u_{k,ij}=\varepsilon_{ik,j}-\varepsilon_{jk,i}-\varepsilon_{ij,k}$, where $\varepsilon=\sym \nabla u$.} of the displacement \begin{equation*}
D^2u=u_{k,ij}=\varepsilon_{ik,j}-\varepsilon_{jk,i}-\varepsilon_{ij,k},\quad \text{where}\quad \varepsilon=\sym \nabla u.
 \end{equation*}
 In general, the strain gradient models have the great advantage of simplicity and physical transparency since there are no new independent degree of freedoms introduced which would require interpretation. Since in this model there are no additional degrees of freedom (as compared to the Cosserat or micromorphic approach) the higher derivatives introduce a ``latent-microstructure" (constrained microstructure). However, this apparent simplicity has to be payed  with much  more complicated traction boundary conditions, as will be seen later.

 We will see in Section \ref{sectcoss}, surprisingly,  that the mentioned rotational interaction can equivalently be viewed as a strain type interaction in the indeterminate couple stress model. Therefore, the first interpretation of rotational interaction (which is classical) is ambiguous as long as the problem is not specified together with boundary conditions appearing as effect of the kind of partial integration  which is performed. We may choose, contrarily to our intuition,   another representation of the  curvature energy motivated by formal considerations of invariance properties. In this regard we highlight the fact that  force stresses for a material of higher order are far from being uniquely defined: it is always possible to add a self-equilibrated (divergence-free tensor field) force field changing the constitutive stress tensor but leaving unaltered the equilibrium equations \cite{DunnSerrin,morro2014interstitial}.

Often, such kind of models introduce too  many additional parameters (or too many additional artificial degrees of freedom) which are neither easily interpreted, nor  easily to be determined from experiments. Our discussion may also be interpreted with the background to  only include those higher order terms that are required to describe the pertinent physics. It is clear that higher order models should not be more complicated than is warranted by experimental observation. A permanent nuisance in this respect is the question of how to identify new material parameters which are connected to the possible non-symmetry of the total force-stress  tensor having the same dimensions as the classical shear modulus $\mu$ ${\rm [N/mm^2]}$. In the Cosserat model the connecting parameter is the Cosserat-couple modulus $\mu_c$ \cite{Eringen99,Neff_Jeong_Conformal_ZAMM08}, which, for the indeterminate  couple stress model considered here,  is formally $\mu_c\rightarrow\infty$.

The Cauchy-Boltzmann axiom, well known from classical elasticity, requires the symmetry of the force stress tensor and may serve us also in the realm of this higher order theory to restrict the bewildering possibilities. Already Cauchy wrote \cite[p.~344-345]{cauchy1851note}:
 \begin{quote}{\it ``... les composantes $\mathcal{A}, \mathcal{F}, \mathcal{E}; \mathcal{F}, \mathcal{B}, \mathcal{D}; \mathcal{E}, \mathcal{D}, \mathcal{C}$ des
pressions support\'ees au point P par trois plans
parall\`{e}les aux plans coordonnés des $yz$, des $zx$
et des $xy$, pourront \^{e}tre g\'en\'eralement
consid\'er\'ees comme des fonctions linéaires des
d\'eplacements $\xi, \eta, \zeta$ et des leurs d\'eriv\'ees des
divers ordres."\footnote{Our translation: The components [of the symmetric total force stress tensor] $\mathcal{A}, \mathcal{F}, \mathcal{E}; \mathcal{F}, \mathcal{B}, \mathcal{D}; \mathcal{E}, \mathcal{D}, \mathcal{C}$  can be considered in general as linear functions $\xi, \eta, \zeta$  of the displacement  and their derivatives of arbitrary order.}.}
 \end{quote}Truesdell and Toupin \cite[p. 390]{Truesdell60}  write:
  \begin{quote}{\it ````Theories of elastic materials of grade 2 or higher had been proposed by several authors [Cauchy \cite{cauchy1851note}, St.~Venant \cite{st1869note}, Jaramillo \cite{Jaramillo}], but under the assumption that the [total-force] stress tensor is symmetric".
  }
 \end{quote}
Indeed, Jaramillo \cite{Jaramillo} considers a second gradient elastic material  and obtains the dynamic equations by Hamilton's principle. He observes dispersion relations in wave propagation problems. For simplicity  only he restricts his discussion to  those second gradient formulations, which give rise to a symmetric total force-stress tensor and obtains a  classification for isotropic materials \cite[p.~51, Eq.~(96)]{Jaramillo}. The subject was pushed forward in the late 1950's with works of Toupin~\cite{Toupin62,Toupin64}, Grioli~\cite{Grioli60,grioli1962mathematical}, Mindlin ~\cite{Mindlin64} and Koiter~\cite{Koiter64}, among others, see the references  later in this paper. Yang et al. \cite{Yang02} give an erroneous motivation for a symmetric moment stress tensor, as will be shown in  \cite{MunchYang}. Neff et al.~\cite{Neff_Jeong_IJSS09} considered the singular stiffening behaviour for arbitrary small samples in the Cosserat and indeterminate couple stress model and concluded that in order to avoid these singular effects one has to take a symmetric moment stress, thus providing the first rational argument in favour of symmetric moment stresses.  In \cite{Neff_Jeong_IJSS09} the same model\footnote{
It must be noted that the grandmaster Koiter \cite[p.~17-19, 23, 41]{Koiter64} came to reject the significant presence of couple stresses because
he based his investigations on the indeterminate couple stress theory with uniformly pointwise positive definite curvature energy, which tends
to maximize the influence of length scale effects in its rotational formulation. His arguments only show that this special
constrained gradient theory together with it's boundary conditions cannot be based on experimental evidence. However, the main
thrust of his comments remains valid and our symmetric formulation may compare favorable. We should also have in mind that Mindlin  ceased to use these models because he could finally not see the physical relevance at this time. Truesdell and Noll also wrote \cite[p.~400]{Truesdell65}: ``In favour of the Grioli-Toupin theory, in which the microrotation and macrorotation coincide, we can find no experimental evidence or theoretical advantage." } has been derived  based on a homogenization procedure and a novel invariance requirement introduced by  Neff et al. \cite{Neff_JeongMMS08}, called micro-randomness and it has
been shown that the model is well-posed.

\subsection{Our perspective}
Our contribution is intended to clarify and delineate under what boundary conditions we may expect or use symmetric nonlocal force stresses in the indeterminate couple stress model. When trying to relax the 4.th order problem (from gradient elasticity), it also seems expedient to retain the symmetry of the force stress tensor and of the moment stress tensor.  Respecting symmetry restricts the possibilities to choose among 2.nd order micromorphic  models. The importance of switching to a 2.nd order problem with new independent degrees of freedom is clear from the implementational  point of view with finite elements: a 2.nd order problem is much easier and more efficient. However, given the antisymmetric classical and our new symmetric formulations we may arrive at completely different 2.nd order formulations in case of mixed displacement-traction boundary conditions.

In general, the hyperstress-tensor (couple stresses, sometimes called double-stress \cite{maugin2010generalized}) in second  gradient elasticity \cite{Toupin62,Mindlin64,Koiter64,Sokolowski72,Dahler61,dahler1963theory} (see also the recent papers \cite{dell1995radius,dell1995validity,dell1996nucleation,yang2010higher,rosi2013propagation,auffray2013analytical,dell2014origins,ferretti2014modeling,MadeoZAMM2014,eremeyev2014equilibrium,rinaldi2014microscale,dell2015two}) may be defined as $\mathfrak{m}=(\mathfrak{m}_{ijk})=D_{D^2 u} W(D^2 u)$. Since $D^2 u=(u_{i,jk})$ is a third order tensor, so is $\mathfrak{m}$.  Moreover, since $u_{i,jk}$ is symmetric in $(jk)$ the same is usually assumed for $\mathfrak{m}_{ijk}$. This, however, is not mandatory, see \cite{morro2014interstitial}.

In the framework considered in this paper, the hyperstress-tensor is defined as $$\widetilde{m}:=D_{\nabla (\curl u)} W_{\rm curv}(\nabla (\curl u))\qquad \ \text{ or }\qquad \ \widehat{m}:=D_{\Curl (\sym\nabla u)} W_{\rm curv}(\Curl (\sym\nabla u)),$$ respectively, and both expressions are 2.nd order tensors\footnote{See the Appendix for the relation between the second order tensor $\widetilde{m}$ and the third order tensor $\mathfrak{m}$.} and are also called couple stress tensors, since they act as dual objects to  gradients of rotations. On the other hand, as we will see, we have two competing expressions of the  nonlocal force stress tensor:   a symmetric tensor $\widehat{\tau}$ versus an anti-symmetric tensor $\widetilde{\tau}$:
\begin{align*}
\widehat{\tau}&=\mu\,L_c^2\,\sym \Curl\{2\, \alpha_1 \dev\sym\Curl (\sym \nabla u)+2\,\alpha_2\,   \skw\Curl (\sym \nabla u)\}=\sym\Curl(\widehat{m})\in {\rm Sym}(3)\\
\widetilde{\tau}&=\frac{\mu}{2}\,L_c^2\,\anti
{\rm Div}\{2\,\alpha_1\, \dev\sym \nabla [\axl (\skw \nabla u)]+ 2\,\alpha_2\,\skw\nabla [\axl (\skw \nabla u)]\}=\frac{1}{2}\anti {
\rm Div}(\widetilde{m})\in\so(3)\notag,
\end{align*}
with ${\rm Div}(\widehat{\tau}-\widetilde{\tau})=0$. Since $[\Curl (\sym \nabla u)]^T=\nabla [\axl (\skw \nabla u)]$, it follows
\begin{align*}
\widehat{\tau}-\widetilde{\tau}=\mu\,L_c^2\,\Big\{&\quad\ 2\,\alpha_1 \,(\sym \Curl-\frac{1}{2}\,\anti
{\rm Div}).[\dev\sym\Curl (\sym \nabla u)]\\
&\,+2\,\alpha_2\, (\sym \Curl+\frac{1}{2}\anti
{\rm Div}).  [\skw\Curl (\sym \nabla u)]\quad\ \ \Big\}.
\end{align*}
 The independent constitutive variable $\widetilde{k}:=\nabla (\curl u)$ is the second gradient contribution considered by Grioli \cite{Grioli60}, Toupin \cite{Toupin62}, Mindlin \cite{Mindlin62}, Koiter \cite{Koiter64} and Sokolowski \cite{Sokolowski72}.  In general, neither $\widetilde{m}$ nor $\widehat{m}$ couple stress tensors are symmetric.

The symmetry of the force stress tensor in continuum mechanics is regulary discussed in the literature, see e.g. \cite{kuiken1995symmetry,mclennan1966symmetry,Neff_ZAMM05}. It has been suggested by McLennan \cite{mclennan1966symmetry} that a symmetric force stress tensor can always be constructed by  adding divergence-free couple stresses, since only its divergence occurs in the local conservation law.
 However, all of the previously given expositions use anti-symmetric nonlocal force-stresses. Since there is no conclusive evidence for the real need of a non-symmetric total force-stress tensor in the purely mechanical context, we apply Ockham's razor  and discard these non-symmetric force stress formulations. Our new alternative formulation will have symmetric couple stresses {\bf and} symmetric force stresses. Thus it satisfies the Cauchy-Boltzmann's axiom. We also show that the new formulation is well-posed in statics. While conceptually very pleasing, the real merits of such a ``completely symmetric'' formulation have yet to be discovered.

  Similarly to the classical indeterminate couple stress model which can be obtained as a constrained Cosserat model, our new $\Curl(\sym\nabla u)$-model can be obtained as a constrained ``microstrain" model \cite{Forest02b,Forest06,Neff_Forest07}.

The question of boundary conditions in higher gradient elasticity models has been a subject of constant attention. Bleustein has formulated the conclusive answer for general gradient elastic models involving the surface divergence operator \cite{bleustein1967note}. However, the traction boundary conditions obtained by  Tiersten and Bleustein in \cite{TierstenBleustein} with respect to the special case of the indeterminate couple stress model are  incomplete. In a forthcoming paper \cite{MadeoGhibaNeffMunchKM} we discuss and correct the form of the traction boundary conditions considered until now in the classical indeterminate couple stress model \cite{Mindlin62,Toupin62,Koiter64,Neff_Jeong_IJSS09,park2008variational,anthoine2000effect,Yang02,Park07}. Here, we just provide the correct answer obtained there in the form of a summarizing box.

\medskip

   The plan of the paper is now as follows: after a subsection fixing the notation, we outline  some related models in isotropic second gradient elasticity; we prove some auxiliary results and we discuss the invariance properties of the considered energy;  we recall the classical indeterminate couple stress model with skew-symmetric nonlocal force-stress (i.e. with non symmetric total force-stress tensor); we formulate the equilibrium problem for the new isotropic gradient elasticity model with symmetric nonlocal force stress (i.e. with symmetric total force-stress tensor) and we give an existence result; we discuss the difference of the classical indeterminate couple stress model with the introduced symmetric model; paying particular attention to the boundary virtual work principle  we show that these two possible formulations are applicable for different  types of traction boundary conditions; we discuss the possibility to switch from a 4.th-order problem to a 2.nd order  micromorphic model. All our existence results can be extended, mutatis mutandis, to first order anisotropic behaviour \cite{MadeoZAMM2014,NthGrad,madeo2012second}, i.e. considering as total energy $\langle\mathbb{C}.\,\sym\nabla u,\sym\nabla u\rangle+W_{\rm curv}(D^2 u)$ as long as $\mathbb{C}$ is a  uniformly positive definite tensor. We finish with some boxes summarizing our models and findings.

\subsection{Notational agreements}
In this paper, we denote by $\R^{3\times 3}$ the set of real $3\times 3$ second order tensors, written with
capital letters. For $a,b\in\R^3$ we let $\langle {a},{b}\rangle_{\R^3}$  denote the scalar product on $\R^3$ with
associated vector norm $\|{a}\|^2_{\R^3}=\langle {a},{a}\rangle_{\R^3}$.
The standard Euclidean scalar product on $\R^{3\times 3}$ is given by
$\langle{X},{Y}\rangle_{\R^{3\times3}}=\tr({X Y^T})$, and thus the Frobenius tensor norm is
$\|{X}\|^2=\langle{X},{X}\rangle_{\R^{3\times3}}$. In the following we omit the index
$\R^3,\R^{3\times3}$. The identity tensor on $\R^{3\times3}$ will be denoted by $\id$, so that
$\tr({X})=\langle{X},{\id}\rangle$. We adopt the usual abbreviations of Lie-algebra theory, i.e.,
 $\so(3):=\{X\in\mathbb{R}^{3\times3}\;|X^T=-X\}$ is the Lie-algebra of  skew symmetric tensors
and $\sL(3):=\{X\in\mathbb{R}^{3\times3}\;| \tr({X})=0\}$ is the Lie-algebra of traceless tensors.
 For all $X\in\mathbb{R}^{3\times3}$ we set $\sym X=\frac{1}{2}(X^T+X)\in\Sym$, $\skw X=\frac{1}{2}(X-X^T)\in \so(3)$ and the deviatoric part $\dev X=X-\frac{1}{3}\;\tr(X)\id\in \sL(3)$  and we have
the \emph{orthogonal Cartan-decomposition  of the Lie-algebra} $\gl(3)$
\begin{align}
\gl(3)&=\{\sL(3)\cap \Sym(3)\}\oplus\so(3) \oplus\mathbb{R}\!\cdot\! \id,\quad
X=\dev \sym X+ \skw X+\frac{1}{3}\tr(X) \id\,.
\end{align}
Throughout this paper (when we do not specify else) Latin subscripts take the values $1,2,3$.  Typical conventions for differential
operations are implied such as comma followed
by a subscript to denote the partial derivative with respect to
 the corresponding cartesian coordinate. We also use the Einstein notation of the sum over repeated indices if not differently specified. Here, for
\begin{align}
\overline{A}=\left(\begin{array}{ccc}
0 &-a_3&a_2\\
a_3&0& -a_1\\
-a_2& a_1&0
\end{array}\right)\in \so(3)
\end{align}
we consider the operators $\axl:\so(3)\rightarrow\mathbb{R}^3$ and $\anti:\mathbb{R}^3\rightarrow \so(3)$ through
\begin{align}
\axl(\overline{A}):=\left(
a_1,
a_2,
a_3
\right)^T,\quad \quad \overline{A}.\, v=(\axl \overline{A})\times v, \qquad \qquad (\anti(v))_{ij}=-\varepsilon_{ijk}v_k, \quad \quad \forall \, v\in\mathbb{R}^3,
\\\notag (\axl \overline{A})_k=-\frac{1}{2}\, \epsilon_{ijk}\overline{A}_{ij}=\frac{1}{2}\,\epsilon_{kij}\overline{A}_{ji}\,, \quad\qquad \overline{A}_{ij}=-\epsilon_{ijk}(\axl \overline{A})_k=:\anti(\axl \overline{A})_{ij},
\quad \quad
\end{align}
 where $\epsilon_{ijk}$ is the totally antisymmetric third order permutation tensor. We recall that  for a  third  order tensor $\mathbb{E}$ and $X\in \mathbb{R}^{3\times 3}$, $v\in \mathbb{R}^3$ we have the contraction operations $\mathbb{E}: X\in \mathbb{R}^{3}$, $\mathbb{E}. \, v\in \mathbb{R}^{3\times 3}$ and $X.\, v\in \mathbb{R}^3$, with the components
\begin{align}
(\mathbb{E}:\, X)_{i}=\mathbb{E}_{ijk}\,X_{kj}\, , \qquad  (\mathbb{E}. \,v)_{ij}=\mathbb{E}_{ijk}\,v_{k}\,,\qquad (X.\, v)_{i}=X_{ij}\,v_j.
\end{align}
For multiplication of two matrices  we will not use  other specific notations.

 We consider a body which  occupies  a bounded open set $\Omega$ of the three-dimensional Euclidian space $\R^3$ and assume that its boundary
$\partial \Omega$ is a piecewise smooth surface. An elastic material fills the domain $\Omega\subseteq \R^3$ and we refer the motion of the body to  rectangular axes $Ox_i$. By $\Co$  we denote the set of infinitely
differentiable functions with compact support in $\Omega$. In order to realize certain boundary conditions on an open subset $\Gamma\subseteq \partial \Omega$ we make use of the space \cite{BNPS2} of functions that vanish in a neighborhood of $\Gamma$, i.e.
\[
C^\infty_0(\overline{\Omega},\Gamma):=\big\{ u\,|\, \exists \,v\in C_0^\infty( \mathbb{R}^n\setminus \overline{\Gamma}) \ \ \text{such that} \ \ v\big|_{\Gamma}=u\big\}\,.
 \]
 \begin{figure}[h!]
\centering
\begin{minipage}[h]{0.5\linewidth}
\includegraphics[scale=0.5]{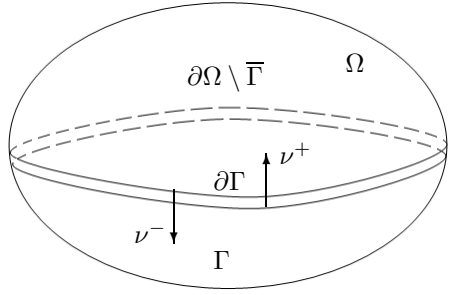}
\centering
 \put(-100,80){$\partial \Omega\setminus\overline{\Gamma}$}
    \put(-40,85){$\Omega$}
    \put(-90,10){$\Gamma$}
   \put(-90,40){$\partial \Gamma$}
  \put(-105,40){\vector(0,-1){20}}
  \put(-120,20){$\nu^-$}
  \put(-70,34){\vector(0,1){20}}
  \put(-65,50){$\nu^+$}
\caption{\footnotesize{The domain $\Omega\subseteq \mathbb{R}^3$ together with the part $\Gamma\subseteq \partial \Omega$, where (geometric) Dirichlet boundary conditions are prescribed. We need to represent the boundary conditions on a disjoint union of $\partial \Omega=(\partial \Omega\setminus\overline{\Gamma})\cup\Gamma\cup\partial \Gamma$, where $\Gamma$ is a open subset of $\partial \Omega$.}}
\label{capture-bc}
\end{minipage}
\end{figure}%
Here, $\nu^-$ is a  vector tangential to the surface $\partial \Omega\setminus \overline{\Gamma}$ and which is orthogonal to its boundary $\partial(\partial \Omega\setminus \overline{\Gamma})=\partial \Gamma$, $\tau^-=n\times \nu^-$ is the tangent to the curve $\partial \Gamma$ with respect to the orientation on $\partial \Omega\setminus \overline{\Gamma}$. Similarly, $\nu:=\nu^+$ is a  vector tangential to the surface $\Gamma$ and which is orthogonal to its boundary $\partial \Gamma$, $\tau:=\tau^+=n\times \nu^+$ is the tangent to the curve $\partial \Gamma$ with respect to the orientation on $\Gamma$. The jump across the joining curve $\partial \Gamma$ is defined by $[\,\cdot\,]^+-[\,\cdot\,]^-$, where
$$[\,\cdot\,]^-:=\hspace*{0cm}\dd\lim\limits_{\footnotesize{\begin{array}{c}x\in\partial \Omega\setminus \overline{\Gamma}\\
 \ x\rightarrow \partial \Gamma\end{array}}}\hspace*{0cm} [\,\cdot\,], \qquad \qquad [\,\cdot\,]^+:=\hspace*{-0.2cm}\dd\lim\limits_{\footnotesize{\begin{array}{c}x\in \Gamma\\
 \ x\rightarrow \partial \Gamma\end{array}}}\hspace*{-0.2cm} [\,\cdot\,].$$
 We assume that $\partial \Omega$ is a smooth surface. Hence, there are no singularities of the boundary and the jump $[\,\cdot\,]^+-[\,\cdot\,]^-$ arises only as consequence of possible discontinuities of the corresponding quantities which follows from the prescribed  boundary conditions on $\Gamma$ and $\partial \Omega\setminus \overline{\Gamma}$.

  The usual Lebesgue spaces of square integrable functions, vector or tensor fields on $\Omega$ with values in $\mathbb{R}$, $\mathbb{R}^3$ or $\mathbb{R}^{3\times 3}$, respectively will be denoted by $L^2(\Omega)$. Moreover, we introduce the standard Sobolev spaces \cite{Adams75,Raviart79,Leis86}
\begin{align}
\begin{array}{ll}
{\rm H}^1(\Omega)=\{u\in L^2(\Omega)\, |\, {\rm grad}\, u\in L^2(\Omega)\}, &\|u\|^2_{{\rm H}^1(\Omega)}:=\|u\|^2_{L^2(\Omega)}+\|{\rm grad}\, u\|^2_{L^2(\Omega)}\,,\vspace{1.5mm}\\
{\rm H}({\rm curl};\Omega)=\{v\in L^2(\Omega)\, |\, {\rm curl}\, v\in L^2(\Omega)\},   &\|v\|^2_{{\rm H}({\rm curl};\Omega)}:=\|v\|^2_{L^2(\Omega)}+\|{\rm curl}\, v\|^2_{L^2(\Omega)}\, ,\vspace{1.5mm}\\
{\rm H}({\rm div};\Omega)=\{v\in L^2(\Omega)\, |\, {\rm div}\, v\in L^2(\Omega)\}, &\|v\|^2_{{\rm H}({\rm div};\Omega)}:=\|v\|^2_{L^2(\Omega)}+\|{\rm div}\, v\|^2_{L^2(\Omega)}\, ,
\end{array}
\end{align}
of functions $u$ or vector fields $v$, respectively. Furthermore, we introduce their closed subspaces $H_0^1(\Omega)$,  ${\rm H}_0({\rm curl};\Omega)$ as completion under the respective graph norms of the scalar valued space $C_0^\infty(\Omega)$.
 We also consider the spaces
\[
 H^1_0(\Omega;\Gamma),\qquad  H^1_0(\dv;\Omega;\Gamma),\qquad  H^1_0(\curl;\Omega;\Gamma)
\]
 as completion under the respective graph norms of the scalar-valued space of the scalar-values space $C^\infty(\overline{\Omega},\Gamma)$. Therefore, these spaces generalize the homogeneous Dirichlet boundary conditions:
 \[u\big|_\Gamma=0,\qquad \text{and}\qquad \langle u, n\rangle|_\Gamma=0\qquad \text{and}\qquad u\times n|_\Gamma=0,
 \]
  respectively.  For vector fields $v$ with components in ${\rm H}^{1}(\Omega)$, i.e.
$
v=\left(    v_1, v_2, v_3\right)^T\, , v_i\in {\rm H}^{1}(\Omega),
$
we define
$
 \nabla \,v=\left(
   (\nabla\,  v_1)^T,
    (\nabla\, v_2)^T,
    (\nabla\, v_3)^T
\right)^T
$, while for tensor fields $P$ with rows in ${\rm H}({\rm curl}\,; \Omega)$, respectively ${\rm H}({\rm div}\,; \Omega)$, i.e.
$
P=\left(
    P_1^T,
    P_2^T,
    P_3^T
\right)$, $P_i\in {\rm H}({\rm curl}\,; \Omega)$ respectively $P_i\in {\rm H}({\rm div}\,; \Omega)$
we define
$
 {\rm Curl}\,P=\left(
   ({\rm curl}\, P_1)^T,
    ({\rm curl}\,P_2)^T,
    ({\rm curl}\,P_3)^T
\right)^T,
$ $ {\rm Div}\,P=\left(
   {\rm div}\, P_1,
    {\rm div}\,P_2,
    {\rm div}\,P_3
\right)^T$.
The corresponding Sobolev-spaces will be denoted by
\[
  H^1(\Omega),\qquad H^1(\Div;\Omega),\qquad H^1(\Curl;\Omega),\qquad H^1_0( \Omega;\Gamma),\qquad H^1_0(\Div;\Omega;\Gamma),\qquad H^1_0(\Curl;\Omega;\Gamma).
 \]

\section{Preliminaries}

\subsection{Related models in isotropic second gradient elasticity}\setcounter{equation}{0}

 One aim of this paper is to propose a new  representation of the curvature energy $W_{\rm curv}(D^2 u)$ and to prove that the corresponding minimization problem
\begin{align}
I(u)=\int_\Omega \left[\mu\, \|{\rm sym} \nabla u\|^2+\frac{\lambda}{2}\, [\tr({\rm sym} \nabla u)]^2+W_{\rm curv}(D^2 u)\right] dV \qquad \mapsto \qquad
\text{min. \ w.r.t.} \quad u,
\end{align}
admit unique minimizers under some appropriate boundary condition. Here $\lambda,\mu$ are the usual Lam\'{e}  constitutive coefficients of isotropic linear elasticity, which is  fundamental to small deformation gradient elasticity. If the curvature energy has the form $W_{\rm curv}(D^2 u)=W_{\rm curv}(D\sym\nabla u)$, the model is called {\bf a strain gradient model}.
We define the third order hyperstress as $D_{D^2 u}W_{\rm curv}(D^2 u)$.

In the following we outline  some curvature energies already proposed in different isotropic second gradient elasticity models:
\begin{itemize}
\item Mindlin \cite{Mindlin64,Mindlin65,Mindlin62}   considered energies (gradient elastic) based on the tensors
\begin{align*}
\eta_{ijk}&=u_{k,ij}, \qquad\quad\qquad\qquad\qquad\qquad\quad \  \ (\text{i.e.}\ \ \eta=\nabla(\nabla u)),\\
\widetilde{\eta}_{ijk}&=\frac{1}{2}(u_{k,ji}+u_{j,ki})=
\varepsilon_{kj,i},\qquad\qquad\quad (\text{i.e.}\ \ \widetilde{\eta}=\nabla(\sym \nabla u))),\\
\widetilde{k}_{ij}&=\frac{1}{2}\epsilon_{jlk}u_{k,li} \qquad\quad\qquad\qquad\qquad\qquad \ (\text{i.e.}\ \
\widetilde{k}=\frac{1}{2}\nabla (\curl u) ),\\ \eta^S_{ijk}&=\frac{1}{3}(u_{k,ij}+u_{i,jk}+u_{j,ki}).
\end{align*}
The most general isotropic curvature energy defined in terms of $D^2u$ has 5 material constants, while the anisotropic representation is much more involved and still subject of ongoing research \cite{Auffray,IsolaSciarraVidoliPRSA}.

Mindlin and Eshel \cite{Mindlin68} have also proposed the following three alternative forms :
\begin{align}\label{alternME}
W_{\rm curv}(D^2 u)&=\mu\,L_c^2\,[a_1^{(1)} \,\eta_{kii}\eta_{kjj}+a_2^{(1)} \, \eta_{ijk}\eta_{ijk}+a_3^{(1)} \,\eta_{ijk}\eta_{jki}+a_4^{(1)} \,\eta_{jji}\eta_{kki}+a_5^{(1)} \,\eta_{iik}\eta_{kjj}] \tag{I}\\
&=\mu\,L_c^2\,[a_1^{(2)} \widetilde{\eta}_{iik}\widetilde{\eta}_{kjj}+
a_2^{(2)} \widetilde{\eta}_{ijj}\widetilde{\eta}_{ikk}+
a_3^{(2)} \widetilde{\eta}_{iik}\widetilde{\eta}_{jjk}+
a_4^{(2)} \widetilde{\eta}_{ijk}\widetilde{\eta}_{ijk}+
a_5^{(2)} \widetilde{\eta}_{ijk}\widetilde{\eta}_{kji} ]\tag{II}\\
&=\mu\,L_c^2\,[\widetilde{a}_1^{(3)} \widetilde{k}_{ij}\widetilde{k}_{ij}+
\widetilde{a}_2^{(3)} \widetilde{k}_{ij}\widetilde{k}_{ji}+
\widetilde{a}_3^{(3)} \eta^S_{iij}\eta^S_{kkj}+
\widetilde{a}_4^{(3)} \eta^S_{ijk}\eta^S_{ijk}+
\widetilde{a}_5^{(3)} \epsilon_{ijk}\widetilde{k}_{ij}\eta^S_{kll} ]\,,\tag{III}
\end{align}
which are frequently cited in the literature, where $L_c$ is the smallest characteristic length in the body and $a_i^{(j)},\widetilde{a}_i^{(j)}$ are dimensionless weighting parameters.
\item a simple curvature energy is considered by Lam \cite{Lam03,nikolov2007origin}
\begin{align}
W_{\rm curv}(D^2 u)&=\mu\,L_c^2[\,a_0\, \|\nabla {\rm div}\, u\|^2+a_1\, \widehat{\eta}_{ijk}\widehat{\eta}_{ijk}+a_2\,\|\sym \nabla (\curl u)\|^2]\\
&=\mu\,L_c^2[\,a_0\, \|\nabla \tr(\sym \nabla u)\|^2+a_1\, \widehat{\eta}_{ijk}\widehat{\eta}_{ijk}+4\,a_2\,\|\sym \Curl (\sym \nabla  u)\|^2]\,,\notag
\end{align}
where $\widehat{\eta}_{ijk}$ is called the deviatoric stretch gradient and which is defined  (see e.g. \cite{Neff_Jeong_IJSS09}) by $$\widehat{\eta}_{ijk}={\eta}_{ijk}^S-{\eta}_{ijk}^{(0)},  \qquad \ {\eta}_{ijk}^{(0)}=\frac{1}{5}(\delta_{ij}\, {\eta}_{mmk}^S+\delta_{jk}\, {\eta}_{mmi}^S+\delta_{ki}\, {\eta}_{mmj}^S).$$
\item another simplified strain gradient elasticity model is proposed in \cite{Aifantis97,Lazar05,lazar2006note} based on the curvature energy
\begin{align}
W_{\rm curv}(D^2 u)&=\mu\,L_c^2\,[a_0\, \|\nabla {\rm tr}(\sym \nabla \, u)\|^2+a_1\,\|\nabla\, ( \sym   \nabla   u)\|^2]\\&=\mu\,L_c^2\,[a_0\, \|\nabla \,{\rm div}\, u\|^2+a_1\,\|\nabla\, (\sym   \nabla   u)\|^2],\notag
\end{align}
which already leads to symmetric nonlocal force-stresses, see Section \ref{str-grad-sapp}.
\item in the same line, using also the second order curvature tensor $\tilde{k}=\frac{1}{2}\nabla \, \curl u$,  in \cite{Sharma05,Kleinert89}  the following energy is considered
\begin{align}
W_{\rm curv}(D^2 u)&=\mu\,L_c^2\,[a_0\, \| \nabla   {\rm div} \,u\|^2+{a}_1\,\| \nabla \curl\, u\|^2]=\mu\,L_c^2\,[a_0\, \|\nabla {\rm tr}(\sym \nabla \, u)\|^2\,+{a}_1\,\| \nabla [\axl (\skw \nabla u)]\|^2]\notag\\
&=\mu\,L_c^2\,[a_0\, \|\nabla {\rm tr}(\sym \nabla \, u)\|^2\,+{a}_1\,\| \nabla\, (\skew \nabla   u)\|^2].
\end{align}
Let us remark that $\tr(\widetilde{k})=\tr(\nabla\axl \skew \nabla u)={\rm div} (\curl\, u)=0$.
\item {\bf the indeterminate couple stress model} (Grioli-Koiter-Mindlin-Toupin model) \cite{Grioli60,Aero61,Koiter64,Mindlin62,Toupin64,Sokolowski72,grioli2003microstructures} in which the higher derivatives (apparently)  appear only through derivatives of the infinitesimal continuum rotation $\curl u$.  Hence, the curvature energy  has  the equivalent forms
\begin{align}\label{KMTe}
W_{\rm curv}(D^2 u)&=\mu\,L_c^2\,\left[\frac{\alpha_1}{4}\, \|\sym \nabla (\curl\, u)\|^2+\frac{\alpha_2}{4}\,\| \skw \nabla (\curl\, u)\|^2\right]\notag\\
&=\mu\,L_c^2\,[\alpha_1\, \|\sym\nabla[\axl(\skw  \nabla u)]\|^2+\alpha_2\,\| \skw \nabla[\axl(\skw  \nabla u)]\|^2]\notag\\
&=\mu\,L_c^2\,\left[\frac{\alpha_1}{4}\, \|\dev \sym \nabla (\curl\, u)\|^2+\frac{\alpha_2}{4}\,\| \skw \nabla (\curl\, u)\|^2\right]\\
&=\mu\,L_c^2\,[\alpha_1\, \|\sym \Curl(\sym \nabla  u)\|^2+ {\alpha_2}\,\| \skw \Curl(\sym \nabla  u)\|^2]\notag\\
&=\mu\,L_c^2\,[\alpha_1\, \|\dev \sym \Curl(\sym \nabla  u)\|^2+{\alpha_2}\,\| \skw \Curl(\sym \nabla  u)\|^2].\notag
\end{align}
Note carefully  that $ \tr[\sym\Curl (\sym \nabla u)]=\tr[\sym\nabla[\axl(\skw  \nabla u)]]=0$. Therefore, we are entitled to use the deviatoric-representation, which is useful when regarding the model in the larger context of micromorphic models.  Here, we have used the master identity to be established in Corollary  \ref{corollaryaxlcurl}
\begin{align*}
\underbrace{\nabla[\axl(\skw  \nabla u)]}_{\text{rotation gradient}}&=\underbrace{[\Curl(\sym\,  \nabla u)]^T}_{\text{strain gradient}},
\end{align*}
which allows us easily to switch from considerations on the level of strain gradients to the level of rotational gradients and vice versa.

We  also used the identities
\begin{align*}
 \ \quad 2\,  \axl (\skw \nabla u)&=\curl u, \qquad
\sym \nabla ( \curl u)=2\,\sym\Curl (\sym \nabla u),\\ \skw \nabla ( \curl u)&=-2\,\skw\Curl (\sym \nabla u), \qquad
 \tr[\Curl (\sym \nabla u)]=0.
\end{align*}
Although  this energy admits  the equivalent forms \eqref{KMTe}$_1$ and \eqref{KMTe}$_6$, the equations and the boundary value problem of the indeterminate couple stress model is usually formulated only using the form \eqref{KMTe}$_1$ of the energy. Hence, we may reformulate the main aim of the present paper: to formulate the boundary value problem for {\bf the indeterminate couple stress model using the alternative form  \eqref{KMTe}$_6$ of the energy} of the Grioli-Koiter-Mindlin-Toupin model.
We also remark that the spherical part of the couple stress tensor remains {\bf indeterminate} since $\tr(\nabla (\curl u))={\rm div} (\curl u)=0$. In order to prove the pointwise uniform positive definiteness it is assumed, following \cite{Koiter64}, that $\alpha_1>0, \alpha_2>0$. Note that pointwise uniform positivity is often  assumed when deriving analytical solutions for simple boundary value problems because it allows to invert the couple stress-curvature relation. We will see subsequently, that pointwise positive definiteness is not necessary for well-posedness.
\newpage
\item In this setting, {\bf Grioli} \cite{Grioli60,grioli2003microstructures} (see also Fleck \cite{Fleck93,Fleck97,Fleck01}) initially considered only the choice $\alpha_1=\alpha_2$. In fact, the energy originally proposed by Grioli \cite{Grioli60} is
     \begin{align}\label{KMTe0}
W_{\rm curv}(D^2 u)&=\mu\,L_c^2\,\left[\frac{\alpha_1}{4}\, \, \|\nabla (\curl\, u)\|^2+\frac{\eta^\prime}{4} \,\tr[ (\nabla (\curl\, u))^2]\right]\notag\\
&=\mu\,L_c^2\,[\alpha_1\, \|\dev \sym \nabla [\axl (\skw \nabla u)]\|^2+\alpha_1\,\| \skw \nabla [\axl (\skw \nabla u)]\|^2\notag\\&\qquad \qquad +\eta^\prime \,\langle \nabla [\axl (\skw \nabla u)],(\nabla [\axl (\skw \nabla u)])^T\rangle]
\\\notag
&=\mu\,L_c^2\,\left[\frac{\alpha_1}{4}\, \, \|\dev \sym \nabla (\curl\, u)\|^2+\frac{\alpha_1}{4}\, \,\| \skw \nabla (\curl\, u)\|^2+\frac{\eta^\prime}{4} \,\langle \nabla (\curl\, u),(\nabla (\curl\, u))^T\rangle\right]
\\
&=\mu\,L_c^2\,\left[\frac{\alpha_1+\eta^\prime}{4}\, \|\dev \sym \nabla (\curl\, u)\|^2+\frac{\alpha_1-\eta^\prime}{4}\,\| \skw \nabla (\curl\, u)\|^2\right].\notag
\end{align}
Mindlin \cite[p. 425]{Mindlin62} explained the relations between  Toupin's constitutive equations \cite{Toupin62} and  Grioli's \cite{Grioli60} constitutive equations and concluded that the obtained equations in the linearized theory are identical, since the extra constitutive parameter $\eta^\prime$ of Grioli's model does not explicitly appear in the equations of motion  but enters only the boundary conditions. The same extra constitutive coefficient appears in  Mindlin and Eshel's (III) and Grioli's version \eqref{KMTe0}.

\item
  {\bf the modified - symmetric couple stress model - the conformal model}.  On the other hand, in the conformal case  \cite{Neff_Jeong_IJSS09,Neff_Paris_Maugin09} one may consider that $\alpha_2=0$, which makes the second order couple stress tensor $\widetilde{m}$ symmetric and trace free \cite{dahler1963theory}.  This conformal curvature case has been
considered  by Neff in  \cite{Neff_Jeong_IJSS09}, the curvature energy having the form
\begin{align}
W_{\rm curv}(D^2 u)&=\mu\,L_c^2\,\frac{\alpha_1}{4}\, \|\sym \nabla (\curl\, u)\|^2=\mu\,L_c^2\,\alpha_1\, \|\dev \sym \Curl(\sym \nabla  u)\|^2.
\end{align}
Indeed, there are two major reasons uncovered  in \cite{Neff_Jeong_IJSS09} for using the modified couple stress model. First, in order to avoid singular stiffening behaviour for smaller and smaller samples in bending \cite{Neff_Jeong_bounded_stiffness09} one has to take $\alpha_2=0$. Second, based on a homogenization procedure invoking an intuitively appealing  natural ``micro-randomness" assumption (a strong statement of microstructural isotropy) requires conformal invariance, which is again equivalent to $\alpha_2=0$. Such a model is still well-posed \cite{Neff_JeongMMS08}, leading to existence and uniqueness results with only one additional material length scale parameter, while it is {\bf not} pointwise uniformly positive definite.

\item {\bf the skew-symmetric couple stress model}.
  { Hadjesfandiari and Dargush} strongly advocate \cite{hadjesfandiari2011couple,hadjesfandiari2013fundamental,hadjesfandiari2013skew} the opposite extreme case, $\alpha_1=0$ and $\alpha_2>0$, i.e. they  propose  the  curvature  energy
\begin{align}
W_{\rm curv}(D^2 u)&=\mu\,L_c^2\,\frac{\alpha_2}{4}\, \|\skw \nabla({\rm curl}\, u)\|^2=\mu\,L_c^2\,\frac{\alpha_2}{2}\,  \|\axl \skw \nabla({\rm curl}\, u)\|^2\\
&=\mu\,L_c^2\,\frac{\alpha_2}{8}\, \|{\rm curl}\,({\rm curl}\, u)\|^2\notag=\mu\,L_c^2\, {\alpha_2}\,\| \skw \Curl(\sym \nabla  u)\|^2\notag.
\end{align}
In that model the nonlocal force stresses and the couple stresses are both assumed to be skew-symmetric. Their reasoning, based in fact on an incomplete understanding of boundary conditions (see \cite{MadeoGhibaNeffMunchKM}) is critically discussed   and generally refuted in \cite{NeffGhibagegenHD}, while mathematically it is also well-posed.
 \end{itemize}

\subsection{Auxiliary results}\label{auxiliarsect}

Further on, we consider a simply connected domain $\Omega\subseteq \mathbb{R}^{3\times 3}$. The starting point is given by the well-known  Nye's formula  \cite{Nye53,Neff_curl06}
\begin{align}
-\Curl\, \overline{A}&=(\nabla \axl \overline{A})^T-\tr[(\nabla \axl \,\overline{A})^T]\, \id ,\notag\\
\overline{\kappa}:=\nabla(\axl \overline{A} )&= -(\Curl\, \overline{A})^T+\frac{1}{2}\tr[(\Curl \,\overline{A})^T]\, \id= \alpha^T-\frac{1}{2}\tr[\alpha]\, \id\,\quad\quad  \text{ ``Nye's curvature tensor"},
\end{align}
for all skew-symmetric matrices $\overline{A}\in \so(3)$, where $\alpha:=-\Curl \overline{A}$ is the micro-dislocation density tensor.

\begin{proposition}\label{lemaaxlcurl} Let $p:\Omega\subseteq \mathbb{R}^3 \rightarrow\mathbb{R}^{3\times 3}$ be given. The formula
\begin{align}
\nabla[\axl \skw p]=[\Curl(\sym\, p)]^T
\end{align}
holds true if and only if there is $u\in C^2(\Omega)$ such that $p=\nabla u$.
\end{proposition}
\begin{proof}
Let us first prove that
\begin{align}
\nabla[\axl (\skw \nabla u)]=[\Curl(\sym\, \nabla u)]^T, \qquad \text{for all} \qquad u\in C^2(\Omega).
\end{align}
On the one hand, using  Nye's formula for $A=\skw \nabla u$, we obtain
\begin{align}
-\Curl\, (\skw \nabla u)&=(\nabla[ \axl (\skw \nabla u)])^T-\tr[(\nabla[ \axl (\skw \nabla u)])^T]\,\id ,
\end{align}
which implies
\begin{align}\label{v25}
-\Curl(\skw \nabla u)&=(\nabla[ \axl (\skw \nabla u)])^T-\frac{1}{2}\tr[(\nabla (\curl u))^T]\,  \id\\
&=(\nabla[ \axl (\skw \nabla u)])^T-\frac{1}{2}{\rm div}(\curl  u)\, \id=(\nabla [\axl (\skw \nabla u)])^T.\notag
\end{align}
On the other hand,
$
\Curl(\nabla u)=0, \ \nabla u=\sym\nabla u+\skw \nabla u.
$
Thus, we deduce
\begin{align}
\Curl(\sym \nabla u)&=\Curl(\nabla u-\skw \nabla u)=\Curl(\nabla u)-\Curl(\skw \nabla u)\\
&\!\!\!\overset{\eqref{v25}}{=}\Curl(\nabla u)+(\nabla [\axl (\skw \nabla u)])^T-\frac{1}{2}{\rm div}(\curl  u)\,  \id=(\nabla [\axl (\skw \nabla u)])^T.\notag
\end{align}
This establishes the first part of the claim.

Now, we prove that $\nabla[\axl \skw p]=[\Curl(\sym\, p)]^T$ implies that there is a function $u\in C^{2}(\Omega)$ such that $p=\nabla u$. Using again  Nye's formula, we obtain
\begin{align}
(\nabla \axl \skw p)^T\overset{\text{(Nye)}}{=}-(\Curl\, (\skw p))+\tr[\nabla [\axl \skw p]]\,  \id.
\end{align}
Hence, our new hypothesis is
$
\Curl(\sym p){=}(\nabla [\axl \skw p])^T,
$
which implies
$
\Curl(\sym p){=}-(\Curl\, \skw p)+\tr[\nabla[ \axl \skw p]]\,  \id.
$
Hence, we obtain
\begin{align}
\Curl(\sym p+ \skw p)&=\tr[\nabla [\axl \skw p]]\,  \id\qquad
 \Leftrightarrow\qquad
\Curl(p)=\tr[\nabla [\axl \skw p]]\, \id,
\end{align}
or, in the  equivalent form
\begin{align}\label{axltrcurl}
\Curl(p)={\rm div}(\axl \skw p)\, \id.
\end{align}
We have obtained the formula
\begin{align}\label{ptaxl}
\tr[\Curl p]=3\,{\rm div}(\axl \skw p).
\end{align}

Let us also remark that  considering a matrix $B\in\mathbb{R}^{3\times 3}$, we have
\begin{align}\label{axlcurldiv}
\tr[\Curl( \skw B)\big]=2(b_{1,1}+b_{2,2}+b_{3,3})=2\,{\rm div}\, b\,,\qquad b=\axl(\skw B).
\end{align} Therefore, from \eqref{axlcurldiv} we also have obtained
\begin{align}\label{axlcurldiv11}
\tr[\Curl( \skw p)\big]=2\,{\rm div}\, [\axl (\skw p)]\,.
\end{align}
Moreover, for a matrix $B\in \mathbb{R}^{3\times 3}$, we have that
\begin{align}
\tr (\Curl B)&=(B_{13,2}-B_{31,2})+(B_{21,3}-B_{12,3})+(B_{32,1}-B_{23,1}).\notag
\end{align}
We deduce
$
\tr(\Curl S)=0 \  \text{for all} \  S\in {\rm Sym}(3).
$
Hence,
\begin{align}\label{axlcurldiv12}
\tr[\Curl( \sym p)\big]=0\, \quad \forall\, p\in \mathbb{R}^{3\times 3}.
\end{align}
The relations \eqref{axlcurldiv11} and \eqref{axlcurldiv12} lead to
$
\tr[\Curl p\big]=2\,{\rm div}\, [\axl (\skw p)]\,,
$
and together with
\eqref{ptaxl} to
\begin{align}
{\rm div}\, [\axl (\skw p)]=0 .
\end{align}
Using \eqref{axltrcurl}, we obtain
$
\Curl p=0.
$
Since $\Omega$ is an open domain in $\R^3$, it follows that there is an vector $u$, such that $p=\nabla u$ and the proof is complete.
\end{proof}
\begin{corollary}\label{corollaryaxlcurl} For $u\in C^2(\Omega)$ the following formula holds true
\begin{align}
\nabla[\axl (\skw \nabla u)]=[\Curl(\sym\, \nabla u)]^T.
\end{align}
\end{corollary}
\begin{corollary}\label{corollaryaxlcurl2} For $u\in C^2(\Omega)$ the following formula holds true
\begin{align}
[\nabla {\rm curl}\, u]^T=\Curl[(\nabla u)^T].
\end{align}
Therefore, $(\nabla u)^T\in H(\Curl;\Omega)$ is equivalent to ${\rm curl}\, u\in H^1(\Omega)$.
\end{corollary}
As consequence of the above remark, it follows that if ${\rm curl}\, u\in H^1(\Omega)$, then $(\nabla u)^T. \tau\in L^2(\partial \Omega)$  for any tengential direction $\tau$ at the boundary and, since
$\langle (\nabla u)^T. \tau,n\rangle=\langle \tau,(\nabla u).n\rangle$, it results that  $(\id-n\otimes n)\,(\nabla u).n\in L^2(\partial \Omega)$, in the sense of trace.

Let us also recall the Saint-Venant compatibility condition
\begin{proposition}\label{compatibilitySV}{\rm(see e.g. \cite{Ciarlet88})} Let a symmetric tensor field $\widehat{\varepsilon}:\Omega\subseteq \mathbb{R}^3 \rightarrow{\rm Sym}(3)$ be given. Then,
\begin{align}
 {\rm \bf inc} (\varepsilon):=\Curl[(\Curl \widehat{\varepsilon})^T]=0 \qquad \Leftrightarrow \qquad \text{there is} \quad u\in C^2(\Omega) \quad \text{such that}\quad \widehat{\varepsilon}=\sym(\nabla u).
\end{align}
\end{proposition}
We note that
\begin{align}
 {\rm \bf inc} (p):=\Curl[(\Curl \sym p)^T]&\in{\rm Sym}(3)\\
 \text{while}\ \ \  \Curl [(\Curl \skew\, p)^T]&\in \so(3) \quad\quad \forall \  p\in\mathbb{R}^{3\times3}.\notag
\end{align}
We also remark that a direct consequence of  Proposition \ref{lemaaxlcurl} is the following first order compatibility condition
\begin{proposition}\label{compatibility} Let $p:\Omega\subseteq \mathbb{R}^3 \rightarrow\mathbb{R}^{3\times 3}$ be given. Then,
\begin{align}
{\rm \bf INC}(p):=[\Curl(\sym\, p)]^T-\nabla[\axl \skw p] \qquad \Leftrightarrow \qquad \text{there is} \quad u\in C^2(\Omega) \quad \text{such that}\quad p=\nabla u.
\end{align}
\end{proposition}

We observe  that
$$
{\rm \bf INC}(p)\in\mathbb{R}^{3\times3}\quad\quad \forall \  p\in\mathbb{R}^{3\times3}.
$$
We  recall  the well known first order compatibility condition
\begin{proposition}\label{wcompatibility} Let $p:\Omega\subseteq \mathbb{R}^3 \rightarrow\mathbb{R}^{3\times 3}$ be given. Then,
\begin{align}
\Curl p=0 \quad  ( [\Curl (\sym p)]^T=-[\Curl (\skw p)]^T) \quad \Leftrightarrow \quad \text{there is} \quad u\in C^2(\Omega) \quad \text{such that}\quad p=\nabla u.
\end{align}
\end{proposition}
Hence, we have the following equivalence
\begin{corollary}\label{equicompatibility} Let $p:\Omega\subseteq \mathbb{R}^3 \rightarrow\mathbb{R}^{3\times 3}$ be given. Then,
\begin{align}
\Curl p=0 \qquad  \Leftrightarrow \qquad {\rm \bf INC}(p)=0.
\end{align}
\end{corollary}
We finally  remark that
\begin{align}
\Curl({\rm \bf INC}(p))=\Curl[(\Curl \sym p)^T]={\rm \bf inc}(\sym p).
\end{align}

\subsection{Discussion of invariance properties}

The difference between the $\nabla [\axl(\skew\nabla u)]$ formulation and the $\Curl(\sym\nabla u)$ formulation can be seen when considering the results under superposed incompatible tensor fields:
\begin{remark}
The quantity $[\Curl(\sym\nabla u)]^T$ is invariant under locally adding a skew-symmetric non-constant tensor field  $W(x)\in\so(3)$, i.e.
\begin{align}\label{coninvcurl}
[\Curl(\sym(\nabla u+W(x)))]^T=[\Curl(\sym \nabla u)]^T.
\end{align}
 However, since  $\nabla[\axl \skw p]\neq[\Curl\,(\sym p)]^T$ for general incompatible $p\in\mathbb{R}^{3\times 3}$, $p\neq \nabla u$, the quantity
  $\nabla[\axl (\skw \nabla u)]$ is not invariant under locally adding $W(x)\in\so(3)$, i.e.
  \begin{align}
  \nabla[\axl(\skw (\nabla u+W(x))]\neq \nabla[\axl (\skw \nabla u)].
  \end{align}
 \end{remark}

\begin{remark}
The term $\nabla[\axl (\skw \nabla u)]$ is invariant under locally adding a symmetric, non-constant tensor field $S(x)\in{\rm Sym}(3)$, i.e.
\begin{align}\label{invsymaxl}
\nabla[\axl\,(\skw (\nabla u+S(x))]=\nabla[\axl (\skw \nabla u)].
\end{align}
 \end{remark}
Let us recall  the Lie-group decomposition ${\rm GL}^+(3)$ and the corresponding Lie-algebra decomposition:
\begin{align}
{\rm GL}^+(3)&=\{{\rm SL}(3)/{\rm SO}(3)\}\cdot {\rm SO}(3)\cdot (\mathbb{R}^+\!\cdot\! \id)  \qquad \quad\qquad \, \text{Lie-group decomposition},\notag\\
T_{\id}{\rm GL}^+(3)=\mathbb{R}^{3\times3}={\mathfrak{gl}}(3)&=\{\sL(3)\cap \Sym(3)\}\oplus\so(3) \oplus\mathbb{R}\!\cdot\! \id \qquad\qquad\quad\   \text{Lie-algebra decomposition}.
\end{align}
The space ${\rm Sym}(3)$ is not a Lie-algebra, it is only a vector space and it does  not have a group structure: the set ${\rm GL}(3)/{\rm SO}(3)={\rm PSym}(3)$ is not a group, neither is the set ${\mathfrak{gl}}(3)\cap \Sym(3)$ a Lie-algebra. Hence, the invariance requirement in \eqref{coninvcurl}, i.e.,  locally adding $W(x)\in\so(3)$ is much more plausible than assuming \eqref{invsymaxl} since it yields $\so(3)$-Lie invariance.

\subsection{Conformal invariance of the curvature energy and group theoretic arguments in favour of the modified couple stress theory}

An infinitesimal conformal mapping \cite{Neff_Jeong_Conformal_ZAMM08,Neff_Jeong_IJSS09} preserves (to first order) angles and shapes of infinitesimal figures. The included inhomogeneity is therefore only a global feature of the mapping (see Figure \ref{capture}). There is locally no shear-type deformation. Therefore it seems natural to require that the second gradient model should not ascribe energy to such deformation modes.

A map $\phi_c:\mathbb{R}^3\rightarrow\mathbb{R}^3$ is infinitesimal conformal if and only if its Jacobian satisfies pointwise $\nabla \phi_c(x)\in \mathbb{R}\cdot \id+\so(3)$, where $\mathbb{R}\cdot \id+\so(3)$ is the conformal Lie-algebra. This implies \cite{Neff_Jeong_Conformal_ZAMM08,Neff_Jeong_IJSS09,Neff_Jeong_bounded_stiffness09} the representation (see Figure \ref{capture})
\begin{align}
\phi_c(x)=\frac{1}{2}\left(2\langle \axl\overline{W},x\rangle \,x-\axl\overline{W}\|x\|^2\right)+[\widehat{p}\, \id+\widehat{A}]\cdot x+\widehat{b}\,,
\end{align}
where $\overline{W},\widehat{A}\in \so(3)$, $\widehat{b}\in \mathbb{R}^3$, $\widehat{p}\in \mathbb{R}$ are arbitrary given constants.
For the infinitesimal conformal mapping $\phi_c$ we note
\begin{align}\label{conf}
\begin{array}{ll}
\nabla \phi_c(x)=[\langle \axl\overline{W},x\rangle+\widehat{p}]\,\id +\anti(\overline{W}.\, x)+\widehat{A},&\qquad
{\rm div}\,\phi_c(x)=\tr[\nabla \phi_c(x)]=3\, [\langle \axl\overline{W},x\rangle+\widehat{p}],\vspace{1.2mm}\\
\skew  \nabla \phi_c(x)=\anti(\overline{W}.\, x)+\widehat{A},&\qquad
\sym \nabla \phi_c(x)=[\langle \axl\overline{W},x\rangle+\widehat{p}]\, \id,\vspace{1.2mm}\\
\dev\sym \nabla \phi_c(x)=0,&\qquad
\nabla \curl \phi_c(x)=2\, \overline{W}\in \so(3),\vspace{1.2mm}\\
\sym \nabla\curl \phi_c(x)=0,&\qquad
\skew\,\nabla \curl \phi_c(x)=2\, \overline{W}.
\end{array}
\end{align}
These relations  are easily established. By {\bf conformal invariance} of the curvature energy term we  mean that the curvature energy vanishes on infinitesimal conformal mappings. This is equivalent to
\begin{align}
W_{\rm curv}(D^2 \phi_c)=0\qquad \text{for all conformal maps} \quad \phi_c,
\end{align}
or in terms of the second order couple stress tensor $\widetilde{m}:=D_{\nabla \curl u}W_{\rm curv}(\nabla \curl u)$,
\begin{align}
\widetilde{m}(D^2 \phi_c)=0 \qquad \text{for all conformal maps} \quad \phi_c.
\end{align}
The classical linear elastic energy still ascribes energy to such a deformation mode, but only related to the bulk modulus, i.e.,
\begin{align}
W_{\rm lin}(\nabla \phi_c)=\underbrace{\mu\, \|\dev\sym \nabla \phi_c\|^2}_{=0}+\frac{2\,\mu+3\,\lambda}{2} \, [\tr(\nabla  \phi_c)]^2=\frac{2\,\mu+3\,\lambda}{2} \, [\tr(\nabla  \phi_c)]^2.
\end{align}
In case of a classical infinitesimal perfect plasticity formulation with von Mises deviatoric flow rule, conformal mappings are precisely those inhomogeneous mappings, that do not lead to plastic flow \cite{Neff_Cosserat_plasticity05}, since the deviatoric stresses remain zero: $\dev \sym \nabla \phi_c=0$.

In that perspective
\[\framebox[1.1\width]{\textit{conformal mappings are ideally elastic transformations and should not lead to moment stresses.}}
\]

\begin{figure}[h!]
\centering
\begin{minipage}[h]{1\linewidth}
\includegraphics[scale=0.5]{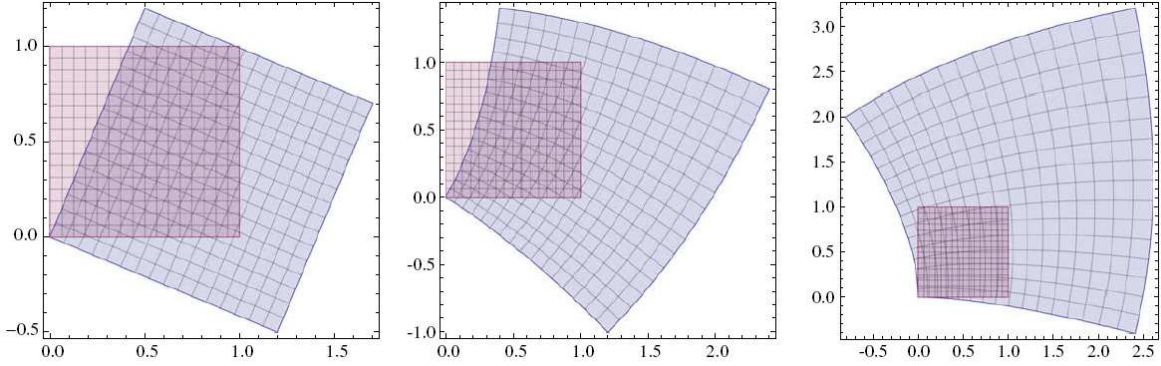}
\centering
\caption{\footnotesize{Infinitesimal conformal mappings \cite{Neff_Jeong_bounded_stiffness09} preserve locally angles and shapes. The corresponding couple stress tensor $\widetilde{m}$ is zero for these deformation modes.}}
\label{capture}
\end{minipage}
\end{figure}%

Using the formulas \eqref{conf}, it can be easily remarked that $\|\nabla [\dev\sym \nabla u]\|^2$, $\|\dev \sym \nabla u\|^2$,\break $\|\dev \sym \nabla (\curl u)\|^2$, $\|\sym \Curl (\sym \nabla u)\|^2=\frac{1}{4}\,\|\sym\nabla (\curl u)\|^2$ are conformally invariant.
Let us note that, using  Lemma \ref{lemaaxlcurl}, we have
\begin{align}
\|\sym \Curl (\sym \nabla u)\|^2&=\|\sym\nabla [\axl (\skw \nabla u)]\|^2=\frac{1}{4}\,\|\sym\nabla (\curl u)\|^2,\\
\|\skw [\Curl(\sym \nabla u)]\|^2&=\|\skw\nabla [\axl (\skw \nabla u)]\|^2=\frac{1}{4}\,\|\skw\nabla (\curl u)\|^2.
\end{align}
Hence $\|\sym \Curl (\sym \nabla u)\|^2=\frac{1}{4}\,\|\sym\nabla (\curl u)\|^2$ is also conformally invariant (use \eqref{conf}$_6$), while
\begin{align}
\|\skw [\Curl (\sym \nabla \phi_c)]\|^2&=\frac{1}{4}\,\|\skw\nabla (\curl \phi_c)\|^2=\|\overline{W}\|^2,
\end{align}
and therefore
$
\|\Curl (\sym\nabla u)\|^2
$ is not conformally invariant, nor is $\|\nabla (\axl (\skw \nabla u))\|^2$ conformally invariant. Nor is $\|\nabla \tr(\sym \nabla u)\|^2=\|\nabla {\rm div} u\|^2$ conformally invariant.

The underlying additional invariance property of the modified couple stress theory is precisely conformal invariance. In the modified couple stress model, these deformations are free of size-effects, while e.g. the Hadjesfandiari and Dargush choice would describe size-effects. In other words, the generated couple stress tensor $\widetilde{m}$ in the modified couple stress model is zero for this inhomogeneous deformation mode, while in the  Hadjesfandiari and Dargush choice $\widetilde{m}$ is constant and skew-symmetric\footnote{This observation is a further development in understanding why the Hadjesfandiari and Dargush \cite{hadjesfandiari2010polar,hadjesfandiari2011couple,hadjesfandiari2014evo} choice is rather meaningless, while mathematically not forbidden \cite{NeffGhibagegenHD}.}.

\subsection{The classical indeterminate couple stress model based on \\$\|\nabla [\axl (\skw \nabla u)]\|^2$ with skew-symmetric nonlocal force-stress}\label{sectaxlb}

We are now re-deriving the classical equations based on the $\nabla [\axl (\skw \nabla u)]$-formulation of the indeterminate couple stress model.  This part does not contain new results, see, e.g., \cite{MadeoGhibaNeffMunchKM} for further details, but is included for setting the stage of our new modelling approach.

Taking free variations $\delta u\in C^\infty(\Omega)$ in the energy $W(e,\widetilde{k})=W_{\rm lin}(e)+W_{\rm curv}(\widetilde{k})$, but using  the following equivalent  curvature energy based on $\widetilde{k}=\nabla [\axl (\skw \nabla u)]=\frac{1}{2}\nabla (\curl u)$\,:
\begin{align}\label{gradeq11}
W_{\rm curv}(\widetilde{k})=& \mu\,L_c^2\,[\alpha_1\, \|\dev\sym \nabla [\axl (\skw \nabla u)]\|^2+
\alpha_2\, \|\skw\nabla [\axl (\skw \nabla u)]\|^2]\\
=& \mu\,L_c^2\,[\alpha_1\, \|\dev\sym \Curl(\sym \nabla u)]\|^2+
\alpha_2\, \|\skw\Curl(\sym \nabla u)]\|^2]\, ,\notag
\end{align}
 we obtain the virtual work principle taking free variations $\delta u\in C^\infty(\Omega)$ in the energy \eqref{gradeq11}
 \begin{align}\label{gradeq211}
\frac{\rm d}{\rm dt}\int_\Omega W(\nabla u+t\nabla \delta u)dv\Big|_{t=0}=&\int_\Omega\bigg[ 2\mu\,\langle\sym \nabla u, \sym \nabla \delta u \rangle+\lambda \tr(\nabla u)\,\tr( \nabla   \delta u)\notag\\&+\mu\,L_c^2\,[2\,  \alpha_1\, \langle \dev\sym \nabla [\axl (\skw \nabla u)],\dev\sym \nabla [\axl (\skw \nabla \delta u)]\rangle \\&+
2\, \alpha_2\, \langle \skw\nabla [\axl (\skw \nabla u)],\skw\nabla [\axl (\skw \nabla \delta u)]\rangle]+\langle f,\delta u\rangle \bigg]dv =0.\notag
\end{align}
The classical divergence theorem
 leads to
\begin{align}\label{germaneq311}
\int_\Omega\langle \Div (\sigma-\widetilde{\tau})+f, \delta u \rangle \, dv-\int_{\partial \Omega}\langle (\sigma-\widetilde{\tau}).\, n, \delta u\rangle \,dv
-\int_{\partial\Omega}&\langle \widetilde{ {m}}.\, n, \axl (\skw \nabla \delta u) \rangle da=0,
\end{align}
for all virtual displacements $ \delta u\in C^\infty(\Omega)$, where   $n$ is the unit outward normal vector at the surface $\partial \Omega$, $\sigma$ is the  symmetric local force-stress tensor
\begin{align}
\sigma=2\, \mu \, \sym \nabla u+\lambda \, \tr(\nabla u)\,\id\! \in {\rm Sym}(3)
\end{align}
 and $\widetilde{\tau}$ represents  the nonlocal force-stress  tensor (which here is automatically skew-symmetric)
 \begin{align}
 \widetilde{\tau}&=\mu\,L_c^2\,[{\alpha_1}\,\anti
{\rm Div}(\dev\sym \nabla [\axl (\skw \nabla u)])+ {\alpha_2}\,\anti {\rm Div}(\skw\nabla [\axl (\skw \nabla u)])\\
&=\mu\,L_c^2\,[\frac{\alpha_1}{2}\anti
{\rm Div}(\dev\sym \nabla (\curl u))+ \frac{\alpha_2}{2}\,\anti {\rm Div}(\skw\nabla (\curl u))]\notag\\
&=\mu\,L_c^2\,\anti
{\rm Div}\left[\frac{\alpha_1}{2}\dev\sym \nabla (\curl u)+ \frac{\alpha_2}{2}\,\skw\nabla (\curl u)\right]\notag=\frac{1}{2}\anti
{\rm Div}[\widetilde{ {m}}]\in \so(3),\notag
 \end{align}
where
\begin{align}
\widetilde{ {m}}&=\mu\,L_c^2\,[{\alpha_1}\sym \nabla (\curl u)+ {\alpha_2}\,\skw\nabla (\curl u)]\notag\\
&=\mu\,L_c^2\,[{\alpha_1}\dev\sym \nabla (\curl u)+ {\alpha_2}\,\skw\nabla (\curl u)]\\
&=\mu\,L_c^2\,[2\,{\alpha_1}\dev\sym \nabla[\axl (\skw \nabla u)]+2\, {\alpha_2}\,\skw\nabla [\axl (\skw \nabla u)]],\notag
\end{align}
is the hyperstress tensor (couple stress tensor) which may or may not be symmetric, depending on the material parameters.

The non-symmetry of force stress is a constitutive assumption.  Thus, if the test function $\delta u\in C^\infty(\Omega)$ also satisfies $\axl (\skw \nabla \delta u)=0$ on $\Gamma$ (equivalently $\curl  \delta u=0$), then we obtain the equilibrium equation
\begin{empheq}[box=\widefbox]{align}
\Div \bigg\{&\underbrace{2\, \mu \, \sym \nabla u+\lambda \, \tr(\nabla u)\,\id}_{\text{\rm symmetric local force stress}\  \sigma\in {\rm Sym}(3)} \notag\\&-\underbrace{\mu\,L_c^2\,[\anti\{
{\rm Div}(\underbrace{\alpha_1\, \dev\sym \nabla [\axl (\skw \nabla u)]+ \alpha_2\,\skw\nabla [\axl (\skw \nabla u)]}_{\text{\rm hyperstress } \frac{1}{2}\widetilde{m}\in \gl(3)})\}]\bigg\}}_{\text{\rm completely skew-symmetric nonlocal force stress}\  \widetilde{\tau}\in\so(3)}+f=0,
\end{empheq}
or equivalently
 \begin{align}\label{ec11}
 \Div \,\widetilde{\sigma}_{\rm total}+f=0,
 \end{align}
 where
 \begin{align}
 \widetilde{\sigma}_{\rm total}=\sigma-\widetilde{\tau}\not\in{\rm Sym}(3).
 \end{align}
The complete consistent boundary conditions for this formulation is presented for the first time in \cite{MadeoGhibaNeffMunchKM,MadeoGhibaNeffMunchCRM} and recapitulated in Figure \ref{limitmodel003} and Figure \ref{limitmodel00}.

\section{The new isotropic gradient elasticity model with symmetric nonlocal force stress and symmetric hyperstresses}\label{str-grad-sapp}
\setcounter{equation}{0}
As {\it independent constitutive variables} for our  novel gradient elastic model we choose now
\begin{align}\label{213}
\varepsilon=\sym \nabla u, \quad \quad \quad \widehat{k}=\Curl(\sym \nabla u)=\Curl \varepsilon.
\end{align}

We use again the orthogonal Lie-algebra decomposition of $\mathbb{R}^{3\times3}$
\begin{align}\label{tratoraxi}
\Curl (\sym \nabla u)&={\dev\sym(\Curl (\sym \nabla u))}+
{\skw\Curl (\sym \nabla u)}\, .
\end{align}
The term $\frac{1}{3}\tr(\Curl (\sym \nabla u))\,\id$ is missing since $\tr(\Curl (\sym \nabla u))=0$ anyway (already $\tr(\Curl S)=0$ for $S\in{\rm Sym}(3)$).
The model is derived from the free energy $W(e,\widehat{k})=W_{\rm lin}(e)+W_{\rm curv}(\widehat{k})$, with
\begin{align}\label{gradeq}
W_{\rm lin}(\varepsilon)&=\,\mu\, \|\sym \nabla u\|^2+\frac{\lambda}{2}\, [\tr(\sym \nabla u)]^2=\,\mu\, \|\dev \sym \nabla u\|^2+\frac{\kappa}{2}\, [\tr(\sym \nabla u)]^2,\\
W_{\rm curv}(\widehat{k})&=\mu\, L_c^2\,\Big[\,\alpha_1\, \underbrace{\|\dev\sym\Curl (\sym \nabla u)\|^2}_{\text{conformally invariant}}+
\alpha_2\, \underbrace{\|\skw\Curl (\sym \nabla u)\|^2}_{\text{not conformally invariant}}\Big]\, ,\notag
\end{align}
where $\alpha_1$, $\alpha_2$ are non-negative constitutive curvature coefficients and $\kappa=\frac{2\mu+3\lambda}{3}$ is the infinitesimal bulk modulus, while $\mu$ is the classical shear modulus.

The {\bf hyperstress}-tensor (moment stress tensor, couple stress tensor)
$$\widehat{ {m}}:=D_{\widehat{k}}W_{\rm curv}(\widehat{k})=\mu\,L_c^2\,[2\,\alpha_1 \,\dev\sym\Curl (\sym \nabla u)+2\,\alpha_2\,\skw\Curl (\sym \nabla u)]$$
 is symmetric in the conformal case $\alpha_2=0$, while the nonlocal force stress tensor is always symmetric, see eq. \eqref{nonlocalcurl}.

Due to isotropy, the curvature energy $W_{\rm curv}(k)$ involves in principle only 2 additional constitutive constants. Taking free variations $\delta u\in C^\infty(\Omega)$ in the energy \eqref{gradeq}, we obtain the virtual work principle
 \begin{align}\label{gradeq2}
\frac{\rm d}{\rm dt}\int_\Omega W(\nabla u+t\nabla \delta u)dv \Big|_{t=0}=&\int_\Omega\bigg[\quad  2\,\mu\,\langle\sym \nabla u, \sym \nabla \delta u \rangle+\lambda \tr(\nabla u)\,\tr( \nabla   \delta u)\notag\\&\qquad +\mu\,L_c^2\,[2\,  \alpha_1\, \langle \dev\sym\Curl (\sym \nabla u),\dev\sym\Curl (\sym \nabla \delta u)\rangle \\&\qquad +
2\, \alpha_2\, \langle \skw\Curl (\sym \nabla u),\skw\Curl (\sym \nabla \delta u)\rangle]+\langle f,\delta u\rangle \bigg]dv =0,\notag
\end{align}
where $f$ is the body force per unit volume.
We have the formulas
\begin{align}\label{formule}
{\rm div} (\varphi_i\, Q_i)&=\langle Q_i, \nabla \, \varphi_i\rangle+\varphi_i\, {\rm div}\,  Q_i\,\qquad \qquad  \text{not summed} ,\\
{\rm div}\,  (R_i\times S_i)&=\langle S_i, \curl \, R_i\rangle-\langle R_i, \curl\,  S_i\rangle\, \quad  \ \ \text{not summed} ,\notag
\end{align}
for all $C^1$-functions $\varphi_i:\Omega\rightarrow \mathbb{R}$ and $Q_i,P_i,S_i:\Omega\rightarrow \mathbb{R}^{3}$, where $\varphi_i$ are the components of the vector $\varphi$ and $Q_i, P_i,S_i$ are the rows of the matrix $Q$, $P$ and $S$, respectively,  where $\times$ denotes the vector product.
If we take in \eqref{formule}
 $
R_i=[\dev\sym \Curl (\sym \nabla u)]_i,  S_i=(\sym \nabla \delta u)_i\, ,
$
we get
\begin{align}
\sum\limits_{i=1}^3{\rm div}\,  ([ \dev\sym\Curl (\sym \nabla u)]_i\times (\sym \nabla \delta u)_i)=&\sum\limits_{i=1}^3\langle (\sym \nabla \delta u)_i, \curl \, [ \dev\sym\Curl (\sym \nabla u)]_i\rangle\\&-\sum\limits_{i=1}^3\langle [ \dev\sym\Curl (\sym \nabla u)]_i, \curl\, (\sym \nabla \delta u)_i\rangle\, .\notag
\end{align}
Hence, we obtain
\begin{align}
\sum\limits_{i=1}^3{\rm div}\,  ([ \dev\sym\Curl (\sym \nabla u)]_i\times (\sym \nabla \delta u)_i)=&\langle (\sym \nabla \delta u), \Curl \, [ \dev\sym\Curl (\sym \nabla u)]\rangle\\&-\langle [ \dev\sym\Curl (\sym \nabla u)], \Curl (\sym \nabla \delta u)\rangle\, .\notag
\end{align}
Doing a similar calculus, but choosing
$
R_i=[\skw \Curl (\sym \nabla u)]_i,  S_i=(\sym \nabla \delta u)_i ,
$
we obtain
\begin{align}
\sum\limits_{i=1}^3{\rm div}\,  ([ \skw\Curl (\sym \nabla u)]_i\times (\sym \nabla \delta u)_i)=&\langle (\sym \nabla \delta u), \Curl \, [ \skw\Curl (\sym \nabla u)]\rangle\\&-\langle [ \skw\Curl (\sym \nabla u)], \Curl (\sym \nabla \delta u)\rangle\, .\notag
\end{align}
The above formulas lead, for all variations $\delta u\in C^\infty(\Omega)$, to
 \begin{align}\label{curvatureg}
\int_\Omega\bigg[ &\alpha_1\, \langle \dev\sym\Curl (\sym \nabla u),\Curl (\sym \nabla \delta u)\rangle+
 \alpha_2\, \langle \skw\Curl (\sym \nabla u),\Curl (\sym \nabla \delta u)\rangle\bigg] dv\notag\\
&=\int_\Omega\bigg[ \alpha_1\, \langle\sym \Curl \, [ \dev\sym\Curl (\sym \nabla u)],  \nabla \delta u, \rangle+\alpha_2\, \langle \sym \Curl \, [ \skw\Curl (\sym \nabla u)], \nabla \delta u, \rangle \\&\quad- \sum\limits_{i=1}^3{\rm div}\, \bigg[\alpha_1\, ([ \dev\sym\Curl (\sym \nabla u)]_i\times (\sym \nabla \delta u)_i)+\alpha_2 ([ \skw\Curl (\sym \nabla u)]_i\times (\sym \nabla \delta u)_i)\bigg] dv\notag\\
&=\int_\Omega\bigg[ \alpha_1\, \langle\sym \Curl \, [ \dev\sym\Curl (\sym \nabla u)],  \nabla \delta u \rangle+\alpha_2\, \langle \sym \Curl \, [ \skw\Curl (\sym \nabla u)], \nabla \delta u \rangle\notag \\&\quad- \sum\limits_{i=1}^3{\rm div}\, \bigg[\alpha_1\, ([ \dev\sym\Curl (\sym \nabla u)]_i\times (\sym \nabla \delta u)_i)+\alpha_2 ([ \skw\Curl (\sym \nabla u)]_i\times (\sym \nabla \delta u)_i)\bigg] dv\,.\notag
\end{align}
Therefore, using the divergence theorem and a special format of the partial integration which is suggested by the matrix $\Curl$-operator, it follows that\footnote{This is an extra constitutive assumption since it is finally the form of the partial integration that determines, on the one hand, which force-stress tensor is generated and, on the other hand, which boundary condition is obtained. It is only the $\Curl$-operator that seems to suggest this choice-but it remains a choice! }
 \begin{align}\label{curvatureg1}
\int_\Omega\bigg[ &\alpha_1\, \langle \dev\sym\Curl (\sym \nabla u),\Curl (\sym \nabla \delta u)\rangle+
 \alpha_2\, \langle \skw\Curl (\sym \nabla u),\Curl (\sym \nabla \delta u)\rangle\bigg] dv\\
&=\int_\Omega\bigg[ \alpha_1\, \langle\sym \Curl \, [ \dev\sym\Curl (\sym \nabla u)],  \nabla \delta u \rangle+\alpha_2\, \langle \sym \Curl \, [ \skw\Curl (\sym \nabla u)], \nabla \delta u \rangle\bigg] dv \notag\\&\quad- \int_{\partial\Omega}\bigg[\sum\limits_{i=1}^3\langle \alpha_1\, ([ \dev\sym\Curl (\sym \nabla u)]_i\times (\sym \nabla \delta u)_i)+\alpha_2  ([ \skw\Curl (\sym \nabla u)]_i\times (\sym \nabla \delta u)_i),n\rangle da\notag
\\
&=-\int_\Omega\bigg[ \alpha_1\, \langle\Div\{\sym \Curl \, [ \dev\sym\Curl (\sym \nabla u)]+\alpha_2\, \langle \sym \Curl \, [ \skw\Curl (\sym \nabla u)]\}, \delta u \rangle\bigg] dv \notag\\&
\qquad+\int_{\partial\Omega}\langle\bigg[ \alpha_1\, \sym \Curl \, [ \dev\sym\Curl (\sym \nabla u)]+\alpha_2\, \sym \Curl \, [ \skw\Curl (\sym \nabla u)\bigg].\, n, \delta u \rangle da\notag\\
&\qquad- \int_{\partial\Omega}\bigg[\sum\limits_{i=1}^3\langle \alpha_1\, [ \dev\sym\Curl (\sym \nabla u)]_i+\alpha_2  [ \skw\Curl (\sym \nabla u)]_i, (\sym \nabla \delta u)_i\times n\rangle da,\notag
\end{align}
where $n$ is the unit outward normal vector at the surface $\partial \Omega$.
Hence, the relation \eqref{gradeq2} leads to
\begin{align}\label{germaneq3}
\int_\Omega\Div \bigg\{2&\, \mu \, \sym \nabla u+\lambda \, \tr(\nabla u)\,\id \\&+ \mu\,L_c^2\,[2\, \alpha_1 \sym \Curl \, [ \dev\sym\Curl (\sym \nabla u)]+2\,\alpha_2\,\sym \Curl \, [ \skw\Curl (\sym \nabla u)]]\bigg\}+f, \delta u \rangle \, dv\notag\\
-\int_{\partial \Omega}&\bigg[  \langle \bigg[2\mu\,\langle\sym \nabla u+\lambda\, \tr(\nabla u)\,\id \bigg].\, n,  \delta u\rangle da\notag
\\
-\int_{\partial\Omega}&\langle2\,\mu\,L_c^2\,\bigg[\alpha_1\, \sym \Curl \, [ \dev\sym\Curl (\sym \nabla u)]+\alpha_2\, \sym \Curl \, [ \skw\Curl (\sym \nabla u)]\bigg].\, n, \delta u \rangle da\notag\\
\qquad+\int_{\partial\Omega}&\bigg[\sum\limits_{i=1}^3 2\, \mu\,L_c^2\,\langle\alpha_1\, [ \dev\sym\Curl (\sym \nabla u)]_i+\alpha_2  [ \skw\Curl (\sym \nabla u)]_i, (\sym \nabla \delta u)_i\times n\rangle da=0\notag,
\end{align}
for all variations $ \delta u\in C^\infty(\Omega)$.

We can write the above variational formulation, for all variations $ \delta u\in C^\infty(\Omega)$, in the following form\footnote{$\langle a\times b,c\rangle=-\langle a, c\times b\rangle$.}
\begin{align}\label{germaneq3n}
\int_\Omega\langle\Div (\sigma+ \widehat{\tau})+f, \delta u \rangle \, dv-\int_{\partial\Omega}  \langle (\sigma+\widehat{\tau}).\, n,  \delta u\rangle da
-\int_{\partial\Omega}\sum\limits_{i=1}^3 \langle \widehat{ {m}}_i\times n, (\sym \nabla \delta u)_i\rangle da=0,
\end{align}
where
\begin{align}\label{nonlocalcurl}
\sigma &=2\, \mu \, \sym \nabla u+\lambda \, \tr(\nabla u)\, \id \in {\rm Sym}(3), \qquad\qquad \qquad \qquad\qquad \qquad\quad\  \text{(local force-stress)}\notag\vspace{1.5mm}\\
\widehat{\tau}&=\mu\,L_c^2\,\sym\{ 2\,\alpha_1 \Curl \, [ \dev\sym\Curl (\sym \nabla u)]+2\,\alpha_2\, \Curl \, [ \skw\Curl (\sym \nabla u)]\}\\
&=\sym\Curl (\widehat{ {m}})\in {\rm Sym}(3), \qquad \qquad \qquad\qquad\qquad\qquad\qquad\qquad\qquad\text{(non-local force stress)}\notag\vspace{1.5mm}\\
\widehat{ {m}}&=\mu\,L_c^2\,\{2\,\alpha_1\dev\sym[\Curl (\sym \nabla u)]+2\,\alpha_2\, \skw[\Curl (\sym \nabla u)]\}\in \gl(3) \qquad (\in {\rm Sym}(3)\ \  \text{for} \ \ \alpha_2=0).\notag
\end{align}
We call $\sigma$ the local force stress tensor, $\widehat{\tau}$  the non-local force stress tensor and  $\widehat{ {m}}=D_{\widehat{k}} W_{\rm curv}(\widehat{k})$ the hyperstress tensor (couple stress tensor).

Thus, if the test function $\delta u\in C^\infty(\Omega)$ also satisfies $(\sym \nabla \delta u)_i\times n=0$ (or equivalently $(\sym \nabla \delta u). \tau =0$ for all tangential vectors $\tau$ at $\Gamma$), then we obtain the equilibrium equation
\begin{empheq}[box=\widefbox]{align}
  \Div \{\!\!\!\!\!\!&\underbrace{2\, \mu \, \sym \nabla u+\lambda \, \tr(\nabla u)\, \id}_{\text{symmetric local force stress}\  \sigma\,\in\, {\rm Sym}(3)} \notag\\&\quad + \underbrace{\mu\,L_c^2\,\sym \Curl \, [ \underbrace{2\, \alpha_1 \dev\sym\Curl (\sym \nabla u)+2\,\alpha_2\,   \skw\Curl (\sym \nabla u)}_{\text{hyperstress} \ \widehat{ {m}}\in \gl(3)}]}_{\text{symmetric nonlocal force stress}\ \widehat{\tau}\,\in\, {\rm Sym}(3)}\}+f=0.
\end{empheq}
The first impulse is to prescribe   on $\Gamma\subseteq\partial \Omega$  the following geometric boundary conditions
 \begin{align}\label{bc1}
 u&=\widehat{u}^0 \qquad \ \qquad \qquad\quad\  \qquad \qquad  \quad \qquad \qquad \qquad\quad  \text{on}\qquad \Gamma,\\
  (\id-n\otimes n)\,(\sym \nabla  u)_i\times n&=(\id-n\otimes n)\,(\sym \nabla  \widehat{u}^0 )_i\times n , \quad i=1,2,3 \qquad \quad\, \text{on} \qquad \Gamma,\notag
 \end{align}
 where $\widehat{u}^0 :\mathbb{R}^3\rightarrow \mathbb{R}^3$ is a prescribed function (i.e. 3+2+2+2=9 boundary conditions), with\footnote{It is always possible to construct a function ${u}_0$ taking on  the desired boundary values and having the needed regularity by solving
 $$
 \left.
 \begin{array}{rcl}
 \int_\Omega \left(\|\sym \nabla u\|^2+\|D^2u\|^2\right) da&\rightarrow&\quad \min.\ \ u,\\
 u\big|_\Gamma&=&u_0,\\
 \curl u\big|_\Gamma&=&\curl u_0\big|_\Gamma
  \end{array}\right\}\quad \Rightarrow\quad u\in H^{2,2}(\Omega).
 $$} $\widehat{u}^0 \in H^{2,2}(\Omega)$,  and the following traction boundary conditions
 on  $\partial \Omega\setminus \overline{\Gamma}$
 \begin{align}\label{bc1}
(\sigma+\widehat{\tau}).n&=\widehat{g}, \   \ \qquad \qquad \qquad \qquad\quad\quad\qquad \quad \qquad\text{on}\qquad \partial \Omega\setminus \overline{\Gamma},\\
 (\id-n\otimes n)\,\widehat{m}_i\times n&=(\id-n\otimes n)\,\widehat{g}_i , \quad i=1,2,3 \  \qquad \quad \quad \ \ \text{on} \qquad \partial \Omega\setminus \overline{\Gamma},\notag
 \end{align}
where $\widehat{g}, \widehat{g}_i:\mathbb{R}^3\rightarrow \mathbb{R}^3$ are prescribed functions (i.e. 3+2+2+2=9 boundary conditions).

However, we need to separate normal and tangential derivatives of the test function  $\delta u$ in \eqref{germaneq3n} which is standard in general strain gradient elasticity, since tangential derivatives of $\delta u$ are not independent of  $\delta u$. Let us define the matrix
\begin{align}
\widehat{M}:= \left(
    \begin{array}{c}
      \widehat{m}_1\times n\_\_\_\_ \\
    \widehat{m}_2\times n\_\_\_\_\\
    \widehat{m}_3\times n\_\_\_\_ \\
    \end{array}
  \right), \qquad \text{where}\qquad \widehat{m}:= \left(
    \begin{array}{c}
      \widehat{m}_1 \_\_\_\_\\
    \widehat{m}_2\_\_\_\_\\
    \widehat{m}_3 \_\_\_\_\\
    \end{array}
  \right).
\end{align}
With the help of this matrix $\widehat{M}$, we may write
\begin{align}
\sum\limits_{i=1}^3 \langle \widehat{ {m}}_i\times n, (\sym \nabla  \delta u)_i\rangle=\langle \widehat{M}, \sym \nabla  \delta u\rangle=
\langle \sym \widehat{M}, \nabla  \delta u\rangle.
\end{align}

At this point, it must be considered that the tangential trace of
the gradient of virtual displacement can be integrated by parts once
again and that the surface divergence theorem can be applied to this
tangential  part of $\nabla\delta{u}$. Before doing so, one needs
to introduce (see also \cite{NthGrad,TheseSeppecher,dell2012beyond} for details)
two second order tensors $T$ and $Q$ which are  the two projectors
on the tangent plane and on the normal to the considered surface, respectively.
As it is well known from differential geometry, such projectors actually
allow to split a given vector or tensor field in one part projected
on the plane tangent to the considered surface and one projected on
the normal to such surface (see e.g. \cite{NthGrad}). Let $\left\{ \tau^{1},\tau^{2}\right\} $
be an orthonormal local basis of the tangent plane to the considered
surface at point $p$ and let ${n}$ be the unit normal vector
at the same point. We can introduce the quoted projectors as
\begin{align}
T&=\tau^{1}\otimes\tau^{1}+\tau^{2}\otimes\tau^{2}=\id -{n}\otimes{n}\quad (T_{ij}=\delta_{ij}-n_in_j),\qquad {Q}={n}\otimes{n}\qquad (Q_{ij}=n_in_j).
\end{align}

 In our abbreviations,  the surface divergence theorem means \cite[p.~58, ex.~7]{gurtin2010mechanics}
\begin{equation}
\int_{\partial S}\langle T.\, v,{\nu}\rangle da=\int_{\partial S}\langle  v,T.\,{\nu}\rangle da=\int_{\partial S}\langle v, {\nu}\rangle ds\,,\label{eq:Surface_Div_Th}
\end{equation}
for any field $v\in \mathbb{R}^3$ and $\nu=\tau\times n$.
Regarding the boundary conditions, similar  as in \cite{Mindlin68}, we obtain
$$\begin{array}{ll}
\dd\sum\limits_{i=1}^3&\langle \widehat{ {m}}_i\times n, (\sym \nabla  \delta u)_i\rangle =\langle \sym \widehat{M}, \nabla  \delta u\, \id\rangle =\langle \sym \widehat{M},\nabla\delta u\,(T+Q)\rangle \notag \vspace*{-3mm}\\
&=\langle (\sym \widehat{M}),(\nabla\delta u)\, T\rangle +\langle (\sym \widehat{M}),(\nabla\delta u)\,Q\rangle =\langle (\sym \widehat{M})\,T,\nabla\delta u\rangle +\langle (\sym \widehat{M}),(\nabla\delta u)\,Q\rangle \\
&=\langle (\sym \widehat{M})\,T\, T,\nabla\delta u\rangle +\langle (\sym \widehat{M}),(\nabla\delta u)\,Q\rangle
=\langle \, T,T(\sym \widehat{M})\nabla\delta u\rangle +\langle (\sym \widehat{M}),(\nabla\delta u)\,Q\rangle
\notag\\
&=\langle \, T,T(\sym \widehat{M})\nabla\delta u\rangle +\langle (\sym \widehat{M}),(\nabla\delta u)\,n\otimes n\rangle.\notag
\end{array}
$$
The last term on the right hand side  may be rewritten in the form
\begin{equation}\begin{array}{ll}\label{gh1}
\langle (\sym \widehat{M}),(\nabla\delta u)\,Q\rangle &=\dd\frac{1}{2}
\langle  \widehat{M},(\nabla\delta u)\,Q+Q(\nabla\delta u)^T\rangle=\dd\frac{1}{2}
\langle  \widehat{M}\,Q,\nabla\delta u\rangle+\frac{1}{2}
\langle  \widehat{M},Q(\nabla\delta u)^T\rangle\vspace{2mm}\\
&=\dd\frac{1}{2}
\langle  \widehat{M}\,n\otimes n,\nabla\delta u\rangle+\frac{1}{2}
\langle  \widehat{M},n\otimes n(\nabla\delta u)^T\rangle.
\end{array}
\end{equation}
Thus, we deduce
\begin{equation}\begin{array}{ll}\label{gh2}
\dd\sum\limits_{i=1}^3&\langle \widehat{ {m}}_i\times n, (\sym \nabla  \delta u)_i\rangle
=\langle \, T,T(\sym \widehat{M})\nabla\delta u\rangle +\dd\langle (\sym \widehat{M}),(\nabla\delta u)\,Q\rangle
\vspace*{-3mm}\\
&=\langle \, T,T(\sym \widehat{M})\nabla\delta u\rangle +\dd\langle (\sym \widehat{M}),(\nabla\delta u)\,n\otimes n\rangle=\langle \, T,T(\sym \widehat{M})\nabla\delta u\rangle +\dd\langle\{ n\otimes [(\nabla\delta u).n]\} (\sym \widehat{M}),\id\rangle\notag
\\
&=\langle \, T,T(\sym \widehat{M})\nabla\delta u\rangle +\dd\langle n\otimes \{(\sym \widehat{M})(\nabla\delta u).n\},\id\rangle
=\langle \, T,T(\sym \widehat{M})\nabla\delta u\rangle +\dd\langle n,(\sym \widehat{M})(\nabla\delta u).n\rangle
\notag\\
&=\langle \, T,T(\sym \widehat{M})\nabla\delta u\rangle +\dd\langle (\sym \widehat{M}).n,(\nabla\delta u).n\rangle
=\langle \, T,T(\sym \widehat{M})\nabla\delta u\rangle +\dd\langle (\sym \widehat{M}).n,\frac{\partial \delta u}{\partial n}\rangle  \,.\notag
\end{array}
\end{equation}
We can therefore recognize in the last term of this formula that the normal
derivative
\begin{align}
\frac{\partial}{\partial n}\,\delta u=(\nabla\delta u).n=(\delta u_{i,h}n_{h})_i
\end{align}
of the test function  field $\delta u$ (the virtual displacement) appears.
As for the other term, it can be manipulated suitably integrating
by parts and then using the surface divergence theorem (\ref{eq:Surface_Div_Th}),
so that we can finally write the last summand  from \eqref{germaneq3n} in the form
\begin{align}\label{gh3}
\int_{\partial \Omega}\sum\limits_{i=1}^3\langle \widehat{ {m}}_i\times n, (\sym \nabla  \delta u)_i\rangle da=&\int_{ \Gamma}\sum\limits_{i=1}^3\langle \widehat{ {m}}_i\times n, (\sym \nabla  \delta u)_i\rangle da+\int_{\partial \Omega\setminus \overline{\Gamma}}\sum\limits_{i=1}^3\langle \widehat{ {m}}_i\times n, (\sym \nabla  \delta u)_i\rangle da.
\end{align}
We deduce by gathering the results in \eqref{gh1}-\eqref{gh3} that the last integral on the right hand side is given by
\begin{align}
&\int_{\partial \Omega\setminus \overline{\Gamma}}\sum\limits_{i=1}^3\langle \widehat{ {m}}_i\times n, (\sym \nabla  \delta u)_i\rangle da=\int_{\partial \Omega\setminus \overline{\Gamma}}\langle \, T,T(\sym \widehat{M})\nabla\delta u\rangle da +\int_{\partial \Omega\setminus \overline{\Gamma}}\dd\langle (\sym \widehat{M}).n,(\nabla\delta u).n\rangle  da.
\end{align}
Since
\begin{align}\{\nabla[\underbrace{T(\sym \widehat{M})\delta u}_{\in \mathbb{R}^3}]\}_{ik}&=
\{T(\sym \widehat{M})).\delta u\}_{i,k}\notag=
\{\big(T(\sym \widehat{M})\big)_{ij}(\delta u)_{j}\}_{,k}\notag\\
&=
\big(T(\sym \widehat{M})\big)_{ij,k}(\delta u)_{j}+
\big(T(\sym \widehat{M})\big)_{ij}(\delta u)_{j,k}\\
&=
\big(T(\sym \widehat{M})\big)_{ij,k}(\delta u)_{j}+
\big(T(\sym \widehat{M})\big)_{ij}(\nabla \delta u)_{jk}\notag\\
&=
\big(T(\sym \widehat{M})\big)_{ij,k}(\delta u)_{j}+\{T(\sym \widehat{M})\,\nabla\delta u\}_{ik},\notag
\end{align}
we obtain
\begin{align}
\int_{\partial \Omega\setminus \overline{\Gamma}}\sum\limits_{i=1}^3\langle \widehat{ {m}}_i\times n, (\sym \nabla  \delta u)_i\rangle da=&\int_{\partial \Omega\setminus \overline{\Gamma}}\langle \, \hspace{-5mm}\underbrace{T,\nabla [T(\sym \widehat{M}).\,\delta u]}_{\text{surface divergence is to be used}}\hspace{-5mm}\rangle  da -\int_{\partial \Omega\setminus \overline{\Gamma}}T_{ik}\big(T(\sym \widehat{M})\big)_{ij,k}(\delta u)_{j}  da\\&+\int_{\partial \Omega\setminus \overline{\Gamma}}\dd\langle (\sym \widehat{M}).n,(\nabla\delta u).n\rangle  da.\notag
\end{align}
In order to write in a compact form the above relation, let us remark that
\begin{align}
\big(T(\sym \widehat{M})\big)_{ij}&=T_{il}(\sym \widehat{M})_{lj}\notag=
T_{il}(\sym \widehat{M})_{jl}=T_{li}(\sym \widehat{M})_{jl}=(\sym \widehat{M})_{jl}T_{li}=\{(\sym \widehat{M})\,T\}_{ji},\notag
\end{align}
and further that
\begin{align}
\big(T(\sym \widehat{M})\big)_{ij,k}\,T_{ik}&=\{(\sym \widehat{M})\,T\}_{ji,k}\,T_{ik}=
\big(\nabla[(\sym \widehat{M})\,T]: T\big)_{j}.
\end{align}
We obtain
\begin{align}
T_{ik}\big(T(\sym \widehat{M})\big)_{ij,k}(\delta u)_{j}=\big(\nabla[(\sym \widehat{M})\,T]: T\big)_{j}(\delta u)_{j}=
\langle\nabla[(\sym \widehat{M})\,T]: T,\delta u\rangle.
\end{align}
We deduce by gathering the results in \eqref{gh1}-\eqref{gh3} that the last integral on the right hand side is given by
\begin{align}
&\int_{\partial \Omega\setminus \overline{\Gamma}}\sum\limits_{i=1}^3\langle \widehat{ {m}}_i\times n, (\sym \nabla  \delta u)_i\rangle da\notag\\&=\int_{\partial \Omega\setminus \overline{\Gamma}}\langle \, \hspace{-5mm}\underbrace{T,\nabla [T(\sym \widehat{M}).\,\delta u]}_{\text{surface divergence is to be used}}\hspace{-5mm}\rangle da-\int_{\partial \Omega\setminus \overline{\Gamma}}\langle\nabla[(\sym \widehat{M})\,T]: T,\delta u\rangle da+\,\int_{\partial \Omega\setminus \overline{\Gamma}}\langle (\sym \widehat{M}).n,(\nabla\delta u).n\rangle  da\notag\\
&= -\int_{\partial \Omega\setminus \overline{\Gamma}}\langle  \nabla[(\sym \widehat{M})\,T]: T,\delta u\rangle da+\,\int_{\partial \Omega\setminus \overline{\Gamma}}\langle (\sym \widehat{M}).n,(\nabla\delta u).n\rangle  da+\int_{\partial(\partial \Omega\setminus \overline{\Gamma})}\langle [(\sym \widehat{M})]^-. \nu^-, \delta u\rangle ds.\notag
\end{align}
In the above computation $ \nabla[(\sym \widehat{M})\,T]$ is not a matrix, rather a third order tensor and  $ \nabla[(\sym \widehat{M})\,T]: T\in \mathbb{R}^3$ is a contraction operation, i.e. $$\{\nabla[(\sym \widehat{M})\,T]: T\}_i=\{(\anti[(\sym \widehat{M})\,T\}_{ij,k}.T_{jk}.$$

Similar, we handle the corresponding integral on $\partial \Gamma$ from \eqref{gh3}
 \begin{align}
\int_{ \Gamma}\sum\limits_{i=1}^3&\langle \widehat{ {m}}_i\times n, (\sym \nabla  \delta u)_i\rangle da\\
&= -\int_{\Gamma}\langle  \nabla[(\sym \widehat{M})\,T]: T,\delta u\rangle da+\,\int_{\Gamma}\langle (\sym \widehat{M}).n,(\nabla\delta u).n\rangle  da+\int_{\partial\Gamma}\langle [(\sym \widehat{M})]^+. \nu^+, \delta u\rangle ds.\notag
\end{align}

Therefore, the  variational formulation \eqref{germaneq3n} can  be rewritten as
\begin{align}\label{germaneq3n0}
\int_\Omega\langle\Div (\sigma+ \widehat{\tau})+f, \delta u \rangle \, dv
&-\int_{\partial \Omega}  \underbrace{\langle\underbrace{ (\sigma+\widehat{\tau}).\, n-\nabla[(\sym \widehat{M})\,T]: T}_{\begin{array}{c}\footnotesize{\delta u}-\text{\footnotesize{surface contact forces depending}}\vspace{-0.5mm}\\
 \text{\footnotesize{on the curvature of the boundary}}\end{array}},  \delta u\rangle}_{\delta u-\text{independent first order variation}} da\\
&-\,\int_{\partial  \Omega}\hspace{-0.2cm}\underbrace{\langle \underbrace{(\sym \widehat{M}).n}_{\text{double force}},(\nabla\delta u).n\rangle}_{\begin{array}{c}\footnotesize{\delta u}-\text{\footnotesize{independent second order}}\vspace{-0.5mm}\\
 \text{\footnotesize{normal variation of the gradient}}\end{array}}\hspace{-0.2cm}  da-
\int_{\partial\Gamma}\underbrace{\langle ([\sym \widehat{M}]^+-[\sym \widehat{M}]^-).\, \nu, \delta u\rangle}_{\text{``edge line force"}} ds=0,\notag
\end{align}
for all variations $ \delta u\in C^\infty(\Omega)$, where we have used that for the regular surface $\partial \Omega$ it holds $\nu^+=-\nu^-=\nu$. Moreover, we also obtain
\begin{align}
\widehat{M}. n=\left[\left(
    \begin{array}{c}
      \widehat{m}_1\times n \\
    \widehat{m}_2\times n \\
    \widehat{m}_3\times n \\
    \end{array}
  \right).n\right]=
  \left[\left(
    \begin{array}{c}
      \langle\widehat{m}_1\times n,n\rangle \\
     \langle\widehat{m}_2\times n,n\rangle \\
     \langle\widehat{m}_3\times n,n\rangle \\
    \end{array}
  \right).n\right]=0.
\end{align}
Hence
\begin{align}\label{simplifM}
(\sym \widehat{M}).n
=\frac{1}{2} \widehat{M}^T.n.
\end{align}
On the other hand, we deduce
\begin{align}
\widehat{M}^T.n&=T\, \widehat{M}^T.n+Q\, \widehat{M}^T.n=(\id-n\otimes n)\, \widehat{M}^T.n+n\otimes n\, \widehat{M}^T.n\notag\\
&=(\id-n\otimes n)\, \widehat{M}^T.n+[n\otimes n\, \widehat{M}^T].n
=(\id-n\otimes n)\, \widehat{M}^T.n+n\otimes [\widehat{M}.\, n].n\notag\\
&=(\id-n\otimes n)\, \widehat{M}^T.n+n\otimes \left[\left(
    \begin{array}{c}
      \widehat{m}_1\times n \\
    \widehat{m}_2\times n \\
    \widehat{m}_3\times n \\
    \end{array}
  \right).n\right].n
  =(\id-n\otimes n)\, \widehat{M}^T.n+n\otimes \left[\left(
    \begin{array}{c}
      \langle\widehat{m}_1\times n,n\rangle \\
     \langle\widehat{m}_2\times n,n\rangle \\
     \langle\widehat{m}_3\times n,n\rangle \\
    \end{array}
  \right)\right].n\notag
  \\
&=(\id-n\otimes n)\, \widehat{M}^T.n+[n\otimes 0].n=(\id-n\otimes n)\, \widehat{M}^T.n\,.\notag
\end{align}
In view of \eqref{simplifM}, we see
\begin{align}
(\sym\widehat{M}).n=(\id-n\otimes n)\, (\sym \widehat{M}).n\,.
\end{align}
Therefore, finally we get from \eqref{germaneq3n}
\begin{align}\label{germaneq3n0}
\int_\Omega\langle\Div (\sigma+ &\widehat{\tau})+f, \delta u \rangle \, dv
-\int_{\partial \Omega}  \underbrace{\langle (\sigma+\widehat{\tau}).\, n-\nabla[(\sym \widehat{M})\,T]: T,  \delta u\rangle}_{\delta u-\text{independent first order variation}} da\\
&-\,\int_{\partial  \Omega}\underbrace{\langle (\id-n\otimes n)(\sym\widehat{M}).n,(\id-n\otimes n)(\nabla\delta u).n\rangle}_{\begin{array}{c}\delta u-\text{\footnotesize{independent second order}}\vspace{-0.5mm}\\
 \text{\footnotesize{normal variation of gradient}}\end{array}}  da-\int_{\partial\Gamma}\underbrace{\langle ([\sym \widehat{M}]^+-[\sym \widehat{M}]^-).\, \nu, \delta u\rangle}_{ \text{\footnotesize{``edge line forces"}}} ds=0,\notag
\end{align}
for all variations $ \delta u\in C^\infty(\Omega)$. An equivalent form, replacing simply $\widehat{\tau}=\sym\Curl (\widehat{ {m}})$, is
\begin{align}\label{germaneq3n01}
\int_\Omega\langle\Div (\sigma+& \sym\Curl (\widehat{ {m}}))+f, \delta u \rangle \, dv
-\int_{\partial \Omega}  \underbrace{\langle (\sigma+\sym\Curl (\widehat{ {m}})).\, n-\nabla[(\sym \widehat{M})\,T]: T,  \delta u\rangle}_{\delta u-\text{independent first order variation}} da\\
&-\,\int_{\partial  \Omega}\underbrace{\langle (\id-n\otimes n)(\sym \widehat{M}).n,(\id-n\otimes n)(\nabla\delta u).n\rangle}_{\begin{array}{c}\delta u-\text{\footnotesize{independent second order}}\vspace{-0.5mm}\\
 \text{\footnotesize{normal variation of the gradient}}\end{array}} da-
\int_{\partial\Gamma}\underbrace{\langle ([\sym \widehat{M}]^+-[\sym \widehat{M}]^-).\, \nu, \delta u\rangle}_{ \text{\footnotesize{``edge line forces"}}} ds=0.\notag
\end{align}
\subsection{Formulation of the complete boundary value problem}
\subsubsection{Equilibrium equation}
In terms of the  symmetric  force-stress tensor $\sigma$
 and of the nonlocal force-stress tensor $\widehat{\tau}$ which is also here symmetric, while  the hyperstress $\widehat{ {m}}\in \gl(3)$ is symmetric only for $\alpha_2=0$,
 the equilibrium equations may now be written in the format\footnote{Here, infinitesimal frame-indifference amounts to $W(\nabla u)=W(\nabla u+\overline{W}), \forall\, \overline{W}\in\so(3)$, which is obviously satisfied.}
 \begin{align}\label{ec1}
 \Div \widehat{\sigma} _{\rm total}+f=0,
 \end{align}
 where the symmetric total force stress\footnote{Vidoli et al. call this tensor the ``effective stress tensor'' \cite{IsolaSciarraVidoliPRSA}.} is given by
$
 \widehat{\sigma}_{\rm total}=\sigma +\widehat{\tau} \in {\rm Sym}(3).
 $

 \subsubsection{Geometric (essential) boundary conditions}
 To the above equilibrium equation, we adjoin on $\Gamma\subseteq\partial \Omega$  the following boundary conditions
 \begin{align}\label{bc111}
 u(x)&=\widehat{u}^0 (x) \ \  \qquad\qquad \qquad \qquad\quad\quad  \text{on}\qquad\, \Gamma,\qquad\qquad\qquad\qquad\  (3\  \text{bc})\\
  [(\id-n\otimes n)(\nabla u).n]\, (x)&= [(\id-n\otimes n)(\nabla \widehat{u}^0 ).n] (x)  \qquad \ \ \text{on} \qquad \Gamma, \ \qquad\qquad\qquad\qquad (2\  \text{bc})\notag
 \end{align}
 where $\widehat{u}^0 :\mathbb{R}^{3}\rightarrow\mathbb{R}^{3}$ is a prescribed function (i.e. 3+2=5 boundary conditions). We assume that $\widehat{u}^0 \in H^{2,2}(\Omega)$  for simplicity and transparency. If $\nabla \widehat{u}^0 \in H^1(\Omega)$, then all tangential and normal traces of $\nabla \widehat{u}^0 $ at $\Gamma$ exist. Therefore, we may evaluate $\nabla \widehat{u}^0 $ at $\Gamma$.

\subsubsection{Traction  boundary conditions}

Corresponding to the geometric  boundary conditions, we have to prescribe the following traction boundary conditions
  \begin{align}\label{bc100}
  \hspace{-0.5cm}
  \left.
        \begin{array}{rcl}
 \{(\sigma+\widehat{\tau}).n-\nabla[(\sym \widehat{M})\,(\id-n\otimes n)]: (\id-n\otimes n)\}\,(x)\!\!&=& \!\!\widehat{t}(x),\vspace{1.2mm}\\
  \dd[(\id-n\otimes n)(\sym \widehat{M}).n]\,(x)\!\!&=&\!\![(\id-n\otimes n)\,\widehat{g}]\,(x),
 \end{array}
 \right\}
 \quad
  &x\in \partial \Omega\setminus {\overline{\Gamma}}
  \quad
 \begin{array}{r}
(3\  \text{bc})\vspace{1.2mm}\\
(2\  \text{bc})
 \end{array}\\
\hspace{-2cm} \begin{array}{rcl}
 \{([\sym \widehat{M}]^+-[\sym \widehat{M}]^-).\,\nu\}\,(x)\!\!&=&\!\!\widehat{\pi}(x),\hspace{2.35cm}
 \end{array}
 \quad
  &x\in \partial {\Gamma}
  \quad \quad\ \
 \begin{array}{r}
(3\  \text{bc})
 \end{array}\notag
 \end{align}
 where  $t,  g:\mathbb{R}^{3}\rightarrow\mathbb{R}^{3}$ are prescribed functions  on $\partial \Omega\setminus \overline{\Gamma}$ (i.e. 3+2=5 boundary conditions), while $\pi:\mathbb{R}^3\rightarrow\mathbb{R}^3$ is prescribed on $\partial \Gamma=\partial (\partial \Omega\setminus \overline{\Gamma})$ and leads to 3 boundary conditions on $\partial \Gamma$.

\begin{remark}
If $\Gamma=\partial \Omega$ then the solution in the $\Curl(\sym \nabla u)$-formulation and in the $\nabla[ {\rm axl}(\skw \nabla u)]$-formulation are the same, since the Euler-Lagrange equations are the same and the geometric boundary conditions are the same. Differences appear only if $\Gamma\neq \partial \Omega$ due to different specifications of traction boundary conditions.
\end{remark}
\subsection{Existence and uniqueness of the solution in the $\Curl(\sym \nabla u)$-formulation}\label{existences}

In the linear couple stress theory with constrained rotations, Hlav{\'a}{\v c}ek and Hlav{\'a}{\v c}ek \cite[Remark 2, p. 426]{Hlavacek69} recognized the couple stress model already in the form  \eqref{energiepolar2} but did not give an existence result. There are many existence and uniqueness results  for the indeterminate couple stress model in its classical anti-symmetric formulation. Recently, optimal results have been obtained in \cite{Neff_JeongMMS08,Jeong_Neff_ZAMM08}.
In this section we  establish an existence
theorem for the solution of the boundary value problem $(\mathcal{P})$ defined by \eqref{ec1}, \eqref{bc111} and \eqref{bc100}, where $\widehat{t}=0$, $\widehat{g}=0$, $\widehat{h}=0$, $\widehat{u}^0 =0$ and $ (\id-n\otimes n)(\nabla \widehat{u}^0 ).n=0$ for simplicity only.

\begin{lemma}\label{lemma30}
Let  $u\in {\rm H}_0^1(\Omega;\Gamma)$ be such that $\sym \Curl (\sym \nabla  u)\in{\rm L}^2(\Omega)$. Then, $\Curl (\sym \nabla  u)\in{\rm L}^2(\Omega)$ and there is a positive constant $c^+$ such that
\begin{align}\label{lema30exp}
\int_\Omega \Big[\|\sym \nabla u\|^2+\|\sym \Curl (\sym \nabla  u)\|^2\Big] dv&\geq c^+\int_\Omega \Big[\|\sym \nabla u\|^2+\|\Curl (\sym \nabla  u)\|^2\Big] dv.
\end{align}
\end{lemma}
\begin{proof}
For $u\in {\rm H}_0^1(\Omega;\Gamma)$, the first Korn's inequality implies that there is a positive constant $c^+$ such that
\begin{align}
\int_\Omega \|\sym \nabla u\|^2 dv=\int_\Omega \left(\frac{1}{2}\|\sym \nabla u\|^2+\frac{1}{2}\|\sym \nabla u\|^2\right) dv\geq\frac{1}{2}\int_\Omega \|\sym \nabla u\|^2 dv+ \frac{c^+}{2}  \int_\Omega \|\nabla u\|^2 dv.
\end{align}
On the other hand the orthogonality of $\sym$ and $\skew$ implies
\begin{align}
\int_\Omega \|\nabla u\|^2 dv\geq \int_\Omega \|\skew \nabla u\|^2 dv.
\end{align}
Therefore, there is another positive constant $c^+$ such that
\begin{align}\label{ine1}
\int_\Omega \Big[\|\sym \nabla u\|^2&+\|\sym \Curl (\sym \nabla  u)\|^2\Big] dv=\int_\Omega \Big[\frac{1}{2}\|\sym \nabla u\|^2+\frac{1}{2}\|\sym \nabla u\|^2+\|\sym \nabla [\axl (\skw \nabla u)]\|^2\Big] dv\notag\\
&\geq c^+\int_\Omega \Big[\|\sym \nabla u\|^2+\|\skew \nabla u\|^2+\|\sym \nabla [\axl (\skw \nabla u)]\|^2\Big] dv\\
&\geq c^+\int_\Omega \Big[\|\sym \nabla u\|^2+\|\axl \skew \nabla u\|^2+\|\sym \nabla [\axl (\skw \nabla u)]\|^2 \Big]dv.\notag
\end{align}
Moreover, since $\axl \skew \nabla u\in {\rm L}^2(\Omega)$, the second  Korn's inequality\footnote{Since $\curl u$ is divergence free we also have the following Maxwell type inequality \cite{BNPS3,BNPS2}:
\begin{align*}
\|\nabla [\curl u]\|_{L^2(\Omega)}^2\leq c_M( \|{\rm curl}\,(\curl  u)\|_{L^2(\Omega)}^2\, ,\ \  \text{for}\ \  u\in \{u\in H_0^1(\Omega)\, |\, \curl u\in H(\curl; \Omega)\}.
\end{align*}
} (without boundary conditions and applied to $\axl (\skw \nabla u)$) implies the existence of a positive constant $c^+$ such that
\begin{align}\label{korn2}
\int_\Omega \Big[\|\axl \skew \nabla u\|^2+\|\sym \nabla [\axl (\skw \nabla u)]\|^2\Big] dv
\geq c^+\int_\Omega \Big[\|\axl \skew \nabla u\|^2+\|\nabla [\axl (\skw \nabla u)]\|^2\Big] dv.
\end{align}
Thus, there are  positive constants $c^+,c_1^+$ such that
\begin{align}\label{korn2}
\int_\Omega \Big[\|\sym \nabla u\|^2+\|\sym \Curl (\sym \nabla  u)\|^2\Big] dv&
\geq c_1^+\int_\Omega \Big[\|\sym \nabla u\|^2+\|\axl \skew \nabla u\|^2+\|\nabla [\axl (\skw \nabla u)]\|^2\Big] dv\notag\\
&= c_1^+\int_\Omega \Big[\|\sym \nabla u\|^2+\|\axl \skew \nabla u\|^2+\|\Curl (\sym \nabla  u)\|^2\Big] dv\notag\\
&\geq c^+\int_\Omega \Big[\|\sym \nabla u\|^2+\|\Curl (\sym \nabla  u)\|^2\Big] dv.
\end{align}
The proof is complete.
\end{proof}

Let us consider that we have considered null boundary conditions for simplicity. Hence, in the following we study the existence of the solution in the space
\begin{equation}
{\mathcal{X}_0}\,{=}\,\big\{ u\,{\in}\,{H}^1_0(\Omega; \Gamma)\,|\,\sym \nabla u\in \, H(\Curl; \Omega)\big\}.
\end{equation}
On ${\mathcal{X}_0}$ we define the norm
\begin{equation}
\| u \|_{\mathcal{X}_0}=\left(\|\nabla u\|^2_{L^2(\Omega)}+\|\Curl(\sym \nabla u)\|^2_{L^2(\Omega)} \right)^{\frac{1}{2}},
\end{equation}
and the bilinear form
\begin{align}\label{proscalar}
 (u,v)=\int_\Omega\bigg[& 2\,\mu\,\langle\sym \nabla u, \sym \nabla v \rangle+\lambda \tr(\nabla u)\,\tr( \nabla   v)\notag\\&+\mu\,L_c^2\,[2\,  \alpha_1\, \langle \dev\sym\Curl (\sym \nabla u),\dev\sym\Curl (\sym \nabla v)\rangle \\&+
2\, \alpha_2\, \langle \skw\Curl (\sym \nabla u),\skw\Curl (\sym \nabla v)\rangle]\bigg]dv ,\notag
\end{align}
where
$u,v\in{\mathcal{X}_0}$. Let us define the linear operator  $l:{\mathcal{X}_0}\rightarrow\mathbb{R}$, describing the influence of external loads,
$
l(v)=\int_\Omega \langle f, v \rangle dv$  {for all} $\widetilde{w}\in{\mathcal{X}_0}.
$ We say that $w$ is a weak solution of the problem $(\mathcal{P})$ if and only if
\begin{equation}\label{wf}
(u,v)=l(v)  \ \  \text{ for all } \ \   v\in {\mathcal{X}_0}.
\end{equation}
A classical solution $u\in{\mathcal{X}_0}$
 of the problem $(\mathcal{P})$ is also a weak solution.

\begin{theorem}\label{thex}
Assume that
\begin{itemize}
\item[i)] the constitutive coefficients satisfy $ \ \ \mu>0, \quad 3\, \lambda+2\mu>0, \quad \alpha_1>0,\quad \alpha_2\geq 0$;
\item[ii)] the loads satisfy the regularity condition $f\in L^2(\Omega)$.
\end{itemize}
Then there exists one and only one solution of the problem {\rm (\ref{wf})}.
\end{theorem}
\begin{proof}
Let us first consider the case $\boldsymbol{\alpha_2> 0}$.
The Cauchy-Schwarz inequality, the inequalities $(a\pm b)^2\leq 2(a^2+b^2)$ and the assumption upon the constitutive coefficients lead to
\begin{align}
(u,v)&\leq \dd \,C\, \Bigg[\int _\Omega\biggl(\|\sym \nabla u\|^2+\|\Curl(\sym \nabla u)\|^2\biggl)dv\Bigg]^{\frac{1}{2}}
\Bigg[\int _\Omega\biggl(\|\sym \nabla{v}\|^2+\|\Curl(\sym \nabla v)\|^2\biggl)dv\Bigg]^{\frac{1}{2}}\\
&\leq \dd \,C\, \|w\|_{{\mathcal{X}_0}}\,\,\|\widetilde{w}\|_{{\mathcal{X}_0}}\, ,\notag
\end{align}
which means that $(\cdot,\cdot)$ is bounded. On the other hand, we have
\begin{align}
({u},{u})
=\int_\Omega\bigg[& 2\mu\,\|\sym \nabla u\|^2+\lambda\,[ \tr(\nabla u)]^2\notag\\
&+\mu\,L_c^2\,[2\,  \alpha_1\, \|\dev\sym\Curl (\sym \nabla u)\|^2 +
2\, \alpha_2\, \|\skw\Curl (\sym \nabla u)\|^2]\bigg]dv ,\notag
\end{align}
for all
$u\in {\mathcal{X}_0}$. Moreover, as a consequence of the properties  i) of the constitutive coefficients
 we have that there exists the positive constant $c$
\begin{align}
({u},{u})
&\geq \dd c\,\int _\Omega\biggl(\| \sym\nabla u\|^2+ \|\Curl(\sym \nabla u) \|^2\biggl)\, dv.
\end{align}
From  linearized elasticity we have the first   Korn's inequality \cite{Neff00b}, that is
\begin{align}\label{Korn0}
\| \nabla  u\|_{L^2(\Omega)}\leq C \|\sym {\nabla}\, u\|_{L^2(\Omega)}\, ,
\end{align}
for all functions $u\in H_0^1(\Omega;\Gamma)$ with some constant $C>0$, for bounding the deformation of an elastic medium in terms of the symmetric strains.
Hence, using the Korn's inequality \eqref{Korn0}, it results that there is a positive constant $C$ such that
\begin{align}
({u},{u})
&\geq \dd c\,\int _\Omega\biggl(\|\nabla u\|^2+ \|\Curl(\sym \nabla u) \|^2\biggl)\, dv=c\,\|u\|^2_{{\mathcal{X}_0}}.
\end{align}
Therefore  our bilinear form $(\cdot,\cdot)$ is coercive.  The Cauchy-Schwarz inequality and the Poincar\'{e}-inequality  imply that the linear operator $l(\cdot)$ is bounded. By the Lax-Milgram theorem it follows that (\ref{wf}) has one and only one solution. The proof is complete in  the case ${\alpha_2> 0}$.

Now, we consider the case   $\boldsymbol{\alpha_2=0}$.  Using Lemma \ref{lemma30} it follows that the bilinear form $(\cdot,\cdot)$ is also coercive  for $\alpha_2=0$. Using similar estimates as above the existence follows also in this case and the proof is complete.
\end{proof}

\begin{remark} The Lax-Milgram theorem used in the proof of the previous theorem also offers a continuous dependence result on the load $f$.   Moreover, the weak solution $u$  minimizes  on ${\mathcal{X}_0}$ the  energy functional
\begin{align}
I(u)=\int_\Omega\bigg[& 2\,\mu\,\|\sym \nabla u\|^2+\lambda[ \tr(\nabla u)]^2\notag\\
&+\mu\,L_c^2\,[2\,  \alpha_1\, \|\dev\sym\Curl (\sym \nabla u)\|^2 +
2\, \alpha_2\, \|\skw\Curl (\sym \nabla u)\|^2]-\langle f,u\rangle \bigg]dv.\notag
\end{align}

\end{remark}
Let us consider $v\in C_0^\infty(\Omega;\Gamma)$ and $u$ a solution of  problem \eqref{ec1}--\eqref{bc100}. Then we obtain
\begin{align}
(u,v)=&\int_{ \Omega}  \langle f,v\rangle dv+\int_{\partial \Omega} \langle (\sigma+\sym\Curl (\widehat{ {m}})).\, n-\nabla[(\sym \widehat{M})\,T]: T,   v\rangle da\\
&+\,\int_{\partial  \Omega}\langle (\id-n\otimes n)(\sym \widehat{M}).n,(\id-n\otimes n)(\nabla v).n\rangle da+
\int_{\partial\Gamma}\langle ([\sym \widehat{M}]^+-[\sym \widehat{M}]^-).\, \nu,  v\rangle ds\notag\\
=&\int_{ \Omega}  \langle f,v\rangle dv+
\int_{\partial\Omega}  \langle (\sigma+\sym\Curl (\widehat{ {m}})).\, n,  v\rangle da
-\int_{\partial\Omega}\sum\limits_{i=1}^3 \langle \widehat{ {m}}_i\times n, (\sym \nabla v)_i\rangle da.\notag
\end{align}

Therefore, the corresponding existence results assures  that there exists the weak solution $u$  minimizing   on $C_0^\infty(\Omega;\Gamma)$ the  energy functional
\begin{align}
I(u)=\int_\Omega\bigg[& 2\,\mu\,\|\sym \nabla u\|^2+\lambda[ \tr(\nabla u)]^2\notag\\
&+\mu\,L_c^2\,[2\,  \alpha_1\, \|\dev\sym\Curl (\sym \nabla u)\|^2 +
2\, \alpha_2\, \|\skw\Curl (\sym \nabla u)\|^2]\bigg]dv-\int_{ \Omega}  \langle f,u\rangle dv\notag\\&-\int_{\partial \Omega} \langle (\sigma+\sym\Curl (\widehat{ {m}})).\, n-\nabla[(\sym \widehat{M})\,T]: T,  u \rangle da\\
&-\,\int_{\partial  \Omega}\langle (\id-n\otimes n)(\sym \widehat{M}).n,(\id-n\otimes n)(\nabla u).n\rangle da-
\int_{\partial\Gamma}\langle ([\sym \widehat{M}]^+-[\sym \widehat{M}]^-).\, \nu,u\rangle ds\notag\\
=\int_\Omega\bigg[& 2\,\mu\,\|\sym \nabla u\|^2+\lambda[ \tr(\nabla u)]^2\notag\\
&+\mu\,L_c^2\,[2\,  \alpha_1\, \|\dev\sym\Curl (\sym \nabla u)\|^2 +
2\, \alpha_2\, \|\skw\Curl (\sym \nabla u)\|^2]\bigg]dv-\int_{ \Omega}  \langle f,u\rangle dv\notag\\
&-\int_{\partial\Omega}  \langle (\sigma+\sym\Curl (\widehat{ {m}})).\, n,  u\rangle da
+\int_{\partial\Omega}\sum\limits_{i=1}^3 \langle \widehat{ {m}}_i\times n, (\sym \nabla v)_i\rangle da.\notag
\end{align}

\subsection{Traction boundary condition in the $\Curl(\sym \nabla u)$-formulation \\versus the $\nabla [\axl(\skew\nabla u)]$-formulation}\label{curlaxltractform}

In this section we compare the possible traction boundary conditions in  the $\nabla [\axl(\skew\nabla u)]$-formulation and the $\Curl (\sym \nabla  u)$-formulation. The conclusion is summarized in Figure \ref{limitmodel00} and Figure \ref{posbc}.
 Prescribing $\delta u$ and $(\id-n\otimes n).\curl u$ on the boundary means that we have prescribed independent geometrical boundary conditions, this is also the argumentation of Mindlin and Tiersten \cite{Mindlin62}, Koiter \cite{Koiter64}, Sokolowski \cite{Sokolowski72},  etc. However, the prescribed traction conditions remain not independent, in the sense that $\widetilde{g}$ leads to a further energetic conjugate, besides $\widetilde{t}$, of $u$. From this reason we claim that, in order to prescribe independent geometric boundary conditions and their corresponding completely independent energetic conjugate (traction boundary conditions), we have to prescribe $u$ and $(\id -n\otimes n)\nabla u.n$. In other words, we prescribe
\begin{align}
\int_{ \partial \Omega}\langle \widetilde{t}, u\rangle\, da+\int_{ \partial \Omega}\langle\widehat{g},   (\id -n\otimes n)\,\nabla  u.n \rangle da,
\end{align}
in which now $ u$ and $(\id -n\otimes n)\nabla  u.n$ are independent and $\widehat{g}$ does not produce work against $ u$, see \cite{MadeoGhibaNeffMunchKM} for further detailed explanations. This  type of independent boundary conditions are also correctly considered already by Bleustein \cite{bleustein1967note}, but for the full strain gradient elasticity case only.
In order to have a complete overview on the subject, in Table 1 we also summarize the equivalent form of the equilibrium equations.

\newpage

\begin{center}
{\small Table 1. Euler-Lagrange equations in various formulations \medskip}
\begin{tabular}{|l|l|}
\hline&\\
{\bf Euler-Lagrange equations } &{\bf Euler-Lagrange equations } \\
{\bf in direct tensor format} &{\bf in indices} \\ &\\
\hline&
\\\footnotesize{$\begin{array}{l}
\text{Euler-Lagrange equations for}\  \boldsymbol{\Curl (\sym \nabla  u)} \vspace{2mm}\\
{\rm Div}(\sigma+\widehat{\tau})+f=0\vspace{2mm}\\
\sigma=D_{\sym \nabla u} W_{\rm lin}(\sym \nabla u)\in {\rm Sym}(3)\vspace{2mm}\\
\widehat{\tau}=\sym (\Curl \widehat{m})\in {\rm Sym}(3)\vspace{2mm}\\
\widehat{m}=D_{\Curl(\sym \nabla u)} W_{\rm curv}(\Curl(\sym \nabla u)), \ \ \ \text{second order}
\end{array}$} & \footnotesize{$\begin{array}{l}
\text{Euler-Lagrange equations for}\  \boldsymbol{\Curl (\sym \nabla  u)} \vspace{2mm}\\
(\sigma_{ij}+\widehat{\tau}_{ij})_{,j}+f_i=0\vspace{2mm}\\
\sigma_{ij}=D_{\frac{1}{2}(u_{i,j}+u_{j,i})} W_{\rm lin}(\frac{1}{2}(u_{i,j}+u_{j,i}))\in {\rm Sym}(3)\vspace{2mm}\\
\widehat{\tau}_{ij}=\frac{1}{2}\left(\epsilon_{ilk}\widehat{m}_{jk,l}+\epsilon_{jlk}\widehat{m}_{ik,l}\right)\in {\rm Sym}(3)\vspace{2mm}\\
\widehat{m}=D_{\frac{1}{2}\epsilon_{ilk}\left(u_{j,kl}+u_{k,jl}\right)} W_{\rm curv}\left(\frac{1}{2}\epsilon_{ilk}\left(u_{j,kl}+u_{k,jl}\right)\right)
\end{array}$}  \\&\\
\hline&
 \\\footnotesize{
$\begin{array}{l}
\text{Euler-Lagrange equations for}\  \boldsymbol{\nabla [\axl\skw \nabla   u]} \vspace{2mm}\\
{\rm Div}(\sigma-\widetilde{\tau})+f=0\vspace{2mm}\\
\sigma=D_{\sym \nabla u} W_{\rm lin}(\sym \nabla u)\in {\rm Sym}(3)\vspace{2mm}\\
\widetilde{\tau}=\frac{1}{2}\anti({\rm Div}\,  \widetilde{m})\in \so(3)\vspace{2mm}\\
\widetilde{m}=D_{\nabla [\axl (\skw \nabla u)]} W_{\rm curv}(\nabla [\axl (\skw \nabla u)]), \ \text{second order}
\end{array}$} & \footnotesize{
$\begin{array}{l}
\text{Euler-Lagrange equations for}\  \boldsymbol{\nabla [\axl\skw \nabla   u] }\vspace{2mm}\\
(\sigma_{ij}-\widetilde{\tau}_{ij})_{,j}+f_i=0\vspace{2mm}\\
\sigma_{ij}=D_{\frac{1}{2}(u_{i,j}+u_{j,i})} W_{\rm lin}(\frac{1}{2}(u_{i,j}+u_{j,i}))\in {\rm Sym}(3)\vspace{2mm}\\
\widetilde{\tau}_{ij}=\frac{1}{2}\epsilon_{jik}  \widetilde{m}_{kl,l}\in \so(3)\vspace{2mm}\\
\widetilde{m}=D_{\left(\epsilon_{ijk}u_{k,j}\right)_{,m}} \widetilde{W}_{\mathrm{curv}}\left(\left(\epsilon_{ijk}u_{k,j}\right)_{,m}\right)
\end{array}$} \\&\\
\hline&
 \\\footnotesize{
$\begin{array}{l}
\text{Euler-Lagrange equations for} \ \boldsymbol{\nabla [\axl\skw \nabla   u]} \ \text{(3rd order)}\vspace{2mm}\\
{\rm Div}(\sigma-\widetilde{\tau})+f=0\vspace{2mm}\\
\sigma=D_{\sym \nabla u} W_{\rm lin}(\sym \nabla u)\in {\rm Sym}(3)\vspace{2mm}\\
\widetilde{\tau}=\Div \widetilde{\mathfrak{m}}=\frac{1}{2}\anti({\rm Div}\,  \widetilde{m})\in \so(3)\vspace{2mm}\\
\widetilde{\mathfrak{m}}=D_{\nabla\nabla u } [W_{\rm curv}(\nabla [\axl (\skw \nabla u)])], \qquad \quad\ \  \text{third order}
\end{array}$} & \footnotesize{
$\begin{array}{l}
\text{Euler-Lagrange equations for} \ \boldsymbol{\nabla [\axl\skw \nabla   u]} \ \text{(3rd order)}\vspace{2mm}\\
(\sigma_{ij}-\widetilde{\mathfrak{m}}_{ijk,k})_{,j}+f_i=0\vspace{2mm}\\
\sigma_{ij}=D_{\frac{1}{2}(u_{i,j}+u_{j,i})} W_{\rm lin}(\frac{1}{2}(u_{i,j}+u_{j,i}))\in {\rm Sym}(3)\vspace{2mm}\\
\widetilde{\tau}_{ij}=\widetilde{\mathfrak{m}}_{ijk,k}=\frac{1}{2}\epsilon_{jik}  \widetilde{m}_{kl,l}\in \so(3)\vspace{2mm}\\
\widetilde{\mathfrak{m}}_{ijk}=D_{u_{,ijk}}\widetilde{W}_{\mathrm{curv}}\left(\left(\epsilon_{ijk}u_{k,j}\right)_{,m}\right)
\end{array}$} \\&\\
\hline
\end{tabular}
\end{center}

\begin{figure}
\setlength{\unitlength}{1mm}
\begin{center}
\begin{picture}(10,46)
\thicklines

\put(10,21){\oval(125,52)}
\put(-48,35)
{
$
\left\{
\begin{array}{ll}
\delta u \quad \text{and} \quad \nabla \delta u\times n& \quad\text{cannot independently be prescribed}\vspace{2mm}\\
\delta u \quad \text{and} \quad \nabla \delta u.\, \tau& \quad\text{cannot independently be prescribed}\vspace{2mm}\\
\delta u \quad \text{and} \quad (\id-n\otimes n)\,\nabla \delta u& \quad\text{cannot independently be prescribed}\vspace{2mm}\\
\end{array}
\right.
$
}
\put(-48,20)
{
$
\ \ \,
\begin{array}{ll}
\delta u \quad \text{and} \quad \nabla \delta u.\, n& \text{can be  independently prescribed (\textbf{3 bc})}\vspace{2mm}\\
\delta u \quad \text{and} \quad (\id-n\otimes n)\,\nabla \delta u.\, n& \text{can be  independently prescribed  (\textbf{2 bc})}
\end{array}
$
}
\put(-48,10)
{
$\ \ \,
\begin{array}{ll}
\delta u \quad \text{and} \quad \langle \curl \delta u ,n\rangle& \qquad \quad\ \text{\, cannot independently be prescribed}
\end{array}
$
}
\put(-48,1)
{
$
\left\{
\begin{array}{ll}
\delta u \quad \text{and} \quad \langle \curl \delta u ,\tau\rangle& \ \text{can be independently prescribed  (\textbf{2 bc})}\vspace{2mm}\\
\delta u \quad \text{and} \quad (\id-n\otimes n). \curl  \delta u& \ \text{can be  independently prescribed  (\textbf{2 bc})}
\end{array}
\right.
$
}
\end{picture}
\end{center}
\caption{The possible independent geometrically boundary conditions. The joining  bracket means that the conditions are equivalent.}\label{posbc}
\end{figure}

We outline that   there exists a relation between the allowed traction boundary conditions in the $\Curl(\sym \nabla u)$-formulation and those from the  $\nabla [\axl(\skew\nabla u)]$-formulation which we take from \cite{MadeoGhibaNeffMunchKM}
\newpage
 \begin{align}
\underbrace{-\int_{\partial \Omega}\langle (\sigma-\widetilde{\tau}).\, n, \delta u\rangle \,da
-\int_{\partial\Omega}\langle \widetilde{ {m}}.\, n, \axl(\skw \nabla \delta u) \rangle\, da}_{\text{axl-formulation}}\hspace{2cm}\\
=
\underbrace{-\int_{\partial\Omega}\langle(\sigma+\widehat{\tau}).\, n,\delta u \rangle da-\int_{\partial\Omega}\sum\limits_{i=1}^3 \langle \widehat{ {m}}_i\times n, (\sym \nabla \delta u)_i\rangle\, da}_{\text{Curl-formulation}}.\notag
\end{align}
or after splitting up with use of the surface divergence theorem on both sides
\begin{align}\label{germaneq3n01}
&\underbrace{-\int_{\partial \Omega} \underbrace{ \langle (\sigma+\sym\Curl (\widehat{ {m}})).\, n-\nabla[(\sym \widehat{M})\,(\id-n\otimes n)]: (\id-n\otimes n)}_{(1)},  \delta u\rangle da}_{(a)}\\
&-\underbrace{\int_{\partial  \Omega}\langle \underbrace{(\id-n\otimes n)(\sym \widehat{M}).n}_{(2)},\nabla\delta u.n\rangle  da}_{(b)}-
\underbrace{\int_{\partial\Gamma}\langle \underbrace{([\sym \widehat{M}]^+-[\sym \widehat{M}]^-).\,\nu}_{(3)}, \delta u\rangle ds}_{(c)}\notag\\
=
&\underbrace{-\int_{\partial \Omega}\langle \underbrace{(\sigma-\frac{1}{2}\anti
{\rm Div}[\widetilde{ {m}}]).\, n-\frac{1}{2}\nabla[(\anti(\widetilde{ {m}}.\, n))\,(\id-n\otimes n)]: (\id-n\otimes n)}_{(1^\prime)}, \delta u\rangle \,da}_{(a^\prime)}\notag\\
&-\underbrace{
\int_{\partial \Omega}\langle \underbrace{\frac{1}{2}(\id -n\otimes n)\anti(\widetilde{ {m}}.\, n).n}_{(2^\prime)},  \nabla \delta u.n \rangle da}_{(b^\prime)}-\underbrace{
\int_{\partial \Gamma}\langle \underbrace{\frac{1}{2}([\anti(\widetilde{ {m}}.\, n)]^+-[\anti(\widetilde{ {m}}.\, n)]^-).\, \nu}_{(3^\prime)}, \delta u\rangle ds}_{(c^\prime)},\notag
\end{align}
for all variations $ \delta u\in C^\infty(\Omega)$. Naively, we might expect that the quantities involved have to be equal term by term, i.e. $(1)=(1^\prime), (2)=(2^\prime),(3)=(3^\prime)$ or
$(a)=(a^\prime), (b)=(b^\prime),(c)=(c^\prime)$. However, this is not true, see Appendix \ref{notthesamebc}.

\subsection{Principle of virtual work in the indeterminate couple stress model}
\subsubsection{Principle of virtual work in Cosserat theory}
Let us first recall that in the Cosserat theory with independent fields of displacement and microrotations the internal energy has the form $W(\nabla u, \overline{A},\nabla \overline{A})$ in which $u:\Omega\rightarrow\mathbb{R}^3$ is the displacement and $\overline{A}:\Omega \rightarrow\so(3)$ is the infinitesimal  microrotation. The virtual work principle of the Cosserat theory is given by
\begin{align}
\mathcal{P}^{\rm int}=\mathcal{P}^{\rm ext},
\end{align}
where
\begin{align}
\mathcal{P}^{\rm int}&=\frac{\rm d}{\rm dt}\int_\Omega W((\nabla u+t\nabla \delta u), \overline{A}+t\delta \overline{A},\nabla (\overline{A}+t\delta \overline{A}))\Big|_{t=0},\notag\\
\mathcal{P}^{\rm ext}&=\int_\Omega\langle \underbrace{f}_{\text{body force}}, u\rangle\, dv+\int_\Omega\langle \underbrace{\axl(\mathfrak{M})}_{\text{body couple}}, \axl(\overline{A})\rangle\, dv\notag\\
&\quad+\int_{\partial \Omega\setminus \overline{\Gamma}}\langle \underbrace{t}_{\text{surface tractions}}, u\rangle\, da+\int_{\partial \Omega\setminus \overline{\Gamma}}\langle\underbrace{\axl(\mathfrak{G})}_{\text{surface couple}}, \axl(\overline{A}) \rangle\, da,
\end{align}
with
$f:\partial \Omega\setminus \overline{\Gamma}\subseteq \mathbb{R}^3\rightarrow\mathbb{R}^3$, $t:\partial \Omega\setminus \overline{\Gamma}\subseteq \mathbb{R}^3\rightarrow\mathbb{R}^3$, $\mathfrak{M}:\Omega \subseteq\mathbb{R}^3 \rightarrow\so(3)$, $\mathfrak{G}:\overline{\Omega} \subseteq\mathbb{R}^3 \rightarrow\so(3)$.
 From this virtual work principle, one obtains the equilibrium equations
\begin{align}
\Div \sigma+f&=0,\\
\Div \mathfrak{m}+\axl(\skew\, \sigma)+\axl(\mathfrak{M})&=0,\notag
\end{align}
where $\sigma=D_{\sym \nabla u}  W(\sym \nabla u, \overline{A})$ and $\mathfrak{m}=D_{\nabla \overline{A}}  W( \nabla u, \overline{A}, \nabla \overline{A})$, and the boundary conditions
\begin{align}
\sigma.\, n={t},\qquad  \mathfrak{m}.\, n=\axl(\mathfrak{G}).
\end{align}
In order to obtain these equilibrium equations and the form of the boundary conditions, we have used the fact that $u$ and $\overline{A}$ are independent constitutive variables.

\subsubsection{Virtual work principle in the $\Curl(\sym (\nabla u))$-formulation of the indeterminate couple stress model}
Using again the fact that  $u$ and $2\,\curl u=\axl(\skew \nabla u)\thicksim \overline{A}$ are not independent constitutive variables, and the identity
$$
\nabla[\axl (\skw \nabla u)]=[\Curl(\sym \nabla u)]^T
$$
we consider the energy  $W(\sym (\nabla u),\Curl(\sym (\nabla u)) )$ and the following new form of the  virtual work principle
\begin{align}
\mathcal{P}^{\rm int}=\mathcal{P}^{\rm ext},
\end{align}
where
\begin{align}
\mathcal{P}^{\rm int}&=\frac{\rm d}{\rm dt}\int_\Omega W(\sym (\nabla u+t\nabla \delta u),\Curl(\sym (\nabla u+t\nabla \delta u)) )\Big|_{t=0},\notag\\
\mathcal{P}^{\rm ext}&
=\int_\Omega\langle \underbrace{f}_{\text{body forces}}, u\rangle dv+\int_\Omega\langle \underbrace{\mathcal{M}}_{\text{body couple}},\sym \nabla u \rangle dv\\
&\quad+\int_{\partial \Omega\setminus \overline{\Gamma}}\langle \underbrace{t}_{\text{surface traction}}, u\rangle da+\int_{\partial \Omega\setminus \overline{\Gamma}}\langle \underbrace{g}_{\text{surface double tractions}},(\id-n\otimes n)\nabla u.n \rangle da+\oint_{\partial\Gamma} \langle \underbrace{\pi}_{\text{edge line force}}, \nu\rangle ds,\notag
\end{align}
with $f:\Omega\subseteq \mathbb{R}^3\rightarrow\mathbb{R}^3$, ${t}:{\partial \Omega\setminus \overline{\Gamma}} \subseteq \mathbb{R}^3\rightarrow\mathbb{R}^3$, $\mathcal{M}:\Omega \subseteq\mathbb{R}^3 \rightarrow\so(3)$, ${g}:{\partial \Omega\setminus \overline{\Gamma}} \subseteq \mathbb{R}^3 \rightarrow\mathbb{R}^3$ and $\pi:\partial \Gamma \subseteq \mathbb{R}^3 \rightarrow\mathbb{R}^3$.

 From this virtual work principle, we obtain the equilibrium equations
\begin{align}\label{ec11}
 \Div \,(\sigma+\widehat{\tau}+\mathcal{M})+\underbrace{\Div \,\mathcal{M}+ f}_{\text{the total body force}}=0,
 \end{align}
where
\begin{align}
\sigma&=D_{\sym \nabla u}  W(\sym \nabla u, \Curl(\sym \nabla u)),\notag\\
\widehat{\tau}&=\sym[\Curl \widehat{m}]\in {\rm Sym}(3),\\
\widehat{m}&=D_{\Curl(\sym \nabla u)}  W(\sym \nabla u,\Curl(\sym \nabla u))\notag
\end{align}
and the following traction boundary conditions
  \begin{align}
  \left.
         \begin{array}{rcl}
 [(\sigma+\widehat{\tau}).n-\nabla[(\sym \widehat{M})\,(\id-n\otimes n)]: (\id-n\otimes n)]\,(x)&=& \underbrace{t(x)-\mathcal{M}.n,}_{\text{the total traction condition}}\vspace{1.2mm}\\
  \dd[(\id-n\otimes n)(\sym \widehat{M}).n]\,(x)&=&[(\id-n\otimes n)\,g]\,(x),
 \end{array}
 \right\}
 \quad
  x\in \partial \Omega\setminus {\overline{\Gamma}}
  \quad
 \begin{array}{r}
(3\  \text{bc})\vspace{1.2mm}\\
(2\  \text{bc})
 \end{array}\\
 \begin{array}{rcl}
\{([\sym \widehat{M}]^+-[\sym \widehat{M}]^-).\,\nu\}\,(x)&=&\pi,\hspace{3.55cm}
 \end{array}
 \quad
  x\in \partial {\Gamma}
  \quad \quad\
 \begin{array}{r}
(3\  \text{bc})
 \end{array}\notag
 \end{align}
  where  $\Gamma$ is an arbitrary open subset of $\partial \Omega$.

In Subsection \ref{curlaxltractform} and Appendix \ref{notthesamebc} we  show  that the traction boundary conditions in the $\Curl(\sym (\nabla u))$-formulation do not coincide pointwise with those arising from $\nabla[\axl (\skw \nabla u)]$-formulation, even if the boundary virtual power works are identical.

\begin{figure}
\setlength{\unitlength}{1mm}
\begin{center}
\begin{picture}(10,95)
\thicklines

\put(5,85){\oval(165,50)}
\put(-72,105){\bf\footnotesize{Incomplete-standard boundary conditions in the $\nabla [\axl (\skw \nabla u)]$--formulation \cite{Mindlin62}}}
\put(-72,100){\footnotesize{Geometric (essential) boundary conditions \ (3+2) [correct]}}
 \put(-72,95){\footnotesize{$u\big|_{\Gamma}=\widetilde{u}^0\in\mathbb{R}^{ 3}, \qquad (\id-n\otimes n).\curl u\big|_{\Gamma}=(\id-n\otimes n).\curl  \widetilde{u}^0\in\mathbb{R}^{3},\quad \text{or}\ \ (\id-n\otimes n)\nabla  u.n\big|_{\Gamma}=(\id-n\otimes n)\,\nabla  \widetilde{u}^0.n\in\mathbb{R}^{3}$}}
 \put(-72,90){\footnotesize{Mechanical (traction) boundary conditions \ (3+2) \textbf{[erroneous]}}}
 \put(-72,85){\footnotesize{$\left((\sigma-\widetilde{\tau}).\, n
 -\frac{1}{2} n\times \nabla[\langle n,(\sym \widetilde{ {m}}).n\rangle]\right)\big|_{\partial \Omega\setminus\overline{\Gamma}}=\widetilde{t}$, \qquad\qquad\qquad\qquad $\widetilde{\tau}=\Div \widetilde{\mathfrak{m}}=\frac{1}{2}\anti({\rm Div}\,  \widetilde{m})\in \so(3)$}}
 \put(-43,80){\footnotesize{$(\id -n\otimes n)\,\widetilde{ {m}}.\, n\big|_{\partial \Omega\setminus\overline{\Gamma}}=(\id-n\otimes n)\,\widetilde{h}$}}
\put(78,85){\footnotesize{3 bc}}
\put(78,80){\footnotesize{2 bc}}
\put(-72,75){\footnotesize{Boundary virtual work }}
\put(-72,70){\footnotesize{$-\int_{\partial \Omega}\langle (\sigma-\widetilde{\tau}).\, n, \delta u\rangle \,da
-\int_{\partial\Omega}\langle \widetilde{ {m}}.\, n, \axl (\skw \nabla \delta u) \rangle\, da=
0\qquad \Leftrightarrow$}}
\put(-72,65){\footnotesize{$-\int_{\partial \Omega}\langle \Big\{(\sigma-\widetilde{\tau}).\, n -\frac{1}{2} n\times \underbrace{\nabla[\langle n,(\sym \widetilde{ {m}}).n\rangle]}_{\text{normal curvature}}\Big\}, \delta u\rangle \,da
-\int_{\partial\Omega}\langle \widetilde{ {m}}.n,\Big\{(\id-n\otimes n)\,[\axl (\skw \nabla \delta u)]\Big\}\rangle\,
 da=0$}}

\put(5,52){\Huge $\Updownarrow$}

\put(5,23){\oval(165,50)}
\put(-72,45){\bf\footnotesize{Incomplete-standard boundary conditions in the $\nabla [\axl (\skw \nabla u)]$--formulation, written in indices \cite{lubarda2003effects}}}
\put(-72,40){\footnotesize{Geometric (essential) boundary conditions \ (3+2) [correct]}}
 \put(-72,35){\footnotesize{$u_i\big|_{\Gamma}=\widetilde{u}^0_{,i}\in\mathbb{R}^{ 3}, \qquad \ \left(\epsilon_{ikl}u_{l,k}-\epsilon_{jkl}u_{l,k}n_jn_i\right)\big|_{\Gamma}= \epsilon_{ikl}\widetilde{u}^0_{l,k}-\epsilon_{jkl}\widetilde{u}^0_{l,k}n_jn_i$}}
 \put(-71.7,30){\footnotesize{$\quad \ \ \qquad \qquad \qquad \quad \text{or}\ \ \left(u_{i,k}n_k-u_{j,k}n_kn_jn_i\right)\big|_{\Gamma}=\widetilde{u}^0_{i,k}n_k-\widetilde{u}^0_{j,k}n_kn_jn_i$}}
 \put(-72,25){\footnotesize{Mechanical (traction) boundary conditions \ (3+2) \textbf{[erroneous]}}}
 \put(-72,20){\footnotesize{$\left((\sigma_{ij}-\widetilde{\tau}_{ij})\, n_j
 -\frac{1}{2} \epsilon_{ikl}n_k (\widetilde{m}_{ij}n_in_j)_{,l}\right)\big|_{\partial \Omega\setminus\overline{\Gamma}}=\widetilde{t}_i$, \qquad\qquad\qquad\qquad $\widetilde{\tau}_{ij}=\frac{1}{2}\epsilon_{jik}  \widetilde{m}_{kl,l}\in \so(3)$}}
 \put(-52.5,15){\footnotesize{$\left(\widetilde{ {m}}_{ik}n_k-\widetilde{ {m}}_{jk}n_kn_jn_i\right)\big|_{\partial \Omega\setminus\overline{\Gamma}}=\widetilde{h}_i-\widetilde{h}_jn_jn_i$}}
\put(78,20){\footnotesize{3 bc}}
\put(78,15){\footnotesize{2 bc}}
\put(-72,10){\footnotesize{Boundary virtual work }}
\put(-72,5){\footnotesize{$-\int_{\partial \Omega}\left((\sigma_{ij}-\widetilde{\tau}_{ij})\, n_j
 \right) \delta u_i \,da
-\int_{\partial\Omega}\langle \widetilde{ {m}}.\, n, \axl (\skw \nabla \delta u) \rangle\, da=
0\qquad \Leftrightarrow$}}
\put(-72,0){\footnotesize{$-\int_{\partial \Omega}\left((\sigma_{ij}-\widetilde{\tau}_{ij})\, n_j
 -\frac{1}{2} \epsilon_{ikl}n_k (\widetilde{m}_{ij}n_in_j)_{,l}\right) \delta u_i \,da
-\frac{1}{2}\,\int_{\partial\Omega}\left(\widetilde{ {m}}_{ik}n_k-\widetilde{ {m}}_{jk}n_kn_jn_i\right)\,\left(\epsilon_{ikl}\delta u_{l,k}-\epsilon_{jkl}\delta u_{l,k}n_jn_i\right)\,
 da=0$}}

\end{picture}
\end{center}
\caption{The incompletely  boundary conditions in the $\nabla [\axl (\skw \nabla u)]$--formulation which have been employed hitherto by all authors to our knowledge. The virtual displacement is denoted by $\delta u\in C^\infty(
\Omega, \mathbb{R}^3)$. The number of traction boundary conditions is correct, but the split into independent variations at the boundary is incomplete, as shown in \cite{MadeoGhibaNeffMunchKM}.}\label{limitmodel003}
\end{figure}

\begin{figure}
\setlength{\unitlength}{1mm}
\begin{center}
\begin{picture}(10,100)
\thicklines

\put(5,87){\oval(165,62)}
\put(-72,115){\bf\footnotesize{Boundary conditions in the $\nabla [\axl (\skw \nabla u)]$--formulation in terms of gradient elasticity}}
\put(-72,110){\bf\footnotesize{ and third order moment tensors, see \cite{MadeoGhibaNeffMunchKM}}}
\put(-72,105){\footnotesize{Geometric (essential) boundary conditions \ (3+2)}}
 \put(-72,100){\footnotesize{$u\big|_{\Gamma}=\widetilde{u}^0\in\mathbb{R}^{ 3}, \qquad (\id-n\otimes n)\nabla  u.n\big|_{\Gamma}=(\id-n\otimes n)\,\nabla  \widetilde{u}^0.n\in\mathbb{R}^{3},\quad \text{or}\ \ (\id-n\otimes n).\curl u\big|_{\Gamma}=(\id-n\otimes n).\curl  \widetilde{u}^0\in\mathbb{R}^{3}$}}
 \put(-72,95){\footnotesize{Mechanical (traction) boundary conditions \ (3+2)}}
 \put(-67.5,90){\footnotesize{$\left((\sigma-\Div \widetilde{\mathfrak{m}}).\, n-\nabla [(\widetilde{\mathfrak{m}}.\, n)\,(\id-n\otimes n)]: (\id-n\otimes n)\right)\big|_{\partial \Omega\setminus\overline{\Gamma}}=\widetilde{t}$,\qquad \quad  $\widetilde{\mathfrak{m}}=D_{\nabla\nabla u}[\widetilde{W}_{\mathrm{curv}}(\nabla [\axl (\skw \nabla u)])]$}}
 \put(-22.2,85){\footnotesize{$(\id -n\otimes n)[\widetilde{\mathfrak{m}}.\, n].n\big|_{\partial \Omega\setminus\overline{\Gamma}}=(\id-n\otimes n)\,\widetilde{h}$}}
\put(-21.3,80){\footnotesize{$ ([\widetilde{\mathfrak{m}}.\, n]^+-[\widetilde{\mathfrak{m}}.\, n]^-).  \, \nu\big|_{\partial \Gamma}=\widetilde{\pi}$\ \ ``edge line force" on $\partial \Gamma$}}
\put(78,90){\footnotesize{3 bc}}
\put(78,85){\footnotesize{2 bc}}
\put(78,80){\footnotesize{3 bc}}
\put(-72,75){\footnotesize{Boundary virtual work}}
\put(-72,70){\footnotesize{$-\int_{\partial \Omega}\langle (\sigma-\Div \widetilde{\mathfrak{m}}).\, n, \delta u\rangle \,da
-\int_{\partial\Omega}\langle \widetilde{\mathfrak{m}}.\, n, \nabla \delta u \rangle\, da=
0\qquad \Leftrightarrow$}}
\put(-72,65){\footnotesize{$ -\int_{\partial \Omega}\langle (\sigma-\Div \widetilde{\mathfrak{m}}).\, n-\nabla [(\widetilde{\mathfrak{m}}.\, n)\,(\id-n\otimes n)]:(\id-n\otimes n), \delta u\rangle \,da-\int_{\partial \Omega}\langle (\id -n\otimes n)[\widetilde{\mathfrak{m}}.\, n].n,  \nabla \delta u.n \rangle da$}}
\put(-72,60){\footnotesize{$ -
\int_{\partial \Gamma}\langle ([\widetilde{\mathfrak{m}}.\, n]^+-[\widetilde{\mathfrak{m}}.\, n]^-).  \, \nu, \delta u\rangle ds=0$}}

\put(5,50){{\huge $\Updownarrow$} \  \footnotesize{equivalent}}

\put(5,21){\oval(165,50)}
 \put(-72,40){\footnotesize{\bf Boundary conditions  (3+2) in the $\nabla [\axl (\skw \nabla u)]$--formulation in terms of gradient elasticity, third order}}
 \put(-72,35){\bf\footnotesize{  moment tensors,  and written in indices, see \cite{MadeoGhibaNeffMunchKM}}}
 \put(-72,30){\footnotesize{Geometric (essential) boundary conditions \ (3+2)}}
 \put(-72,25){\footnotesize{$u_i\big|_{\Gamma}=\widetilde{u}^0_{ i}, \qquad (u_{i,k}n_k-u_{j,k}n_jn_k n_i)\big|_{\Gamma}=\widetilde{u}^0_{i,k}n_k-\widetilde{u}^0_{j,k}n_in_jn_k  $,}}
  \put(-79,20){\footnotesize{$\hspace{2cm} \text{or}\ \ (\epsilon_{ikl}u_{l,k}-\epsilon_{jkl} u_{l,k}n_jn_i)\big|_{\Gamma}=(\epsilon_{ikl}\widetilde{u}^0_{l,k}-\epsilon_{jkl} \widetilde{u}^0_{l,k}n_jn_i) $,}}
 \put(-72,15){\footnotesize{Mechanical (traction) boundary conditions \ (3+2)}}
 \put(-72,10){\footnotesize{$\left[\left(\sigma_{ij}-\widetilde{\mathfrak{m}}_{ijk,k}\right)\, n_{j}-\left(\widetilde{\mathfrak{m}}_{ipk}\,n_k-\widetilde{\mathfrak{m}}_{ijk}\, n_{k}n_jn_p\right)_{,h}\left(\delta_{ph}-n_p n_h\right)\right]\big|_{\partial \Omega\setminus\overline{\Gamma}}=\widetilde{t}_i$,\quad  $\widetilde{\mathfrak{m}}_{ijk}=D_{u_{,ijk}}\widetilde{W}_{\mathrm{curv}}\left(\left(\epsilon_{ijk}u_{k,j}\right)_{,m}\right)$}}
 \put(-25.5,5){\footnotesize{$\left(\widetilde{\mathfrak{m}}_{ijp}\, n_{j}-\widetilde{\mathfrak{m}}_{pjk}\, n_{k}n_{j}n_i\right)\,n_p\big|_{\partial \Omega\setminus\overline{\Gamma}}=\widetilde{h}_i-\widetilde{h}_pn_pn_i$}}
\put(-24,0){\footnotesize{$ ([\widetilde{\mathfrak{m}}_{pjk}\, n_{k}]^+-[\widetilde{\mathfrak{m}}_{pjk}\, n_{k}]^-).  \, \nu_j\big|_{\partial \Gamma}=\widetilde{\pi}_p$\ \ ``edge line force" on $\partial \Gamma$}}
\put(78,10){\footnotesize{3 bc}}
\put(78,5){\footnotesize{2 bc}}
\put(78,0){\footnotesize{3 bc}}
\end{picture}
\end{center}
\caption{The complete-standard boundary conditions in the $\nabla [\axl (\skw \nabla u)]$--formulation in terms of a third order couple stress tensor coming from full gradient elasticity.  The virtual displacement is denoted by $\delta u\in C^\infty(
\Omega, \mathbb{R}^3)$.  The summation convention was used in index notations. }\label{limitmodel03order}
\end{figure}

\newpage

\begin{figure}
\setlength{\unitlength}{1mm}
\begin{center}
\begin{picture}(10,170)
\thicklines

\put(5,156){\oval(165,69)}
\put(-72,185){\bf \footnotesize{Correct boundary conditions in the $\boldsymbol{\nabla[\axl(\skew\nabla u)]}$-formulation ($\nabla (\curl u)$-formulation), see \cite{MadeoGhibaNeffMunchKM}}}
\put(-72,180){\footnotesize{Geometric (essential) boundary conditions \ (3+2)}}
 \put(-72,175){\footnotesize{$u\big|_{\Gamma}=\widetilde{u}^0\in\mathbb{R}^{ 3}, \qquad (\id-n\otimes n)\nabla  u.n\big|_{\Gamma}=(\id-n\otimes n)\,\nabla  \widetilde{u}^0.n\in\mathbb{R}^{3},\quad \text{or}\ \ (\id-n\otimes n).\curl u\big|_{\Gamma}=(\id-n\otimes n).\curl  \widetilde{u}^0\in\mathbb{R}^{3}$}}
 \put(-72,170){\footnotesize{Mechanical (traction) boundary conditions \ (3+2)}}
 \put(-75,165){\footnotesize{$\left((\sigma-\widetilde{\tau}).\, n
 -\frac{1}{2} n\times \nabla[\langle n,(\sym \widetilde{ {m}}).n\rangle]-\frac{1}{2}\nabla[(\anti[(\id-n\otimes n)\widetilde{ {m}}.\, n])(\id-n\otimes n)]: (\id-n\otimes n)\right)\big|_{\partial \Omega\setminus\overline{\Gamma}}=\widetilde{t}$,}}
 \put(2,160){\footnotesize{$(\id -n\otimes n)\anti[(\id-n\otimes n)\widetilde{ {m}}.\, n].n\big|_{\partial \Omega\setminus\overline{\Gamma}}=(\id-n\otimes n)\,\widetilde{h}$}}
\put(-24,155){\footnotesize{``edge line force" on $\partial \Gamma$: $\langle ([\anti[\widetilde{ {m}}.\, n]]^+-[\anti[\widetilde{ {m}}.\, n]]^-).  \, \nu\big|_{\partial \Gamma}=\widetilde{\pi}$}}
\put(81.5,165){\footnotesize{3 bc}}
\put(81.5,160){\footnotesize{2 bc}}
\put(81.5,155){\footnotesize{3 bc}}
\put(-72,150){\footnotesize{Boundary virtual work}}
\put(-72,145){\footnotesize{$-\int_{\partial \Omega}\langle (\sigma-\widetilde{\tau}).\, n, \delta u\rangle \,da
-\int_{\partial\Omega}\langle \widetilde{ {m}}.\, n, \axl (\skw \nabla \delta u) \rangle\, da=
0\qquad \Leftrightarrow$}}
\put(-72,140){\footnotesize{$-\int_{\partial \Omega}\langle \Big\{(\sigma-\widetilde{\tau}).\, n -\frac{1}{2} n\times \underbrace{\nabla[\langle n,(\sym \widetilde{ {m}}).n\rangle]}_{\text{normal curvature}}\Big\}, \delta u\rangle \,da
+\int_{\partial\Omega}\langle \widetilde{ {m}}.n,\Big\{(\id-n\otimes n)\,[\axl (\skw \nabla \delta u)]\Big\}\rangle\,
 da=0$}}

 \put(-72,130){\footnotesize{$ -\int_{\partial \Omega}\langle (\sigma-\widetilde{\tau}).\, n
 -\frac{1}{2} n\times \nabla[\langle n,(\sym \widetilde{ {m}}).n\rangle]-\frac{1}{2}\nabla[(\anti[(\id-n\otimes n)\widetilde{ {m}}.\, n])\,(\id-n\otimes n)]: (\id-n\otimes n), \delta u\rangle \,da\notag$}}

\put(-72,125){\footnotesize{$ \qquad-\frac{1}{2}
\int_{\partial \Omega}\langle (\id -n\otimes n)\anti[(\id-n\otimes n)\widetilde{ {m}}.\, n].n,  \nabla \delta u.n \rangle da-\frac{1}{2}
\int_{\partial\Gamma}\langle ([\anti[\widetilde{ {m}}.\, n]]^+-[\anti[\widetilde{ {m}}.\, n]]^-).  \, \nu, \delta u\rangle ds\notag =0$}}

\put(5,115){{\huge $\Updownarrow$} \  \footnotesize{equivalent}}

\put(5,82){\oval(165,60)}
\put(-72,105){\bf \footnotesize{Alternative equivalent correct boundary conditions in the $\boldsymbol{\nabla[\axl(\skew\nabla u)]}$-formulation,  see \cite{MadeoGhibaNeffMunchKM}}}
\put(-72,100){\footnotesize{Geometric (essential) boundary conditions \ (3+2)}}
 \put(-72,95){\footnotesize{$u\big|_{\Gamma}=\widetilde{u}^0\in\mathbb{R}^{ 3}, \qquad (\id-n\otimes n)\nabla  u.n\big|_{\Gamma}=(\id-n\otimes n)\,\nabla  \widetilde{u}^0.n\in\mathbb{R}^{3},\quad \text{or}\ \ (\id-n\otimes n).\curl u\big|_{\Gamma}=(\id-n\otimes n).\curl  \widetilde{u}^0\in\mathbb{R}^{3}$}}
 \put(-72,90){\footnotesize{Mechanical (traction) boundary conditions \ (3+2)}}
 \put(-72,85){\footnotesize{$\left((\sigma-\widetilde{\tau}).\, n-\frac{1}{2}\nabla[(\anti(\widetilde{ {m}}.\, n))\,(\id-n\otimes n)]: (\id-n\otimes n)\right)\big|_{\partial \Omega\setminus\overline{\Gamma}}=\widetilde{t}$,}}
 \put(-28,80){\footnotesize{$(\id -n\otimes n)\anti[\widetilde{ {m}}.\, n].n\big|_{\partial \Omega\setminus\overline{\Gamma}}=(\id-n\otimes n)\,\widetilde{h}$}}
\put(-35.5,75){\footnotesize{$ ([\anti[\widetilde{ {m}}.\, n]]^+-[\anti[\widetilde{ {m}}.\, n]]^-).  \, \nu\big|_{\partial \Gamma}=\widetilde{\pi}$\ \ ``edge line force" on $\partial \Gamma$}}
\put(78,85){\footnotesize{3 bc}}
\put(78,80){\footnotesize{2 bc}}
\put(78,75){\footnotesize{3 bc}}
\put(-72,70){\footnotesize{Boundary virtual work}}
\put(-72,65){\footnotesize{$-\int_{\partial \Omega}\langle (\sigma-\widetilde{\tau}).\, n, \delta u\rangle \,da
-\int_{\partial\Omega}\langle \widetilde{ {m}}.\, n, \axl (\skw \nabla \delta u) \rangle\, da=
0\qquad \Leftrightarrow$}}
\put(-72,60){\footnotesize{$ -\int_{\partial \Omega}\langle (\sigma-\widetilde{\tau}).\, n-\frac{1}{2}\nabla[(\anti(\widetilde{ {m}}.\, n))\,(\id-n\otimes n)]: (\id-n\otimes n), \delta u\rangle \,da$}}
\put(-72,55){\footnotesize{$ \qquad-\frac{1}{2}
\int_{\partial \Omega}\langle (\id -n\otimes n)\anti(\widetilde{ {m}}.\, n).n,  \nabla \delta u.n \rangle da-\frac{1}{2}
\int_{\partial \Gamma}\langle ([\anti[\widetilde{ {m}}.\, n]]^+-[\anti[\widetilde{ {m}}.\, n]]^-).  \, \nu, \delta u\rangle ds=0$}}

\put(5,45){{\huge $\Updownarrow$} \  \footnotesize{equivalent}}

\put(5,18){\oval(165,45)}
\put(-72,35){\bf \footnotesize{Alternative equivalent correct boundary conditions in the $\boldsymbol{\nabla[\axl(\skew\nabla u)]}$-formulation,  written in indices}}
\put(-72,30){\footnotesize{Geometric (essential) boundary conditions \ (3+2)}}
 \put(-72,25){\footnotesize{$u_i\big|_{\Gamma}=\widetilde{u}^0_{,i}\in\mathbb{R}^{ 3}, \qquad \ \left(\epsilon_{ikl}u_{l,k}-\epsilon_{jkl}u_{l,k}n_jn_i\right)\big|_{\Gamma}= \epsilon_{ikl}\widetilde{u}^0_{l,k}-\epsilon_{jkl}\widetilde{u}^0_{l,k}n_jn_i$}}
 \put(-71.7,20){\footnotesize{$\quad \ \ \qquad \qquad \qquad \quad \text{or}\ \ \left(u_{i,k}n_k-u_{j,k}n_kn_jn_i\right)\big|_{\Gamma}=\widetilde{u}^0_{i,k}n_k-\widetilde{u}^0_{j,k}n_kn_jn_i$}}
 \put(-72,15){\footnotesize{Mechanical (traction) boundary conditions \ (3+2)}}
 \put(-72,10){\footnotesize{$\left((\sigma_{ij}-\widetilde{\tau}_{ij})\, n_j+\frac{1}{2}(\epsilon_{ihk}\widetilde{m}_{ks}n_s-\epsilon_{ijk}\widetilde{m}_{ks}n_sn_jn_h)_{,p}(\delta_{hp}-n_hn_p)\right)\big|_{\partial \Omega\setminus\overline{\Gamma}}=\widetilde{t}_i$,}}
 \put(-52,5){\footnotesize{$\hspace{2.23cm}(\epsilon_{ipk}\widetilde{m}_{ks}n_s-\epsilon_{jpk}\widetilde{m}_{ks}n_sn_jn_i)\,n_p\big|_{\partial \Omega\setminus\overline{\Gamma}}=\widetilde{h}_i-\widetilde{h}_pn_pn_i,$}}
\put(-59.5,0){\footnotesize{$ \hspace{3.23cm}([\epsilon_{ipk}\widetilde{m}_{ks}n_s]^+-[\epsilon_{ipk}\widetilde{m}_{ks}n_s]^-)  \, \nu_p\big|_{\partial \Gamma}=\widetilde{\pi}_i$\ \ ``edge line force" on $\partial \Gamma$}}
\put(78,10){\footnotesize{3 bc}}
\put(78,5){\footnotesize{2 bc}}
\put(78,0){\footnotesize{3 bc}}
\end{picture}
\end{center}
\caption{The possible boundary conditions in the $\nabla [\axl (\skw \nabla u)]$ and $\Curl(\sym \nabla u)$--formulation. The equivalence of the geometric boundary condition is clear. The virtual displacement is denoted by $\delta u\in C^\infty(
\Omega, \mathbb{R}^3)$. }\label{limitmodel00}
\end{figure}

\newpage

\begin{figure}
\setlength{\unitlength}{1mm}
\begin{center}
\begin{picture}(10,100)
\thicklines

\put(5,82){\oval(165,59)}
\put(-72,107){\bf\footnotesize{Correct boundary conditions in the $\boldsymbol{\Curl(\sym\nabla u)}$--formulation}}
\put(-72,102){\footnotesize{Geometric (essential) boundary conditions \ (3+2)}}
 \put(-72,97){\footnotesize{$u\big|_{\Gamma}=\widehat{u}^0 \in \mathbb{R}^3, \qquad (\id-n\otimes n)\nabla u.n\big|_{\Gamma}=(\id-n\otimes n)\,\nabla\widehat{u}^0 .n\in \mathbb{R}^{3}, \quad \text{or}\ \  (\id-n\otimes n).\curl u\big|_{\Gamma}=(\id-n\otimes n)\,\curl \widehat{u}^0  \in \mathbb{R}^{3}$}}
 \put(-72,92){\footnotesize{Mechanical (traction) boundary conditions \ (3+2)}}
 \put(-72,87){\footnotesize{$(\sigma+\widehat{\tau}).\, n-\nabla[(\sym \widehat{M})\,(\id-n\otimes n)]: (\id-n\otimes n)\big|_{\partial \Omega\setminus\overline{\Gamma}}=\widehat{t}$,}}
\put(-36,82){\footnotesize{$(\id-n\otimes n)(\sym \widehat{M}).n\big|_{\partial \Omega\setminus \overline{\Gamma}}=(\id-n\otimes n)\,\widehat{h}$,}}
\put(-38,77){\footnotesize{$([\sym \widehat{M}]^+-[\sym \widehat{M}]^-). \, \nu\big|_{\partial \Gamma}=\widehat{\pi}$\ \ ``edge line force" on $\partial \Gamma$}}
\put(80,87){\footnotesize{3 bc}}
\put(20,87){\footnotesize{ $\widehat{M}=\left(
                                                 \begin{array}{c}
                                                   \widehat{m}_1\times n \_\_\_ \\
                                                    \widehat{m}_2\times n \_\_\_\\
                                                    \widehat{m}_3\times n\_\_\_
                                                 \end{array}
                                               \right), \ \widehat{m}=\left(
                                                 \begin{array}{c}
                                                   \widehat{m}_1 \_\_\_\\
                                                    \widehat{m}_2 \_\_\_\\
                                                    \widehat{m}_3\_\_\_
                                                 \end{array}
                                               \right)
$}}
\put(80,82){\footnotesize{2 bc}}
\put(80,77){\footnotesize{3 bc}}
  \put(-72,72){\footnotesize{Boundary virtual work}}
\put(-72,67){\footnotesize{$-\int_{\partial\Omega}\langle(\sigma+\widehat{\tau}).\, n,\delta u \rangle da- \int_{\partial\Omega}\sum\limits_{i=1}^3 \langle \widehat{ {m}}_i\times n, (\sym \nabla \delta u )_i\rangle\, da=
0 \quad \qquad \Leftrightarrow$}}
\put(-72,60){\footnotesize{$-\int_{\partial \Omega}  \langle (\sigma+\widehat{\tau}).\, n-\nabla[(\sym \widehat{M})\,(\id-n\otimes n)]: (\id-n\otimes n),  \delta u\rangle da-\int_{\partial \Omega}\langle (\id-n\otimes n)(\sym \widehat{M}).n,(\nabla\delta u).n\rangle  da$}}
\put(-65,55){\footnotesize{$-
\int_{\partial\Gamma}\langle ([\sym \widehat{M}]^+-[\sym \widehat{M}]^-). \, \nu, \delta u\rangle ds=0$}}

\put(5,44){{\huge $\Updownarrow$} \  \footnotesize{equivalent}}

\put(5,18){\oval(165,44)}
\put(-72,35){\bf\footnotesize{Correct boundary conditions in the $\boldsymbol{\Curl(\sym \nabla u)}$--formulation, written in indices}}
\put(-72,30){\footnotesize{Geometric (essential) boundary conditions \ (3+2)}}
 \put(-72,25){\footnotesize{$u_i\big|_{\Gamma}=\widehat{u}^0_{,i}\in\mathbb{R}^{ 3}, \qquad \ \left(\epsilon_{ikl}u_{l,k}-\epsilon_{jkl}u_{l,k}n_jn_i\right)\big|_{\Gamma}= \epsilon_{ikl}\widehat{u}^0_{l,k}-\epsilon_{jkl}\widehat{u}^0_{l,k}n_jn_i$}}
 \put(-71.7,20){\footnotesize{$\quad \ \ \qquad \qquad \qquad \quad \text{or}\ \ \left(u_{i,k}n_k-u_{j,k}n_kn_jn_i\right)\big|_{\Gamma}=\widehat{u}^0_{i,k}n_k-\widehat{u}^0_{j,k}n_kn_jn_i$}}
 \put(-72,15){\footnotesize{Mechanical (traction) boundary conditions \ (3+2)}}
 \put(-72,10){\footnotesize{$(\sigma_{ij}+\widehat{\tau}_{ij})\,n_j-\frac{1}{2}(\epsilon_{ikl}\widehat{m}_{hk}n_l
 +\epsilon_{hkl}\widehat{m}_{ik}n_l-
 \epsilon_{jkl}\widehat{m}_{ik}n_ln_jn_h- \epsilon_{ikl}\widehat{m}_{jk}n_ln_jn_h)_{,p}(\delta_{hp}-n_hn_p)\big|_{\partial \Omega\setminus\overline{\Gamma}}=\widehat{t}_i$,}}
\put(-57,5){\footnotesize{$\hspace{2.25cm}\frac{1}{2}(\epsilon_{ikl}\widehat{m}_{pk}n_l+\epsilon_{pkl}\widehat{m}_{ik}n_l-
 \epsilon_{jkl}\widehat{m}_{pk}n_ln_jn_i- \epsilon_{pkl}\widehat{m}_{jk}n_ln_jn_i)\,n_p\big|_{\partial \Omega\setminus \overline{\Gamma}}=\widehat{h}_i-\widehat{h}_pn_pn_i$,}}
\put(-57.7,0){\footnotesize{$\hspace{3.2cm}\frac{1}{2}([\epsilon_{pkl}\widehat{m}_{ik}n_l+\epsilon_{ikl}\widehat{m}_{pk}n_l]^+
-[\epsilon_{pkl}\widehat{m}_{ik}n_l+\epsilon_{ikl}\widehat{m}_{pk}n_l]^-) \, \nu_p\big|_{\partial \Gamma}=\widehat{\pi}_i$}}
\put(80,10){\footnotesize{3 bc}}
\put(80,5){\footnotesize{2 bc}}
\put(80,0){\footnotesize{3 bc}}
 \end{picture}
\end{center}
\caption{The possible boundary conditions in the $\nabla [\axl (\skw \nabla u)]$ and $\Curl(\sym \nabla u)$--formulation. The equivalence of the geometric boundary condition is clear. The virtual displacement is denoted by $\delta u\in C^\infty(
\Omega, \mathbb{R}^3)$.}\label{limitmodel00}
\end{figure}

\newpage

\section{Relation to the Cosserat-micropolar and micromorphic model}\label{sectcoss}\setcounter{equation}{0}

We have seen that it is irrelevant whether we take $\nabla [\axl (\skw \nabla u)]$ or $\Curl\, (\sym \nabla u)$  as basic curvature measures for the indeterminate couple stress model as long as consistent requirements on $\Gamma= \partial \Omega$ are considered and the following Dirichlet conditions are used  both together
\begin{align}
 u\big|_{\Gamma}=u_0, \quad  (\id-n\otimes n)\nabla  u.n\big|_{\Gamma}= (\id-n\otimes n) \nabla  u_0.n \  \Leftrightarrow \
  u\big|_{\Gamma}=u_0, \quad    (\id-n\otimes n).\curl u\big|_{\Gamma}=  (\id-n\otimes n).\curl u_0.\notag
 \end{align}
  The difference of the formulation appears only when considering mixed Dirichlet-Neumann boundary conditions.
 However, when we want to switch from a 4th-order (gradient elastic) problem to a 2nd-order micromorphic model or Cosserat model \cite{Cosserat09,NeffGhibaMicroModel,Eringen99}, we need to introduce new independent variables and decide about the useful coupling conditions in terms of adding a penalty term. It is also clear that adding more variables it depends on the number of the added fields whether the  new formulation is weaker softer in the language of a finite element context. In general, more degrees of freedom mean weaker response, at the prize of needing to specify more boundary conditions.

We discuss the following cases:
\begin{itemize}
\item[i)] {\bf [Cosserat]} {\bf $\boldsymbol{\skew \nabla u \mapsto \overline{A}\in \so(3)}$}. In the case $\nabla[\axl (\skw \nabla u)]$ we are led to introduce a skew-symmetric variable $\overline{A}\in \so(3)$ instead of $\skw \nabla u$, thus using the curvature tensor  $\nabla \axl(\overline{A})$ together with the coupling $$\mu_c\,\|\skw(\nabla u)-\overline{A}\|^2=\frac{\mu_c}{2}\,\|\curl u-2\axl(\overline{A})\|^2,$$ leading to the classical Cosserat model, with a new penalty parameter $\mu_c>0$ known as the Cosserat couple modulus. To be more precise, the corresponding minimization problem becomes
\begin{align*}
I(u,\overline{A})=\int_\Omega \bigg[&\mu\, \|{\rm sym} \nabla u\|^2+\frac{\lambda}{2}\, [\tr({\rm sym} \nabla u)]^2+\mu_c\|\skw(\nabla u)-\overline{A}\|^2\\
&+\mu\,L_c^2\,\big(\alpha_1\|\dev\sym \nabla \axl(\overline{A})\|^2+\alpha_2\tr[\nabla \axl(\overline{A})]^2+\alpha_2\|\skew \nabla \axl(\overline{A})\|^2\big)\bigg] dV\ \ \rightarrow\ \ \text{min.}
\end{align*}
 w.r.t $u\in H^1_0(\Omega), \ \overline{A}\in H^1_0(\Omega)$. In this case, the force-stress tensor is clearly non-symmetric
\begin{align}
\sigma=2\,\mu\, \sym\nabla u+2\,\mu_c\,(\skew\nabla u-\overline{A})+ \lambda\,\tr(\nabla u)\, \id\not\in {\rm Sym}(3),
\end{align}
and the couple stress tensor (hyperstress tensor)  is given by
\begin{align}
\widetilde{ {m}}=2\,\mu\,L_c^2\,\big(\alpha_1 \,\dev\sym \nabla \axl(\overline{A})+\alpha_2\,\tr[\nabla \axl(\overline{A})]\, \id+\alpha_2\,\skew \nabla \axl(\overline{A})\big),
\end{align}
which is also in general non-symmetric. Note that $\widetilde{m}$ has now 3 independent length scale parameters.
\item[ii)] {\bf [microstrain]} {\bf $ \boldsymbol{\sym \nabla u \mapsto \widehat{\varepsilon}\in {\rm Sym}(3)}$}.
 In the case of starting with the representation $\Curl(\sym \nabla u)$ we are led to introduce a symmetric tensor variable $\widehat{\varepsilon}\in{\rm Sym}(3)$ instead of $\sym \nabla u$, thus using the curvature measure $\Curl \widehat{\varepsilon}$ together with the coupling $$\varkappa^+\|\sym \nabla u-\widehat{\varepsilon}\|^2,$$ leading to a ``microstrain" theory \cite{Forest06,Neff_Forest07}, the minimization problem is now
 \begin{align*}
I(u,\widehat{\varepsilon})=\int_\Omega \bigg[&\mu\, \|{\rm sym} \nabla u\|^2+\frac{\lambda}{2}\, [\tr({\rm sym} \nabla u)]^2+\varkappa^+\|\sym \nabla u-\widehat{\varepsilon}\|^2\\&+\mu\,L_c^2\big(\beta_1\|\dev\sym \Curl \widehat{\varepsilon}\|^2+\beta_3\|\skew \Curl \widehat{\varepsilon}\|^2\big)\bigg] dV\ \ \rightarrow\ \ \text{min.}
\end{align*}
w.r.t $u\in H^1_0(\Omega), \ \widehat{\varepsilon}\in H_0^1(\Curl;\Omega)$, and, in this case, the force-stress tensor is symmetric
\begin{align}
\sigma=2\,\mu\, \sym\nabla u+2\,\varkappa^+\,(\sym\nabla u-\widehat{\varepsilon})+ \lambda\,\tr(\nabla u)\, \id\in {\rm Sym}(3),
\end{align}
and the hyperstress-tensor is given by
\begin{align}
\widehat{ {m}}=2\,\mu\,L_c^2\,\big(\beta_1\, \dev\sym \Curl \widehat{\varepsilon}+\beta_3\,\skew \Curl \widehat{\varepsilon}\big),
\end{align}
which is non-symmetric in general, depending on the material parameters. Note again that $\tr(\Curl \widehat{\varepsilon})=0$, thus the spherical part of the hyperstress tensor vanishes and $\widehat{m}$ features only 2 independent length scale parameters.
\item[iii)] {\bf[micromorphic]} {\bf $ \boldsymbol{\nabla u \mapsto p}$}.
 In this  case  we  may introduce a tensor $p\in\mathbb{R}^{3\times 3}$ instead of $\nabla u$,  and use the coupling $$\varkappa^+\|\nabla u-p\|^2$$ leading to a micromorphic  theory \cite{Eringen99,NeffGhibaMicroModel}, the minimization problem being
 \begin{align*}
I(u,p)=\int_\Omega \bigg[&\mu\, \|{\rm sym} \nabla u\|^2+\frac{\lambda}{2}\, [\tr({\rm sym} \nabla u)]^2+\varkappa^+\|\nabla u-p\|^2\\
&+\mu\,L_c^2\big(\gamma_1\|\dev\sym \Curl (\sym p)\|^2+\gamma_3\|\skew \Curl (\sym p)\|^2\big)\bigg] dV\ \ \rightarrow\ \ \text{min.}
\end{align*}
w.r.t $u\in H^1_0(\Omega), \ p\in H_0^1(\Curl;\Omega)$. We also point out that the force-stress tensor in this formulation will be
 non-symmetric
\begin{align}
\sigma&=2\,\mu\, \sym\nabla u+2\,\varkappa^+\,(\nabla u-P)+ \lambda\,\tr(\nabla u)\, \id\\
&=2\,\mu\, \sym\nabla u+2\,\varkappa^+\,\skw(\nabla u-P)+ 2\,\varkappa^+\,\sym(\nabla u-P)+\lambda\,\tr(\nabla u)\, \id\not\in {\rm Sym}(3)\notag
\end{align}
and the hyperstress tensor   (non-symmetric) is given by
\begin{align}
\widehat{ {m}}=2\,\mu\,L_c^2\,\big(\gamma_1\, \dev\sym \Curl (\sym p)+\gamma_3\,\skew \Curl (\sym p)\big).
\end{align}
In this formulation $\tr(\Curl \sym p)$ does not appear since $\tr(\Curl \sym p)=0, \forall p\in \mathbb{R}^{3\times3}$. Thus the spherical part of the hyperstress tensor vanishes, and $\widehat{m}$ features only 2 independent length scale parameters.
\item[iv)] {\bf [relaxed micromorphic]} for comparison with other extended continuum models we present the {\bf relaxed micromorphic model} \cite{NeffGhibaMicroModel,GhibaNeffExistence,MadeoNeffGhibaW,MadeoNeffGhibaWZAMM,NeffGhibaMadeoLazar}.
 In the relaxed micromorphic model, the minimization problem is of the type
 \begin{align*}
I(u,p)=\int_\Omega \bigg[&\mu\, \|{\rm sym} \nabla u\|^2+\frac{\lambda}{2}\, [\tr({\rm sym} \nabla u)]^2+\varkappa^+\,\|\sym(\nabla u-p)\|^2\\
&+L_c^2\big(\alpha_1\,\|\dev \sym\Curl p\|^2+\alpha_2\,\|\skew\Curl p\|^2+\alpha_3\,[\tr(\Curl p)]^2\big)\bigg] dV
\end{align*}
w.r.t $u\in H^1_0(\Omega), \ p\in H_0^1(\Curl;\Omega)$, and the corresponding force-stress tensor is symmetric
\begin{align}
\sigma=2\,\mu\, \sym\nabla u+2\,\varkappa^+\,\sym(\nabla u-p)+ \lambda\,\tr(\nabla u)\, \id\in{\rm Sym}(3)
\end{align}
and the hyperstress tensor is given by
\begin{align}
\widehat{ {m}}=2\,\mu\,L_c^2\,\big(\alpha_1\,\dev \sym\Curl p+\alpha_2\,\skew\Curl p+\alpha_3 \tr[\Curl p]\, \id\big),
\end{align}
with a non-vanishing spherical part of the hyperstress tensor. Note that $\widehat{m}$ has 3 independent material parameters.

\item[v)] we have also proposed  a  {\bf further relaxed micromorphic model} \cite{NeffGhibaMicroModel,GhibaNeffExistence,
    MadeoNeffGhibaW,MadeoNeffGhibaWZAMM,NeffGhibaMadeoLazar}, in which case
 the minimization problem is of the type \begin{align*}
I(u,p)=\int_\Omega \bigg[&\mu\, \|{\rm sym} \nabla u\|^2+\frac{\lambda}{2}\, [\tr({\rm sym} \nabla u)]^2+\varkappa^+\,\|\sym(\nabla u-p)\|^2\\
&+\mu\,L_c^2\big(\alpha_1\,\|\dev \sym\Curl p\|^2+\alpha_2\,\|\skew\Curl p\|^2\big)\bigg] dV
\end{align*}
w.r.t $u\in H^1_0(\Omega), \ p\in H_0^1(\Curl;\Omega)$,  the corresponding force-stress tensor is symmetric
\begin{align}
\sigma=2\,\mu\, \sym\nabla u+2\,\varkappa^+\,\sym(\nabla u-p)+ \lambda\,\tr(\nabla u)\, \id\in{\rm Sym}(3),
\end{align} and the hyperstress $\widehat{ {m}}$ is trace free
\begin{align}
\widehat{ {m}}=2\,\mu\,L_c^2\,\big(\alpha_1\,\dev \Curl p+\alpha_2\,\skew\Curl p\big).
\end{align}
The further relaxed micromorphic model remains well-posed \cite{GhibaNeffExistence}. A still weaker variant is v) with $\alpha_2=0$. Whether this choice is mathematically well-posed is yet unclear.
\end{itemize}

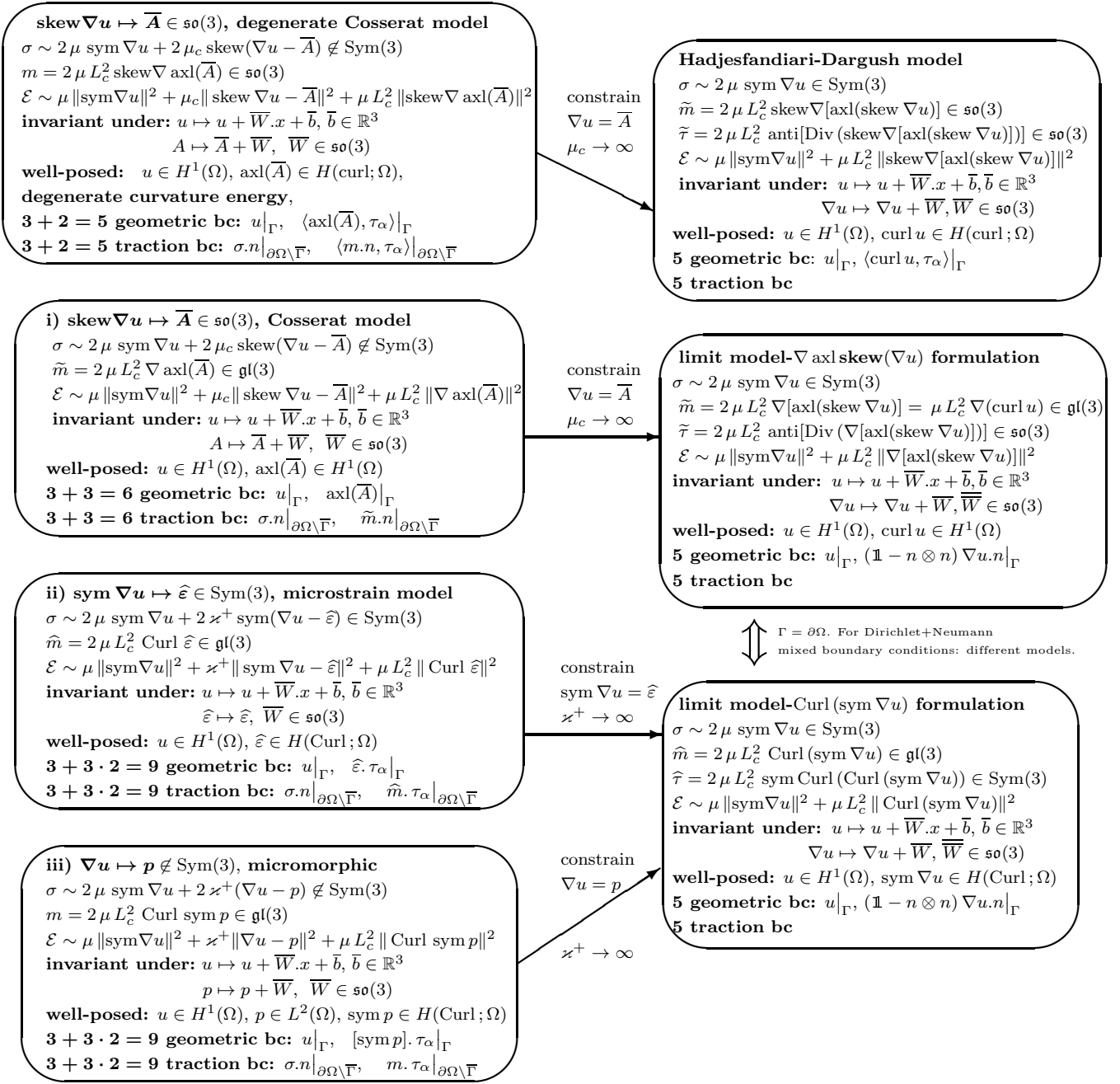
\begin{figure}
\setlength{\unitlength}{1mm}
\begin{center}
\begin{picture}(10,140)
\thicklines

\put(-40,147){\oval(86,42)}
\put(-81,164){\footnotesize{\bf \ \ \ $\boldsymbol{\skew \nabla u \mapsto \overline{A}}\in \so(3)$, degenerate Cosserat model }}
\put(-81,160){\footnotesize{ ${\sigma\sim 2\,\mu\, \sym\nabla u+2\, \mu_c \skw(\nabla u-\overline{A})\not\in {\rm Sym}(3)}$}}
\put(-81,156){\footnotesize{ $m=2\,\mu\,L_c^2\, \skew\nabla \axl(\overline{A})\in \so(3)$}}
\put(-81,152){\footnotesize{ $\mathcal{E}\sim\mu\,\|{\rm sym} \nabla u\|^2+\mu_c\|\skw \nabla u-\overline{A}\|^2+\mu\,L_c^2\,\|\skew\nabla \axl(\overline{A})\|^2$}}
\put(-81,148){\footnotesize{ \bf invariant under: \!$u\mapsto u+\overline{W}.x+\overline{b}, \, \overline{b}\in \mathbb{R}^3$}}
\put(-81,144){\footnotesize{\bf \qquad\qquad\quad \qquad\    $A\mapsto \overline{A}+{\overline{W}}, \,\,\,\overline{W}\in \so(3)$}}
\put(-81,140){\footnotesize{ {\bf well-posed: } $u\in H^1(\Omega)$, $\axl(\overline{A})\in H({\rm curl};\Omega)$, }}
\put(-81,136){\footnotesize{ {\bf degenerate curvature energy}, }}
\put(-81,132){\footnotesize{ \bf   $\boldsymbol{3+2=5}$ geometric bc: $u\big|_\Gamma, \ \  \langle \axl(\overline{A}),\tau_\alpha\rangle\big|_\Gamma$}}
\put(-81,128){\footnotesize{ \bf   $\boldsymbol{3+2=5}$ traction bc: $\sigma.n\big|_{\partial \Omega\setminus \overline{\Gamma}},\ \ \ \langle m.n,\tau_\alpha\rangle\big|_{\partial \Omega\setminus \overline{\Gamma}}$ }}

\put(3,144){\vector(2,-1){19}}
\put(8,148){\footnotesize{$\nabla u=\overline{A}$}}
\put(8,152){\footnotesize{constrain}}
\put(8,144){\footnotesize{$\mu_c\rightarrow\infty$}}

\put(-40,101){\oval(82,38)}
\put(-76,116){\footnotesize{\bf i) $\boldsymbol{\skew \nabla u \mapsto \overline{A}}\in \so(3)$, Cosserat model}}
\put(-76,112){\footnotesize{ ${\sigma\sim 2\,\mu\,\sym\nabla u+2\, \mu_c \skw(\nabla u-\overline{A})\not\in {\rm Sym}(3)}$}}
\put(-76,108){\footnotesize{ \bf $\widetilde{m}=2\,\mu\,L_c^2\, \nabla \axl(\overline{A})\in \gl(3)$}}
\put(-76,104){\footnotesize{ \bf $\mathcal{E}\sim\mu\,\|{\rm sym} \nabla u\|^2+\mu_c\|\skw \nabla u-\overline{A}\|^2\!+\mu\,L_c^2\,\|\nabla \axl(\overline{A})\|^2$}}
\put(-76,100){\footnotesize{ \bf invariant under: \!$u\mapsto u+\overline{W}.x+\overline{b}, \, \overline{b}\in \mathbb{R}^3$}}
\put(-76,96){\footnotesize{ \bf \qquad\qquad\quad \qquad\    $A\mapsto \overline{A}+{\overline{W}}, \,\,\,\overline{W}\in \so(3)$}}
\put(-76,92){\footnotesize{{\bf  well-posed:} $u\in H^1(\Omega)$, $\axl(\overline{A})\in H^1(\Omega)$ }}
\put(-76,88){\footnotesize{\bf $\boldsymbol{3+3=6}$ geometric bc: $u\big|_\Gamma, \ \   \axl(\overline{A})\big|_\Gamma$}}
\put(-76,84){\footnotesize{\bf $\boldsymbol{3+3=6}$ traction bc: $\sigma.n\big|_{\partial \Omega\setminus \overline{\Gamma}},\ \ \  \widetilde{m}.n\big|_{\partial \Omega\setminus \overline{\Gamma}}$}}

\put(-40,57){\oval(82,38)}
\put(-76,72){\footnotesize{\bf ii)  $ \boldsymbol{\sym\nabla u \mapsto \widehat{\varepsilon}}\in{\rm Sym}(3)$, microstrain model}}
\put(-76,68){\footnotesize{\bf ${\sigma\sim 2\,\mu\,\sym\nabla u+ 2\, \varkappa^+ \sym(\nabla u-\widehat{\varepsilon})\in {\rm Sym}(3)}$}}
\put(-76,64){\footnotesize{\bf $\widehat{m}=2\,\mu\,L_c^2\, \Curl \widehat{\varepsilon}\in \gl(3)$}}
\put(-76,60){\footnotesize{\bf $\mathcal{E}\sim\mu\,\|{\rm sym} \nabla u\|^2+\varkappa^+\|\sym \nabla u-\widehat{\varepsilon}\|^2+\mu\,L_c^2\,\|\Curl \widehat{\varepsilon}\|^2$}}
\put(-76,56){\footnotesize{\bf invariant under: \!$u\mapsto u+\overline{W}.x+\overline{b}, \,  \overline{b}\in \mathbb{R}^3$}}
\put(-76,52){\footnotesize{\bf \qquad\qquad\quad \qquad\    $\widehat{\varepsilon}\mapsto \widehat{\varepsilon}, \,\,\overline{W}\in \so(3)$}}
\put(-76,48){\footnotesize{{\bf  well-posed:} $u\in H^1(\Omega)$, $\widehat{\varepsilon}\in H(\Curl;\Omega)$  }}
\put(-76,44){\footnotesize{\bf  $\boldsymbol{3+3\cdot 2=9}$ geometric bc: $u\big|_\Gamma, \ \    \widehat{\varepsilon}.\,\tau_\alpha\big|_\Gamma$ }}
\put(-76,40){\footnotesize{\bf $\boldsymbol{3+3\cdot 2=9}$  traction bc: $\sigma.n\big|_{\partial \Omega\setminus \overline{\Gamma}},\ \ \ \widehat{m}.\,\tau_\alpha\big|_{\partial \Omega\setminus \overline{\Gamma}}$}}

\put(-40,13){\oval(80,38)}
\put(-76,28){\footnotesize{\bf iii) $ \boldsymbol{\nabla u \mapsto p}\not\in {\rm Sym}(3)$}, \textbf{micromorphic}}
\put(-76,24){\footnotesize{\bf ${\sigma\sim  2\,\mu\,\sym\nabla u+ 2\,\varkappa^+(\nabla u-p)\not\in {\rm Sym}(3)}$}}
\put(-76,20){\footnotesize{\bf $m=2\,\mu\,L_c^2\, \Curl \sym p\in \gl(3)$}}
\put(-76,16){\footnotesize{\bf $\mathcal{E}\sim\mu\,\|{\rm sym} \nabla u\|^2+\varkappa^+\|\nabla u-p\|^2+\mu\,L_c^2\,\|\Curl \sym p\|^2$}}
\put(-76,12){\footnotesize{\bf invariant under: \!$u\mapsto u+\overline{W}.x+\overline{b}, \,  \overline{b}\in \mathbb{R}^3$}}
\put(-76,8){\footnotesize{\bf \qquad\qquad\quad \qquad\    $p\mapsto p+{\overline{W}},\,\,\,\overline{W}\in \so(3) $}}
\put(-76,4){\footnotesize{{\bf  well-posed:} $u\in H^1(\Omega)$, $p\in L^2(\Omega)$, $\sym p\in H(\Curl;\Omega)$ }}
\put(-76,0){\footnotesize{\bf $\boldsymbol{3+3\cdot 2=9}$ geometric bc:  $u\big|_\Gamma, \ \    [\sym p].\,\tau_\alpha\big|_\Gamma$  }}
\put(-76,-4){\footnotesize{\bf $\boldsymbol{3+3\cdot 2=9}$  traction bc: $\sigma.n\big|_{\partial \Omega\setminus \overline{\Gamma}},\ \ \  m.\,\tau_\alpha\big|_{\partial \Omega\setminus \overline{\Gamma}}$}}

\put(58,141){\oval(72,42)}
\put(25,158){\footnotesize{ \bf Hadjesfandiari-Dargush model}}
\put(25,154){\footnotesize{ \bf $\sigma\sim 2\,\mu\, \sym\nabla u\in {\rm Sym}(3)$}}
\put(25,150){\footnotesize{ $\widetilde{ {m}}=2\,\mu\,L_c^2\,  \skew \nabla [\axl (\skw \nabla u)]\in\so(3)$}}
\put(25,146){\footnotesize{ $\widetilde{\tau}=2\,\mu\,L_c^2\,  \anti[\Div(\skew\nabla [\axl (\skw \nabla u)])]\in\so(3)$}}
\put(25,142){\footnotesize{ $\mathcal{E}\sim\mu\,\|{\rm sym} \nabla u\|^2+\mu\,L_c^2\,\|\skew \nabla [\axl (\skw \nabla u)]\|^2$}}
\put(25,138){\footnotesize{ \bf invariant under:  $u\mapsto u+\overline{W}.x+\overline{b},  \overline{b}\in \mathbb{R}^3$}}
\put(25,134){\footnotesize{ \bf \ \quad\quad\quad\quad \qquad\    $\nabla u\mapsto \nabla u+\overline{W}, {\overline{W}}\in \so(3)$}}
\put(25,130){\footnotesize{{\bf well-posed:}  $u\in H^1(\Omega)$,  $\curl u\in H(\curl;\Omega)$  }}
\put(25,126){\footnotesize{\bf 5 geometric bc}:  $u\big|_\Gamma$, $\langle \curl u,\tau_\alpha\rangle\big|_{\Gamma}$}
\put(25,122){\footnotesize{\bf  5 traction bc}}

\put(59,92.5){\oval(72,44)}
\put(26,110){\footnotesize{\bf limit model-$\nabla \axl \skew(\nabla u)$ formulation}}
\put(25,106){\footnotesize{\bf $\sigma\sim  2\,\mu\,\sym\nabla u\in {\rm Sym}(3)$}}
\put(25,102){\footnotesize{ $\widetilde{ {m}}=2\,\mu\,L_c^2\,  \nabla [\axl (\skw \nabla u)]=\,\mu\,L_c^2\,  \nabla (\curl u)\in\gl(3)$}}
\put(25,98){\footnotesize{ $\widetilde{\tau}=2\,\mu\,L_c^2\,  \anti[\Div(\nabla [\axl (\skw \nabla u)])]\in\so(3)$}}
\put(25,94){\footnotesize{ $\mathcal{E}\sim\mu\,\|{\rm sym} \nabla u\|^2+\mu\,L_c^2\,\|\nabla [\axl (\skw \nabla u)]\|^2$}}
\put(25,90){\footnotesize{\bf invariant under:  $u\mapsto u+\overline{W}.x+\overline{b},  \overline{b}\in \mathbb{R}^3$}}
\put(25,86){\footnotesize{\bf \qquad\qquad\quad \qquad\    $\nabla u\mapsto \nabla u+\overline{W}, \overline{\overline{W}}\in \so(3)$}}
\put(25,82){\footnotesize{{\bf well-posed:}  $u\in H^1(\Omega)$,  $\curl u\in H^1(\Omega)$  }}
\put(25,78){\footnotesize{\bf  5 geometric bc:} $u\big|_\Gamma$, $(\id-n\otimes n)\,\nabla u.n\big|_{\Gamma}$}
\put(25,74){\footnotesize{\bf   5 traction bc} }

\put(57,37){\oval(67,43)}
\put(26,54){\footnotesize{\bf limit model-$\Curl (\sym \nabla  u)$ formulation}}
\put(25,50){\footnotesize{\bf $\sigma\sim  2\,\mu\,\sym\nabla u\in {\rm Sym}(3)$}}
\put(25,46){\footnotesize{\bf $\widehat{ {m}}=2\,\mu\,L_c^2\, \Curl (\sym \nabla  u)\in\gl(3)$}}
\put(25,42){\footnotesize{\bf $\widehat{\tau}=2\,\mu\,L_c^2\, \sym \Curl (\Curl (\sym \nabla  u))\in {\rm Sym}(3)$}}
\put(25,38){\footnotesize{\bf $\mathcal{E}\sim\mu\,\|{\rm sym} \nabla u\|^2+\mu\,L_c^2\,\|\Curl(\sym \nabla u)\|^2$}}
\put(25,34){\footnotesize{\bf invariant under:  $u\mapsto u+\overline{W}.x+\overline{b}, \, \overline{b}\in \mathbb{R}^3$}}
\put(25,30){\footnotesize{\bf \qquad\qquad\quad \quad\    $\nabla u\mapsto \nabla u+\overline{W},\, \overline{\overline{W}}\in \so(3)$}}
\put(25,26){\footnotesize{{\bf well-posed:}  $u\in H^1(\Omega)$,  $\sym\nabla u\in H(\Curl;\Omega)$  }}
\put(25,22){\footnotesize{\bf  5 geometric bc:} $u\big|_\Gamma$, $(\id-n\otimes n)\,\nabla u.n\big|_{\Gamma}$}
\put(25,18){\footnotesize{\bf   5 traction bc} }

        \put(1,98){\vector(1,0){22}}
      \put(8,100){\footnotesize{$\mu_c\rightarrow\infty$}}
      \put(8,108){\footnotesize{constrain}}
      \put(8,104){\footnotesize{$\nabla u=\overline{A}$}}

      \put(7,60){\footnotesize{constrain}}
      \put(7,56){\footnotesize{$\sym \nabla u=\widehat{\varepsilon}$}}
      \put(7,52){\footnotesize{$\varkappa^+\rightarrow\infty$}}
      \put(1,50){\vector(1,0){22.5}}

       \put(36,63){\huge$ \Updownarrow$}
        \put(42,66){\tiny{$\Gamma=\partial \Omega$. For Dirichlet+Neumann  }}
        \put(42,63){\tiny{mixed boundary  conditions: different models.}}

 \put(7,29){\footnotesize{constrain}}
    \put(7,25){\footnotesize{$\nabla u=p$}}
      \put(7,14){\footnotesize{$\varkappa^+\rightarrow\infty$}}
      \put(0,13){\vector(3,2){23}}
\end{picture}
\end{center}
\caption{Different possibilities of lifting the variants of the 4th.-order indeterminate  couple stress model to a 2nd.-order micromorphic or Cosserat-type formulation formulation. In the penalty case, all the considered alternatives lead to the same limit model provided only geometric  boundary conditions are imposed. It is not surprising that the limit model has some peculiarities since the limit procedure is itself  singular.  Note that different micromorphic or Cosserat type formulations generate  different sets of boundary conditions. Here $\tau_\alpha$, $\alpha=1,2$ denote two independent tangential vectors on the boundary. }\label{limitmodel}
\end{figure}
\begin{figure}
\setlength{\unitlength}{1mm}
\begin{center}
\begin{picture}(0,0)
\thicklines
\put(-5,6){\oval(85,15)}
\put(-45,8){$\nabla [\axl (\skw \nabla u)]=[\Curl (\sym \nabla  u)]^T$ \quad \text{but}\qquad}
\put(-45,0){\quad$\nabla \axl(\skew\, p)\neq[\Curl \sym p]^T \qquad \text{for} \qquad p\neq \nabla u$}
\end{picture}
\end{center}
\caption{Integrability conditions.}\label{integral}
\end{figure}

\begin{figure}
\setlength{\unitlength}{1mm}
\begin{center}
\begin{picture}(10,150)
\thicklines

\put(-35,150){\oval(82,40)}
\put(-72,166){\footnotesize{\bf iv) the relaxed micromorphic model}}
\put(-72,162){\footnotesize{\bf $\sigma\sim\,2\,\mu\,\sym \nabla u+2\,\varkappa^+\sym(\nabla u-p)\in {\rm Sym}(3)$}}
\put(-72,158){\footnotesize{\bf $m=2\,\mu\,L_c^2\, \Curl p\in \gl(3)$}}
\put(-72,154){\footnotesize{\bf $\mathcal{E}\sim\mu\,\|{\rm sym} \nabla u\|^2+\varkappa^+\,\|\sym(\nabla u-p)\|^2+\mu\,L_c^2\,\|\Curl p\|^2$}}
\put(-72,150){\footnotesize{\bf invariant under: \!$u\mapsto u+\overline{W}.x+\overline{b}, \,  \overline{b}\in \mathbb{R}^3$}}
\put(-72,146){\footnotesize{\bf \qquad\qquad\quad \qquad\    $p\mapsto p+\overline{\overline{W}}, \ \ {\overline{W}}\neq \overline{\overline{W}},\,\,\,\overline{W}, \overline{\overline{W}}\in \so(3)$}}
\put(-72,142){\footnotesize{\bf well-posed:}  $u\in H^1(\Omega)$,  $p\in H(\Curl;\Omega)$ }
\put(-72,138){\footnotesize{\bf $\boldsymbol{3+3\cdot 2=9}$ geometric bc:}
$u\big|_\Gamma$, $p.\,\tau_\alpha \big|_{\Gamma}$}
\put(-72,134){\footnotesize{\bf $\boldsymbol{3+3\cdot 2=9}$ traction bc:}
$\sigma.n\,\big|_{\partial \Omega\setminus \overline{\Gamma}}$, $ {m}.\, \tau_\alpha\big|_{\partial \Omega\setminus \overline{\Gamma}}$}

\put(-35,104){\oval(86,38)}
\put(-74,120){\footnotesize{\bf v) the further relaxed micromorphic model}}
\put(-74,116){\footnotesize{\bf $\sigma\sim\,2\,\mu\,\sym \nabla u+2\,\varkappa^+\sym(\nabla u-p)\in {\rm Sym}(3)$}}
\put(-74,112){\footnotesize{\bf $m=2\,\mu\,L_c^2\, \dev\Curl p\in\mathfrak{sl}(3)$}}
\put(-74,108){\footnotesize{\bf $\mathcal{E}\sim\mu\,\|{\rm sym} \nabla u\|^2+\varkappa^+\,\|\sym(\nabla u-p)\|^2+\mu\,L_c^2\,\|\dev \Curl p\|^2$}}
\put(-74,104){\footnotesize{\bf invariant under: \!$u\mapsto u+\overline{W}.x+\overline{b}, \,  \overline{b}\in \mathbb{R}^3$}}
\put(-74,100){\footnotesize{\bf \qquad\qquad\quad \qquad\    $p\mapsto p+\overline{\overline{W}}, \ \ {\overline{W}}\neq \overline{\overline{W}},\,\,\,\overline{W}, \overline{\overline{W}}\in \so(3)$}}
\put(-74,96){\footnotesize{\bf well-posed:}  $u\in H^1(\Omega)$,  $p\in H(\Curl;\Omega)$ }
\put(-74,92){\footnotesize{\bf $\boldsymbol{3+3\cdot 2=9}$ geometric bc:}
$u\big|_\Gamma$, $p.\,\tau_\alpha \big|_{\Gamma}$}
\put(-74,88){\footnotesize{\bf $\boldsymbol{3+3\cdot 2=9}$ traction bc:}
$\sigma.n\,\big|_{\partial \Omega\setminus \overline{\Gamma}}$, $ {m}.\, \tau_\alpha\big|_{\partial \Omega\setminus \overline{\Gamma}}$}

\put(62,140){\oval(71,37)}
\put(30,154){\footnotesize{\bf limit model: incompatible linear elasticity}}
\put(30,150){\footnotesize{\bf $\sigma\sim \,2\,\mu\,\sym p\in {\rm Sym}(3)$}}
\put(30,146){\footnotesize{\bf $m=2\,\mu\,L_c^2\, \Curl p\in\gl(3)$}}
\put(30,142){\footnotesize{$\mathcal{E}\sim\mu\,\|{\rm sym}\, p\|^2+\mu\,L_c^2\,\|\Curl  p\|^2$}}
\put(30,138){\footnotesize{\bf invariant under:  $p\mapsto p+\overline{\overline{W}}, \, \overline{\overline{W}}\in \so(3)$}}
\put(30,134){\footnotesize{\bf well-posed:}    $p\in H(\Curl;\Omega)$ }
\put(30,130){\footnotesize{\bf $\boldsymbol{3\cdot 2=6}$ geometric bc:}
 $p.\,\tau_\alpha \big|_{\Gamma}$}
\put(30,126){\footnotesize{\bf $\boldsymbol{3\cdot 2=6}$ traction bc:}
 $ \Curl p.\, \tau_\alpha\big|_{\partial \Omega\setminus \overline{\Gamma}}$}

\put(61,90){\oval(67,35)}
\put(30,104){\footnotesize{\bf classical linear elasticity}}
\put(30,100){\footnotesize{\bf $\sigma\sim \,2\,\mu\,\sym \nabla u\in {\rm Sym}(3)$}}
\put(30,96){\footnotesize{$\mathcal{E}\sim\mu\,\|{\rm sym}\, \nabla u\|^2$}}
\put(30,92){\footnotesize{\bf invariant under:  $u\mapsto u+\overline{W}.x+\overline{b}$}}
\put(30,88){\footnotesize{\bf \qquad\qquad \qquad\ \    $\nabla u\mapsto \nabla u+\overline{W}, \overline{W}\in\so(3)$}}
\put(30,84){\footnotesize{\bf well-posed:}  $u\in H^1(\Omega)$,  $\nabla u\in H(\Curl;\Omega)$ }
\put(30,80){\footnotesize{\bf $3$ geometric bc:}
$u\big|_\Gamma$}
\put(30,76){\footnotesize{\bf $3$ traction bc:}
$\sigma.n\,\big|_{\partial \Omega\setminus \overline{\Gamma}}$}

\put(12,151){\tiny{constrain}}
 \put(8,148){\tiny{$\sym\nabla u=\sym p$}}
      \put(12,142){\footnotesize{$\varkappa^+\rightarrow\infty$}}
      \put(6,145){\vector(1,0){20}}
      \put(62,115){\footnotesize{$L_c\rightarrow\infty$}}
      \put(60,121){\vector(0,-1){13}}

      \put(61,42){\oval(67,37)}
\put(30,56){\footnotesize{\bf limit model: incompatible linear elasticity}}
\put(30,52){\footnotesize{\bf $\sigma\sim \,2\,\mu\,\sym p\in {\rm Sym}(3)$}}
\put(30,48){\footnotesize{\bf $m=2\,\mu\,L_c^2\, \sym \Curl p\in\gl(3)$}}
\put(30,44){\footnotesize{$\mathcal{E}\sim\mu\,\|{\rm sym}\, p\|^2+\mu\,L_c^2\,\|\sym \Curl  p\|^2$}}
\put(30,40){\footnotesize{\bf invariant under:  $p\mapsto p+\overline{\overline{W}}, \, \overline{\overline{W}}\in \so(3)$}}
  \put(30,36){\footnotesize{\bf well-posed:}    $p\in H(\Curl;\Omega)$ }
\put(30,32){\footnotesize{\bf $3\cdot 2=6$ geometric bc:}
 $p.\,\tau_\alpha \big|_{\Gamma}$}
\put(30,28){\footnotesize{\bf $3\cdot 2=6$ traction bc:}
 $ {m}.\, \tau_\alpha\big|_{\partial \Omega\setminus \overline{\Gamma}}$}

      \put(61,60.5){\vector(0,1){12}}
\put(63,66){\footnotesize{$L_c\rightarrow\infty$}}

\put(18,66){\tiny{constrain}}
  \put(12,70){\vector(1,-2){15.5}}
 \put(18,62){\tiny{$\sym\nabla u=\sym p$}}

      \put(12,142){\footnotesize{$\varkappa^+\rightarrow\infty$}}

\put(-33,61){\oval(90,39)}
\put(-74,76){\footnotesize{\bf vi) another relaxed micromorphic model}}
\put(-74,72){\footnotesize{\bf $\sigma\sim\,2\,\mu\,\sym \nabla u+2\,\varkappa^+\sym(\nabla u-p)\in {\rm Sym}(3)$}}
\put(-74,68){\footnotesize{\bf $m=2\,\mu\,L_c^2\,  \sym \Curl p\in {\rm Sym}(3)$, symmetric}}
\put(-74,64){\footnotesize{\bf $\mathcal{E}\sim\mu\,\|{\rm sym} \nabla u\|^2+\varkappa^+\,\|\sym(\nabla u-p)\|^2+\mu\,L_c^2\,\|\sym \Curl p\|^2$}}
\put(-74,60){\footnotesize{\bf invariant under: \!$u\mapsto u+\overline{W}.x+\overline{b}, \,  \overline{b}\in \mathbb{R}^3$}}
\put(-74,56){\footnotesize{\bf \qquad\qquad\quad \qquad\    $p\mapsto p+\overline{\overline{W}}, \ \ {\overline{W}}\neq \overline{\overline{W}},\,\,\,\overline{W}, \overline{\overline{W}}\in \so(3)$}}
\put(-74,52){\footnotesize{\bf well-posedness not clear :}  $u\in H^1(\Omega)$,  $p\in H(\Curl;\Omega)$  }
\put(-74,48){\footnotesize{\bf $\boldsymbol{3+3\cdot 2=9}$ geometric bc:}
$u\big|_\Gamma$, $p.\,\tau_\alpha \big|_{\Gamma}$}
\put(-74,44){\footnotesize{\bf $\boldsymbol{3+3\cdot 2=9}$ traction bc:}
$\sigma.n\,\big|_{\partial \Omega\setminus \overline{\Gamma}}$, $ {m}.\, \tau_\alpha\big|_{\partial \Omega\setminus \overline{\Gamma}}$}

\put(-28,15){\oval(105,41)}
\put(-77,32){\footnotesize{\bf vii)  relaxed micromorphic model with integral boundary coupling}}
\put(-77,28){\footnotesize{\bf $\sigma\sim\,2\,\mu\,\sym \nabla u+2\,\varkappa^+\sym(\nabla u-p)\in {\rm Sym}(3)$}}
\put(-77,24){\footnotesize{\bf $m=2\,\mu\,L_c^2\,  \Curl p\in\gl(3)$}}
\put(-77,20){\footnotesize{\bf $\mathcal{E}\sim\mu\,\|{\rm sym} \nabla u\|^2+\varkappa^+\,\|\sym(\nabla u-p)\|^2+\mu\,L_c^2\,\|\Curl p\|^2$}}
\put(-25,16){\footnotesize{\bf $+\int_\Gamma \|(\nabla u-p)\times n\|^2 da$}}
\put(-77,12){\footnotesize{\bf invariant under: \!$u\mapsto u+\overline{W}.x+\overline{b}, \,  \overline{b}\in \mathbb{R}^3$}}
\put(-77,8){\footnotesize{\bf \qquad\qquad\quad \qquad\    $p\mapsto p+\overline{\overline{W}}, \ \ {\overline{W}}\neq \overline{\overline{W}},\,\,\,\overline{W}, \overline{\overline{W}}\in \so(3)$}}
\put(-77,4){\footnotesize{\bf well-posed:}  $u\in H^1(\Omega)$,  $p\in H(\Curl;\Omega)$  }
\put(-77,0){\footnotesize{\bf $\boldsymbol{3}$ geometric bc:}
$u\big|_\Gamma$}
\put(-77,-4){\footnotesize{\bf $\boldsymbol{3+3\cdot 2=9}$ traction bc:}
$\sigma.n\,\big|_{\partial \Omega\setminus \overline{\Gamma}}$, $ {m}.\, \tau_\alpha\big|_{\partial \Omega\setminus \overline{\Gamma}}$}

\end{picture}
\end{center}
\caption{For comparison, the same situation as in Fig. \ref{limitmodel} for the relaxed micromorphic model  and the further relaxed micromorphic model \cite{NeffGhibaMicroModel}. The first two micromorphic formulations are well-posed. In our view these models are advantageous as compared to the models in Fig. \ref{limitmodel}. Note also that the boundary conditions for the new microdistortion field $p$ has 6 degrees of freedom, the only physical and mathematical possible choice is $p\times n\big|_{\Gamma}=0$ on $\Gamma$, where Dirichlet-boundary conditions $u\big|_\Gamma=0$ are prescribed. In the case of non-homogeneous (non-zero) boundary prescription $u$ on $\Gamma$, we would need to modify the total  energy by adding  as a boundary term $\int_\Gamma \|(\nabla u-p)\times n\|^2 ds$. This will introduce a certain coupling at the boundary. The model in this  format is still well-posed and awaits to be further investigated. A  nonlinear modification of this model is investigated in \cite{NeffLankeitOsterbrink}. }\label{limitmodel2}
\end{figure}
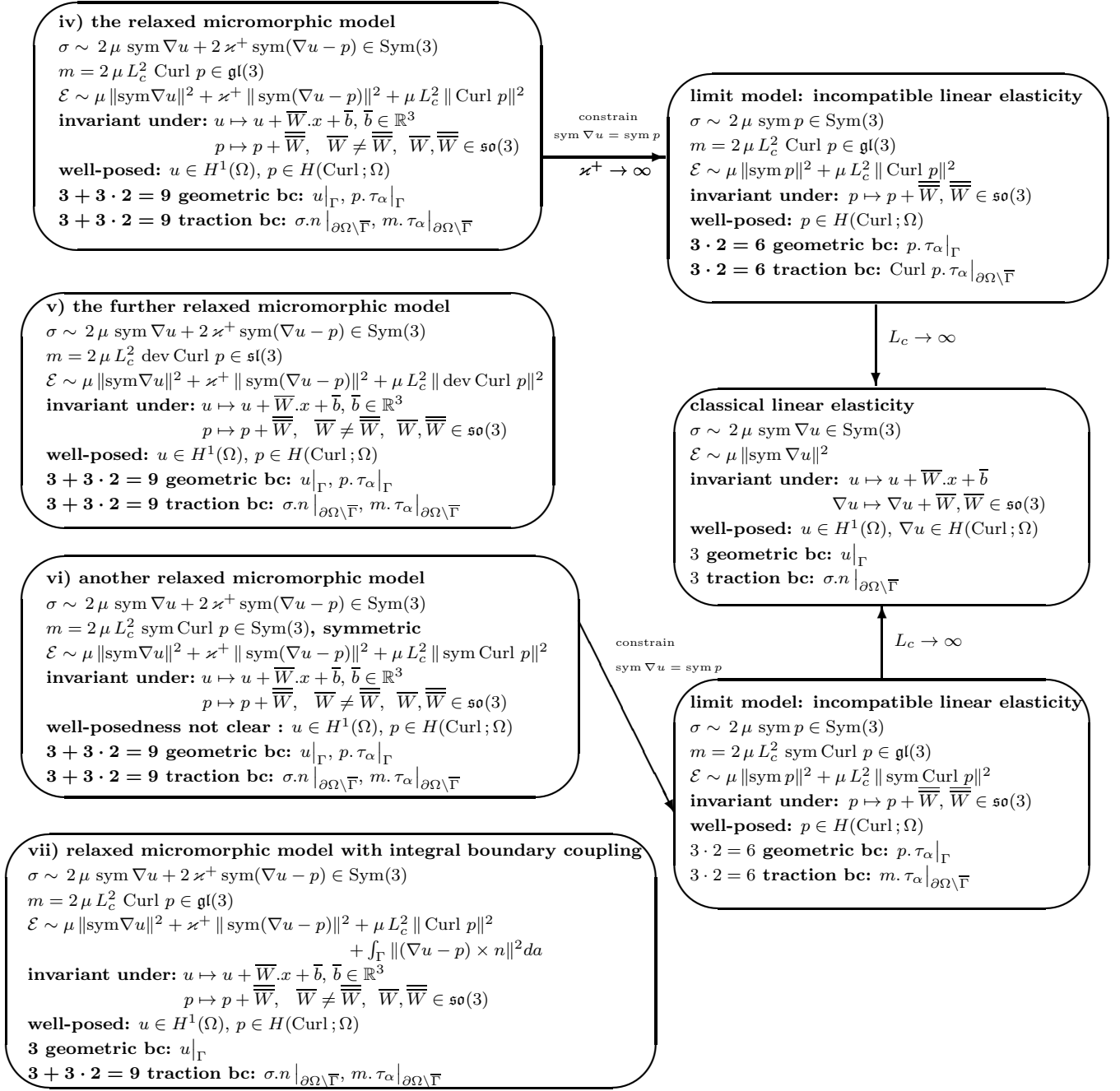

\section{Conclusion}

Our new symmetric-conformal $\Curl(\sym \nabla u)$-reformulation has the following crucial properties setting it apart  from existing formulations of couple-stress models:
\begin{itemize}
\item the local and the nonlocal force stress tensors ($\sigma$, $\widehat{\tau}$) are both symmetric, while the couple stress tensor $\widehat{m}$ is symmetric in the conformally-invariant model.
\item the curvature energy is conformally invariant and the couple stress tensor $\widehat{m}$ vanishes for conformal displacement.
\item the model has only one additional length scale parameter, similar to the modified couple stress model.
\item the model is derived with consistent boundary conditions: either 5 geometrical conditions or 5 mechanical (traction) conditions. The mechanical conditions are separated into force stress tractions and couple stress tractions and correspond to completely independent boundary conditions.
\item for mixed Dirichlet-Neumann boundary conditions the model does not reduce to the modified indeterminate couple stress model.
\end{itemize}

\begin{figure}
\setlength{\unitlength}{1mm}
\begin{center}
\begin{picture}(10,120)
\thicklines
\put(4,134){\oval(32,10)}
\put(-10,135){\footnotesize{\bf {\ Size effects}}}
\put(-10,131){\footnotesize{\bf {"Smaller is stiffer"}}}
\put(20,135){\vector(1,0){8.5}}
\put(5,129){\vector(0,-1){9}}

\put(-58,134){\oval(57,18)}
\put(-84,137){\footnotesize{\bf{Not only rotational interaction $D^2u$}}}
\put(-84,133){\footnotesize{\bf {gradient elasticity: $W_{\rm curv}(D^2 u)$}}}
\put(-84,129){\footnotesize{\bf { $\mathfrak{m}=D_{D^2 u}W_{\rm curv}(D^2 u)$,  third order }}}

\put(-12,135){\vector(-1,0){17.5}}
\put(64,134){\oval(71,17)}
\put(31,137){\footnotesize{\bf{Strain gradient elasticity}}}
\put(31,133){\footnotesize{\bf { $W_{\rm curv}(\nabla\, (\sym   \nabla   u))$}}}
\put(31,129){\footnotesize{\bf { $\mathfrak{m}=D_{\nabla\, (\sym   \nabla   u)}W_{\rm curv}(\nabla\, (\sym   \nabla   u))$,  third order }}}

\put(-90,116){\footnotesize{\bf \framebox[14.5cm][l]{\ Only rotational interaction through continuum rotation $W_{\rm curv}(\nabla (\curl u))=W_{\rm curv}(\Curl(\sym\nabla u))$}}}
\put(-60,108){\footnotesize{ $\widetilde{k}=\nabla [\axl (\skw \nabla u)]=\frac{1}{2}\nabla\curl u$}}

\put(77,117){\oval(40,9)}
\put(60,119){\tiny{\bf Cauchy-Boltzmann axiom}}
\put(60,117){\tiny{ $\boldsymbol{\sigma_{\rm total}=\sigma_{\rm total}^T}$}}
\put(60,114){\tiny{\bf purely mechanical interaction}}
\put(71,125.5){\vector(0,-1){3.5}}
\put(71,112.5){\vector(0,-1){11}}

\put(5,107){\oval(27,9)}
\put(0,107){\footnotesize{ $\widetilde{k}=\widehat{k}^T$}}
\put(-7,104){\tiny{ integrability condition}}

\put(38,108){\footnotesize{\bf $\widehat{k}=\Curl (\sym\nabla u)$}}
\put(-10,114){\vector(-1,-1){10}}
\put(10,114){\vector(4,-1){49}}

\put(-32,97){\oval(49,14)}
\put(-55,100){\footnotesize{\bf{ polar gradient elasticity}}}
\put(-55,96){\footnotesize{\bf{ indeterminate couple stress}}}
\put(-55,92){\footnotesize{\bf{ $\widetilde{ {m}}=D_{\widetilde{k}} W_{\rm curv}(\widetilde{k})$,  second order }}}
\put(-15,89.5){\vector(0,-1){52.5}}
\put(-9.5,91.5){\vector(4,-1){15}}
\put(-54.5,91.5){\vector(-4,-1){15}}

\put(68,96){\oval(47,11)}
\put(-7.2,97){\vector(1,0){51.2}}
\put(45,97){\footnotesize{\bf{ non-polar gradient elasticity}}}
\put(45,93){\footnotesize{\bf{ $\widehat{ {m}}=D_{\widehat{k}} W_{\rm curv}(\widehat{k})$,  second order}}}
\put(2,98){\tiny{\bf difference of formulation }}
\put(2,95){\tiny{\bf is a Null-Lagrangian in}}
\put(2,93){\tiny{\bf the balance of linear momentum}}
\put(6,91){\tiny{\bf and $\sym \widetilde{m}=\sym \widehat{m}$}}
\put(60,90.5){\vector(0,-1){3.5}}

\put(-80,84){\footnotesize{\bf\framebox[6cm][l]{ Grioli-Toupin-Koiter-Mindlin}}}
\put(-80,78){\footnotesize{$\mu\,\|\sym\nabla u\|^2$}}
\put(-80,74){\footnotesize{$+\mu\,L_c^2\,\|\nabla \axl(\skw\nabla u)\|^2$}}
\put(-80,70){\footnotesize{2 curvature parameters}}
\put(-80,66){\footnotesize{couple stress tensor $\widetilde{ {m}}\in\gl(3)$}}
\put(-80,62){\footnotesize{nonlocal force-stress $\widetilde{\tau}\in\so(3)$ }}
\put(-80,58){\footnotesize{unique solution}}
\put(-80,54){\footnotesize{extendable to a}}
\put(-80,50){\footnotesize{Cosserat formulation}}
\put(-80,46){\footnotesize{$\axl (\overline{A})\in H^1(\Omega)$}}

\put(-35,33){\oval(49,7)}
\put(-56,34){\footnotesize{\bf {Modified couple stress model}}}
\put(-56,30){\footnotesize{\bf {conformally invariant}}}
\put(-56,26){\footnotesize{$\mu\,\|\sym\nabla u\|^2$}}
\put(-56,22){\footnotesize{$+\mu\,L_c^2\|\dev\sym\nabla [\axl (\skw \nabla u)]\|^2$}}
\put(-56,18){\footnotesize{1 curvature parameter}}
\put(-56,14){\footnotesize{couple stress tensor $\widetilde{{m}}\in{\rm Sym}(3)$}}
\put(-56,10){\footnotesize{nonlocal force-stress $\widetilde{\tau}\in\so(3)$}}
\put(-56,6){\footnotesize{unique solution}}
\put(-56,2){\footnotesize{extendable to a}}
\put(-56,-2){\footnotesize{Cosserat formulation}}
\put(-56,-6){\footnotesize{$\axl (\overline{A})\in H^1(\Omega)$}}

\put(-7,84){\footnotesize{\bf \framebox[5cm][l]{Hadjesfandiari-Dargush}}}
\put(-7,78){\footnotesize{$\mu\,\|\sym\nabla u\|^2$}}
\put(-7,74){\footnotesize{$+\mu\,L_c^2\,\|\skw\nabla [\axl (\skw \nabla u)]\|^2$}}
\put(-7,70){\footnotesize{1 curvature parameter}}
\put(-7,66){\footnotesize{couple stress  tensor $\widetilde{ {m}}\in\so(3)$}}
\put(-7,62){\footnotesize{nonlocal force-stress $\widetilde{\tau}\in\so(3)$}}
\put(-7,58){\footnotesize{unique solution}}
\put(-7,54){\footnotesize{extendable to a}}
\put(-7,50){\footnotesize{Cosserat formulation}}
\put(-7,46){\footnotesize{$\axl (\overline{A})\in H(\curl;\Omega)$}}

\put(46,84){\footnotesize{\bf \framebox[4.2cm][l]{$\boldsymbol{\Curl}$-formulation}}}
\put(46,78){\footnotesize{$\mu\,\|\sym\nabla u\|^2$}}
\put(46,74){\footnotesize{$+\mu\,L_c^2\,\|\Curl(\sym\nabla  u)\|^2$}}
\put(46,70){\footnotesize{2 curvature parameters}}
\put(46,66){\footnotesize{couple stress  tensor $\widehat{ {m}}\in\gl(3)$}}
\put(46,62){\footnotesize{nonlocal force-stress $\widehat{\tau}\in{\rm Sym}(3)$}}
\put(46,58){\footnotesize{unique solution}}
\put(46,54){\footnotesize{extendable to a}}
\put(46,50){\footnotesize{microstrain formulation}}
\put(46,46){\footnotesize{$\widehat{\varepsilon}\in H(\Curl;\Omega)$}}
\put(60,45){\vector(0,-1){8}}

\put(-12,14.5){\vector(1,0){56}}
\put(10,14.5){\vector(-1,0){22}}
\put(8,20){\tiny{$\widetilde{ {m}}=\widehat{ {m}}$ for $\alpha_2=0$}}
\put(6,23){\tiny{$\widetilde{ {m}}=-\widehat{ {m}}$ for $\alpha_1=0$}}
\put(0,29){\tiny{$\widetilde{m}=\mu\,L_c^2\, \{\alpha_1\, \dev\sym {\nabla}[{\axl(\skw} \nabla u)]$}}
\put(0,26){\tiny{$\qquad\quad\quad\ \ \  + \alpha_2\,\skw\nabla [\axl (\skw \nabla u)]\}$}}
\put(0,35){\tiny{$\widehat{m}=\mu\,L_c^2\, \{2\, \alpha_1 \dev\sym{\Curl} ({\sym} \nabla u)$}}
\put(0,33){\tiny{$\qquad\quad\quad\ \ \ +2\,\alpha_2\,   \skw\Curl (\sym \nabla u)\}$}}

\put(-12,10){\vector(1,0){56}}
\put(10,10){\vector(-1,0){22}}
\put(17,9){\footnotesize{$\boldsymbol{/}$}}
\put(6,6){\tiny{$\widetilde{ {\tau}}\neq\widehat{ {\tau}}$}}
\put(6,3){\tiny{$\widetilde{ {\tau}}:=\anti{\rm Div}\, \widetilde{m}\in\mathfrak{so}(3)$}}
\put(6,0){\tiny{$\widehat{ {\tau}}:=\sym\Curl\widehat{m}\in{\rm Sym}(3)$}}

\put(71,33){\oval(52,5)}
\put(46,32){\footnotesize{\bf {new conformally-invariant model}}}
\put(46,26){\footnotesize{$\mu\,\|\sym\nabla u\|^2$}}
\put(46,22){\footnotesize{$\ \ +\mu\,L_c^2\,\|\dev\sym\Curl(\sym\nabla  u)\|^2$}}
\put(46,18){\footnotesize{1 curvature parameter}}
\put(46,14){\footnotesize{couple stress  tensor $\widehat{ {m}}\in\Sym(3)$}}
\put(46,10){\footnotesize{nonlocal force-stress $\widehat{\tau}\in{\rm Sym}(3)$}}
\put(46,6){\footnotesize{unique solution, well posed}}
\put(46,2){\footnotesize{extendable to a}}
\put(46,-2){\footnotesize{microstrain formulation}}
\put(46,-6){\footnotesize{$\widehat{\varepsilon}\in H(\Curl;\Omega)$}}

\end{picture}
\end{center}
\caption{Some models involving size effects and their interrelation. The given energy expressions are only meant to represent the main features of the models: $\mu\,\|\sym \nabla u\|^2$ represents the linear elastic energy which may be anisotropic $\langle \mathbb{C}.\sym \nabla u,\sym\nabla u\rangle$, while $\mu\,L_c^2\|\nabla [\axl (\skw \nabla u)]\|^2$ represents the curvature energy, which could also be anisotropic.}\label{figsize}
\end{figure}

\begin{figure}
\setlength{\unitlength}{1mm}
\begin{center}
\begin{picture}(10,23)
\thicklines

\put(3,16){\oval(170,40)}

\put(-80,30){$\widehat{\tau}=\mu\,L_c^2\,\underbrace{\sym}_{} \underbrace{\Curl}_{} \, \{2\, \alpha_1 \dev\sym\underbrace{\Curl}_{} (\underbrace{\sym}_{} \nabla u)+2\,\alpha_2\,   \skw\Curl (\sym \nabla u)\}=\sym\Curl(\widehat{m})\in {\rm Sym}(3)$}

\put(-16,21){\line(-1,0){46.4}}
\put(-16,21){\line(0,1){4.6}}
\put(-62.3,21){\vector(0,1){5.1}}

\put(-24.5,24){\line(-1,0){30.2}}
\put(-24.5,24){\line(0,1){2.6}}
\put(-54.5,24){\vector(0,1){2.6}}

\put(-80,15){$\widetilde{\tau}=\mu\,L_c^2\,\,\underbrace{\anti}_{}
\underbrace{{\rm Div}}_{}\{\alpha_1\, \dev\sym \underbrace{\nabla}_{} \underbrace{\axl(\skw}_{} \nabla u)+ \alpha_2\,\skw\nabla [\axl (\skw \nabla u)]\}=\anti {
\rm Div}(\widetilde{m})\in\so(3)$}

\put(-17.5,7){\line(-1,0){44.6}}
\put(-17.5,11.5){\line(0,-1){4.6}}
\put(-62,7){\vector(0,1){5.1}}

\put(-28.2,10){\line(-1,0){27}}
\put(-28.2,10){\line(0,1){2.6}}
\put(-55,10){\vector(0,1){2.6}}

\put(-10,0){${\rm Div}(\widehat{\tau}-\widetilde{\tau})=0$}
\end{picture}
\end{center}
\caption{The two possibilities of defining a nonlocal force stress tensor: either $\widehat{\tau}$ is symmetric in the $\Curl(\sym \nabla u)$ formulation or $\widetilde{\tau}$ is antisymmetric in the $\nabla [\axl (\skw \nabla u)]$ formulation. The difference between both stresses is a divergence-free stress field (a self-equilibrated force field). }\label{limitmodel2tau}
\end{figure}
The energies in both possible formulations (in terms of $\nabla [\axl (\skw \nabla u)]$  or $\Curl(\sym \nabla u)$)  are the same, differences appear only once traction boundary conditions are specified. The need for prescribing this or that boundary conditions determines which model should be used.

In a polar gradient elasticity model we could influence directly continuum rotations without prescribing $u\big|_{\Gamma}=0$. But this should only be possible in a theory which extends beyond mechanics: for example to magnetic or electric effects, i.e. needed for particular loading and boundary conditions which excite particular micro-rotations (``polarization"). In contrast, in a non-polar elasticity model it is not possible to influence directly continuum rotations but a non-polar model is applicable and much more appropriate in a purely mechanical context (see Figure \ref{figsize}).  The case iv) in Fig. \ref{limitmodel2} needs mathematical discussion. The extension of the well-posedness  to the finite strain case in which the corresponding Lagrangian may be written as $W=W(U)+W_{\rm curv}(U,\Curl U)$, where $F=R\,U$ is the polar decomposition is yet missing. Some steps in this direction are presented in \cite{NeffLankeitOsterbrink}.

\section{Epilogue: Much ado about nothing}
We have seen how  much effort it took us to derive the consistent boundary conditions in the indeterminate couple stress model. The conceptual advantage of not having to discuss the physical meaning of independent degrees of freedom is, now,  more than outweighed by the burdensome interpretation of traction boundary conditions. Nevertheless, all presented formulations are shown to be mathematically well-posed. In the last part of the paper we have had a look at 2nd-order (micromorphic) approximations of the given gradient elastic models. In these micromorphic models, the boundary conditions are completely transparent. However, it seems that in this larger class of models there is yet another variant (the relaxed micromorphic model with integral boundary coupling) which combines conceptual simplicity, symmetry of force stress tensor and symmetry of moment stress tensor, simplicity of traction boundary conditions and well-posedness to make it superior to all other presented formulations. With hindsight,  we understand why the indeterminate couple stress model had been abandoned in the late '60ies.  For us it is a mystery how it was  possible at all to identify material parameters in a theory in which boundary conditions had not been  conclusively settled?

\section*{Acknowledgement} We are grateful to Ali Reza Hadjesfandiari and Gary F. Dargush for sending us the paper \cite{hadjesfandiari2014evo} prior to publication. Discussions with X.L. Gao and  S. Forest  on  a prior version of the paper have been helpful. The ideas for this paper have been discussed at the 50th Annual Technical Meeting/ASME-AMD Annual Summer Meeting, July 2013, at the Brown School of Engineering  with David J. Steigmann and Francesco  dell'Isola. I.D. Ghiba acknowledges support from the Romanian National Authority for Scientific Research (CNCS-UEFISCDI), Project No. PN-II-ID-PCE-2011-3-0521.

\bibliographystyle{plain} %plain
\addcontentsline{toc}{section}{References}

\begin{footnotesize}

\end{footnotesize}

\appendix
\setcounter{section}{0}

\begin{footnotesize}

\section{The traction boundary conditions in the $\boldsymbol{\Curl(\sym \nabla u)}$-formulation and in the $\nabla[\axl(\skew \nabla u)]$-formulation are different}\label{notthesamebc}
\setcounter{equation}{0}

In this section we prove the claim from Subsection  \eqref{curlaxltractform}, i.e. we show that the possible traction boundary conditions in  the $\nabla [\axl(\skew\nabla u)]$-formulation and the $\Curl (\sym \nabla  u)$-formulation are different.

We consider a point $P$ at the boundary and we show that  $(2)\neq (2^\prime)$ in this point.  Without confining, we  consider that the system of coordinates is initially  chosen such that the normal vector on the boundary at this point $P$ is $n=e_1:=(1,0,0)$.  Since there are no derivatives in $(2)$ and $(2^\prime)$, it is enough to prove that
\begin{align}
(\id-e_1\otimes e_1)(\sym \widehat{M}).e_1\neq \frac{1}{2}(\id-e_1\otimes e_1)\anti(\widetilde{ {m}}.\, e_1).e_1.
\end{align}
On one hand, we have
\begin{align}
\sym \widehat{M}.e_1&=\sym\left(
                   \begin{array}{c}
                     (\widehat{m}_{11}e_1+\widehat{m}_{12}e_2+\widehat{m}_{13}e_3)\times e_1 \\
                      (\widehat{m}_{21}e_1+\widehat{m}_{22}e_2+\widehat{m}_{23}e_3)\times e_1 \\
                      (\widehat{m}_{31}e_1+\widehat{m}_{32}e_2+\widehat{m}_{33}e_3)\times e_1
                   \end{array}
                 \right).\, e_1=\sym\left(
                   \begin{array}{c}
                     -\widehat{m}_{12}e_3+\widehat{m}_{13}e_2 \\
                       -\widehat{m}_{22}e_3+\widehat{m}_{23}e_2 \\
                       -\widehat{m}_{32}e_3+\widehat{m}_{33}e_2
                   \end{array}
                 \right).\, e_1\\
                 &=\sym\left(
                   \begin{array}{ccc}
                     0&\widehat{m}_{13}&-\widehat{m}_{12} \\
                        0&\widehat{m}_{23}&-\widehat{m}_{22} \\
                       0&\widehat{m}_{33}&-\widehat{m}_{32}
                   \end{array}
                 \right).\, e_1
                 =\frac{1}{2}\left(
                   \begin{array}{ccc}
                     0&\widehat{m}_{13}&-\widehat{m}_{12} \\
                        \widehat{m}_{13}&\widehat{m}_{23}&-\widehat{m}_{22}+\widehat{m}_{33} \\
                       -\widehat{m}_{12}&-\widehat{m}_{22}+\widehat{m}_{33}&-\widehat{m}_{32}
                   \end{array}
                 \right).\, e_1
                 =\frac{1}{2}\left(
                   \begin{array}{c}
                     0\\
                        \widehat{m}_{13} \\
                       -\widehat{m}_{12}\\
                   \end{array}
                 \right),\notag
\end{align}
and therefore
\begin{align}
(\id-e_1\otimes e_1)(\sym \widehat{M}).e_1=\left(
                                             \begin{array}{ccc}
                                               0 & 0 & 0 \\
                                               0 & 1 & 0 \\
                                               0 & 0 & 1 \\
                                             \end{array}
                                           \right)
\frac{1}{2}\left(
                   \begin{array}{c}
                     0\\
                        \widehat{m}_{13} \\
                       -\widehat{m}_{12}\\
                   \end{array}
                 \right)=\frac{1}{2}\left(
                   \begin{array}{c}
                     0\\
                        \widehat{m}_{13} \\
                       -\widehat{m}_{12}\\
                   \end{array}
                 \right).
\end{align}

On the other hand, we obtain
\begin{align}
\anti(\widetilde{m}.\, e_1).e_1=\anti(\widetilde{m}_{11}, \widetilde{m}_{21},\widetilde{m}_{31}).e_1=\left(
                                                                      \begin{array}{ccc}
                                                                        0 & \widetilde{m}_{31} & -\widetilde{m}_{21} \\
                                                                        -\widetilde{m}_{31} & 0 & \widetilde{m}_{11} \\
                                                                        \widetilde{m}_{21} & -\widetilde{m}_{11} & 0
                                                                      \end{array}
                                                                    \right).e_1=\left(
                                                                      \begin{array}{c}
                                                                        0 \\
                                                                        -\widetilde{m}_{31}  \\
                                                                        \widetilde{m}_{21}
                                                                      \end{array}
                                                                    \right)\,.
\end{align}
Hence, we deduce
\begin{align}
(\id-e_1\otimes e_1)\anti(\widetilde{m}.\, e_1).e_1=\left(
                                             \begin{array}{ccc}
                                               0 & 0 & 0 \\
                                               0 & 1 & 0 \\
                                               0 & 0 & 1 \\
                                             \end{array}
                                           \right)\left(
                                                                      \begin{array}{c}
                                                                        0 \\
                                                                        -\widetilde{m}_{31}  \\
                                                                        \widetilde{m}_{21}
                                                                      \end{array}
                                                                    \right)=\left(\begin{array}{c}
                                                                        0 \\
                                                                        -\widetilde{m}_{31}  \\
                                                                        \widetilde{m}_{21}
                                                                      \end{array}
                                                                    \right).
\end{align}
We may conclude that $(2)=(2^\prime)$ implies
\begin{align}\label{symm1331}
\widehat{m}_{13}=-\widetilde{m}_{31}, \qquad \widehat{m}_{23}=-\widetilde{m}_{32}.
\end{align}

Let us now point out that $\widehat{m}$ and $\widetilde{m}$ are not independent, see Figure  \ref{figsize}. Considering the case $\alpha_2=0$ we have $\widehat{m}=\widetilde{m}\in {\rm Sym}(3)$, while considering $\alpha_1=0$ we have $\widehat{m}=-\widetilde{m}\in \so(3)$. Therefore,  in the conformal invariant case $\alpha_2=0$ and also in the case $\alpha_1=0$, since the condition \eqref{symm1331} does not hold true, it follows that $(2)\neq(2^\prime)$.

If the boundary conditions imply the continuity of $\widetilde{m}$ and $\widehat{m}$, then $(3)=(3^\prime)=0$. However, if $\widetilde{m}$ and $\widehat{m}$ are not continuous across the curve $\partial \Gamma$, considering again a point $P\in \partial \Gamma$  and considering, without confining,  that the system of coordinates is initially chosen such that the normal vector on the boundary at this point $P$ is $n=e_1:=(1,0,0)$ and $\nu=e_3:=(0,0,1)$, we prove that
\begin{align}
\sym \widehat{M}.e_3\neq\anti(\widetilde{ {m}}.\, e_1).e_3.
\end{align}
Doing similar calculations as above, we deduce
\begin{align}
\sym \widehat{M}.e_3=\frac{1}{2}\left(
                   \begin{array}{ccc}
                     0&\widehat{m}_{13}&-\widehat{m}_{12} \\
                        \widehat{m}_{13}&\widehat{m}_{23}&-\widehat{m}_{22}+\widehat{m}_{33} \\
                       -\widehat{m}_{12}&-\widehat{m}_{22}+\widehat{m}_{33}&-\widehat{m}_{32}
                   \end{array}
                 \right).\, e_3=\frac{1}{2}\left(
                   \begin{array}{ccc}
                     -\widehat{m}_{12} \\
                       -\widehat{m}_{22}+\widehat{m}_{33} \\
                       -\widehat{m}_{32}
                   \end{array}
                 \right)
\end{align}
and
\begin{align}
\anti(\widetilde{ {m}}.\, e_1).e_3=\left(
                                                                      \begin{array}{ccc}
                                                                        0 & \widetilde{m}_{31} & -\widetilde{m}_{21} \\
                                                                        -\widetilde{m}_{31} & 0 & \widetilde{m}_{11} \\
                                                                        \widetilde{m}_{21} & -\widetilde{m}_{11} & 0
                                                                      \end{array}
                                                                    \right).e_3=\left(
                                                                      \begin{array}{ccc}
                                                                        -\widetilde{m}_{21} \\
                                                                         \widetilde{m}_{11} \\
                                                                         0
                                                                      \end{array}
                                                                    \right).
\end{align}
We remark that in both particular cases, the conformal invariance model $\alpha_2=0$ and the case $\alpha_1=0$, we deduce that $\sym \widehat{M}.e_3\neq\anti(\widetilde{ {m}}.\, e_1).e_3$, in general. However, even if $\sym \widehat{M}.e_3\neq\anti(\widetilde{ {m}}.\, e_1).e_3$, the jump $(c)$ may coincide with the jump $(c^\prime)$.

Let us remark that in order to compare $(1)$ and $(1^\prime)$ we may not proceed as above. However, we will prove that $(1)\neq (1^\prime)$ in a specific situation. We assume that there is an open subset $\omega\subset\partial \Omega$ such that on $\omega$ the normal vector $n$ is constant. Let us consider a point $P\in \omega$. We may assume for simplicity that $n=e_1$ at all points $P\in \omega$. Upon this assumption on the domain $\Omega$, at the point $P\in \omega$ we obtain
\begin{align}
(\nabla[(\sym \widehat{M})\,T]: T&=
\nabla[\frac{1}{2}\left(
                   \begin{array}{ccc}
                     0&\widehat{m}_{13}&-\widehat{m}_{12} \\
                        \widehat{m}_{13}&\widehat{m}_{23}&-\widehat{m}_{22}+\widehat{m}_{33} \\
                       -\widehat{m}_{12}&-\widehat{m}_{22}+\widehat{m}_{33}&-\widehat{m}_{32}
                   \end{array}
                 \right)\,\left(
                                             \begin{array}{ccc}
                                               0 & 0 & 0 \\
                                               0 & 1 & 0 \\
                                               0 & 0 & 1 \\
                                             \end{array}
                                           \right)].\,\left(
                                             \begin{array}{ccc}
                                               0 & 0 & 0 \\
                                               0 & 1 & 0 \\
                                               0 & 0 & 1 \\
                                             \end{array}
                                           \right)\notag\\&=\frac{1}{2}\nabla[\left(
                   \begin{array}{ccc}
                     0&\widehat{m}_{13}&-\widehat{m}_{12} \\
                       0&\widehat{m}_{23}&-\widehat{m}_{22}+\widehat{m}_{33} \\
                       0&-\widehat{m}_{22}+\widehat{m}_{33}&-\widehat{m}_{32}
                   \end{array}
                 \right)].\,\left(
                                             \begin{array}{ccc}
                                               0 & 0 & 0 \\
                                               0 & 1 & 0 \\
                                               0 & 0 & 1 \\
                                             \end{array}
                                           \right)\\
                                           &=\frac{1}{2}\left(
                   \begin{array}{c}
                     \widehat{m}_{13,2}-\widehat{m}_{12,3} \\
                     \widehat{m}_{23,2}-\widehat{m}_{22,3}+\widehat{m}_{33,3} \\
                     -\widehat{m}_{22,2}+\widehat{m}_{33,2}-\widehat{m}_{32,3}
                   \end{array}
                 \right)\notag
\end{align}
and
\begin{align}
&\sym\Curl (\widehat{ {m}}).\, n=\sym\left(
                                       \begin{array}{ccc}
                                         \widehat{ {m}}_{13,2}-\widehat{ {m}}_{12,3} & \widehat{ {m}}_{11,3}-\widehat{ {m}}_{13,1} & \widehat{ {m}}_{12,1}-\widehat{ {m}}_{11,2} \\
                                        \widehat{ {m}}_{23,2}-\widehat{ {m}}_{22,3} & \widehat{ {m}}_{21,3}-\widehat{ {m}}_{23,1} & \widehat{ {m}}_{22,1}-\widehat{ {m}}_{21,2} \\
                                         \widehat{ {m}}_{33,2}-\widehat{ {m}}_{32,3} & \widehat{ {m}}_{31,3}-\widehat{ {m}}_{33,1} & \widehat{ {m}}_{32,1}-\widehat{ {m}}_{31,2} \\
                                       \end{array}
                                     \right).e_1\\
                                     &=\frac{1}{2}\,\left(
                                       \begin{array}{ccc}
                                         \widehat{ {m}}_{13,2}-\widehat{ {m}}_{12,3} & \widehat{ {m}}_{11,3}-\widehat{ {m}}_{13,1}+\widehat{ {m}}_{23,2}-\widehat{ {m}}_{22,3} & \widehat{ {m}}_{12,1}-\widehat{ {m}}_{11,2}+\widehat{ {m}}_{33,2}-\widehat{ {m}}_{32,3} \\
                                        \widehat{ {m}}_{11,3}-\widehat{ {m}}_{13,1}+\widehat{ {m}}_{23,2}-\widehat{ {m}}_{22,3} & \widehat{ {m}}_{21,3}-\widehat{ {m}}_{23,1} & \widehat{ {m}}_{22,1}-\widehat{ {m}}_{21,2}+\widehat{ {m}}_{31,3}-\widehat{ {m}}_{33,1}  \\
                                         \widehat{ {m}}_{12,1}-\widehat{ {m}}_{11,2}+\widehat{ {m}}_{33,2}-\widehat{ {m}}_{32,3} & \widehat{ {m}}_{22,1}-\widehat{ {m}}_{21,2}+\widehat{ {m}}_{31,3}-\widehat{ {m}}_{33,1} & \widehat{ {m}}_{32,1}-\widehat{ {m}}_{31,2} \\
                                       \end{array}
                                     \right).e_1\notag\\
                                     &=\frac{1}{2}\,\left(
                                       \begin{array}{ccc}
                                         \widehat{ {m}}_{13,2}-\widehat{ {m}}_{12,3}  \\
                                        \widehat{ {m}}_{11,3}-\widehat{ {m}}_{13,1}+\widehat{ {m}}_{23,2}-\widehat{ {m}}_{22,3}  \\
                                         \widehat{ {m}}_{12,1}-\widehat{ {m}}_{11,2}+\widehat{ {m}}_{33,2}-\widehat{ {m}}_{32,3}
                                       \end{array}
                                     \right).\notag
\end{align}
Therefore, we deduce
\begin{align}
\sym\Curl (\widehat{ {m}}).\, n-(\nabla[(\sym \widehat{M})\,T]: T&=\frac{1}{2}\,\left(
                                       \begin{array}{ccc}
                                         \widehat{ {m}}_{13,2}-\widehat{ {m}}_{12,3}  \\
                                        \widehat{ {m}}_{11,3}-\widehat{ {m}}_{13,1}+\widehat{ {m}}_{23,2}-\widehat{ {m}}_{22,3}  \\
                                         \widehat{ {m}}_{12,1}-\widehat{ {m}}_{11,2}+\widehat{ {m}}_{33,2}-\widehat{ {m}}_{32,3}
                                       \end{array}
                                     \right)-\frac{1}{2}\left(
                   \begin{array}{c}
                     \widehat{m}_{13,2}-\widehat{m}_{12,3} \\
                     \widehat{m}_{23,2}-\widehat{m}_{22,3}+\widehat{m}_{33,3} \\
                     -\widehat{m}_{22,2}+\widehat{m}_{33,2}-\widehat{m}_{32,3}
                   \end{array}
                 \right)\notag\\
                 &
                 =\frac{1}{2}\,\left(
                                       \begin{array}{ccc}
                                         \widehat{ {m}}_{13,2}-\widehat{ {m}}_{12,3}-\widehat{m}_{13,2}+\widehat{m}_{12,3}  \\
                                        \widehat{ {m}}_{11,3}-\widehat{ {m}}_{13,1}+\widehat{ {m}}_{23,2}-\widehat{ {m}}_{22,3}-\widehat{m}_{23,2}+\widehat{m}_{22,3}-\widehat{m}_{33,3}  \\
                                         \widehat{ {m}}_{12,1}-\widehat{ {m}}_{11,2}+\widehat{ {m}}_{33,2}-\widehat{ {m}}_{32,3}+\widehat{m}_{22,2}-\widehat{m}_{33,2}+\widehat{m}_{32,3}
                                       \end{array}
                                     \right)
                                     \\
                 &
                 =\frac{1}{2}\,\left(
                                       \begin{array}{ccc}
                                         \widehat{ {m}}_{13,2}-\widehat{ {m}}_{12,3}-\widehat{m}_{13,2}+\widehat{m}_{12,3}  \\
                                        \widehat{ {m}}_{11,3}-\widehat{ {m}}_{13,1}-\widehat{m}_{33,3}  \\
                                         \widehat{ {m}}_{12,1}-\widehat{ {m}}_{11,2}+\widehat{m}_{22,2}
                                       \end{array}
                                     \right).\notag
\end{align}
Moreover, we obtain
\begin{align}
&\frac{1}{2}\nabla[(\anti(\widetilde{ {m}}.\, n))\,T]: T
=\frac{1}{2}\nabla[(\anti(\widetilde{ {m}}.\, e_1))\,T]: T=\frac{1}{2}\nabla[(\anti(\widetilde{ {m}}_{11},\widetilde{ {m}}_{21},\widetilde{ {m}}_{31} ))\,T]: T\\
&=\frac{1}{2}\nabla[\left(
                                                                      \begin{array}{ccc}
                                                                        0 & \widetilde{m}_{31} & -\widetilde{m}_{21} \\
                                                                        -\widetilde{m}_{31} & 0 & \widetilde{m}_{11} \\
                                                                        \widetilde{m}_{21} & -\widetilde{m}_{11} & 0
                                                                      \end{array}
                                                                    \right)\,\,\left(
                                             \begin{array}{ccc}
                                               0 & 0 & 0 \\
                                               0 & 1 & 0 \\
                                               0 & 0 & 1 \\
                                             \end{array}
                                           \right)].\,T\notag\\
&=\frac{1}{2}\nabla[\left(
                                                                      \begin{array}{ccc}
                                                                        0 & \widetilde{m}_{31} & -\widetilde{m}_{21} \\
                                                                       0 & 0 & \widetilde{m}_{11} \\
                                                                        0 & -\widetilde{m}_{11} & 0
                                                                      \end{array}
                                                                    \right)].\,\,\,\left(
                                             \begin{array}{ccc}
                                               0 & 0 & 0 \\
                                               0 & 1 & 0 \\
                                               0 & 0 & 1 \\
                                             \end{array}
                                           \right)=\frac{1}{2}\left(
                                                                      \begin{array}{c}
                                                                         \widetilde{m}_{31,2}  -\widetilde{m}_{21,3} \\
                                                                        \widetilde{m}_{11,3} \\
                                                                         -\widetilde{m}_{11,2}
                                                                      \end{array}
                                                                    \right)\notag
\end{align}
and
\begin{align}
\frac{1}{2}\anti
{\rm Div}[\widetilde{ {m}}].\, n&=\frac{1}{2}\anti(\widetilde{ {m}}_{1j,j},\widetilde{ {m}}_{2j,j}, \widetilde{ {m}}_{3j,j}).e_1
=\frac{1}{2}\left(
                                                                      \begin{array}{ccc}
                                                                        0 & \widetilde{m}_{3j,j} & -\widetilde{m}_{2j,j} \\
                                                                        -\widetilde{m}_{3j,j} & 0 & \widetilde{m}_{1j,j} \\
                                                                        \widetilde{m}_{2j,j} & -\widetilde{m}_{1j,j} & 0
                                                                      \end{array}
                                                                    \right).e_1\notag\\
&=\frac{1}{2}\left(
                                                                      \begin{array}{c}
                                                                        0  \\
                                                                        -\widetilde{m}_{3j,j}  \\
                                                                        \widetilde{m}_{2j,j}
                                                                      \end{array}
                                                                    \right)
\end{align}
Hence, it follows that
\begin{align}
-\frac{1}{2}\anti
{\rm Div}[\widetilde{ {m}}].\, n-\frac{1}{2}\nabla[(\anti(\widetilde{ {m}}.\, n))\,T]: T&=-\frac{1}{2}\left(
                                                                      \begin{array}{c}
                                                                        0  \\
                                                                        -\widetilde{m}_{3j,j}  \\
                                                                        \widetilde{m}_{2j,j}
                                                                      \end{array}
                                                                    \right)-\frac{1}{2}\left(
                                                                      \begin{array}{c}
                                                                         \widetilde{m}_{31,2}  -\widetilde{m}_{21,3} \\
                                                                        \widetilde{m}_{11,3} \\
                                                                         -\widetilde{m}_{11,2}
                                                                      \end{array}
                                                                    \right)\notag\\
&=-\frac{1}{2}\left(
                                                                      \begin{array}{c}
                                                                         \widetilde{m}_{31,2}  -\widetilde{m}_{21,3}  \\
                                                                        -\widetilde{m}_{3j,j}+ \widetilde{m}_{11,3}  \\
                                                                        \widetilde{m}_{2j,j}-\widetilde{m}_{11,2}
                                                                      \end{array}
                                                                    \right).
\end{align}
Concluding, $(1)=(1^\prime)$ if and only if
\begin{align}
\left(
                                       \begin{array}{ccc}
                                         \widehat{ {m}}_{13,2}-\widehat{ {m}}_{12,3}  \\
                                        \widehat{ {m}}_{11,3}-\widehat{ {m}}_{13,1}+\widehat{ {m}}_{23,2}-\widehat{ {m}}_{22,3}  \\
                                         \widehat{ {m}}_{12,1}-\widehat{ {m}}_{11,2}+\widehat{ {m}}_{33,2}-\widehat{ {m}}_{32,3}
                                       \end{array}
                                     \right)=
                                    \left(
                                                                      \begin{array}{c}
                                                                         -\widetilde{m}_{31,2}  +\widetilde{m}_{21,3}  \\
                                                                      \widetilde{m}_{3j,j}- \widetilde{m}_{11,3}  \\
                                                                        -\widetilde{m}_{2j,j}+\widetilde{m}_{11,2}
                                                                      \end{array}
                                                                    \right),
\end{align}
which holds not true, in general. In the conformal case, the above condition reads
\begin{align}
\widehat{ {m}}_{13,2}=\widehat{ {m}}_{12,3}, \qquad 2\, \widehat{ {m}}_{11,3}-2\,\widehat{ {m}}_{13,1}=\widehat{ {m}}_{22,3}+\widehat{ {m}}_{33,3}
\qquad  2\, \widehat{ {m}}_{11,2}-2\,\widehat{ {m}}_{12,1}=\widehat{ {m}}_{33,2}+\widehat{ {m}}_{22,2},
\end{align}
while in the case $\alpha_1=0$ it becomes
\begin{align}
\widehat{ {m}}_{23,2}=0,\qquad
\widehat{ {m}}_{32,3}=0\quad \Rightarrow \quad \widehat{ {m}}_{23}=\widehat{ {m}}_{23}(x_1),
\end{align}
which is clearly not satisfied, in general.

\section{From second order couple stress tensors to third order moment stress tensors and back}\setcounter{equation}{0}

Let us consider the  general anisotropic case and
\begin{align*}
W_{\rm curv}(D^2u)=\langle \mathbb{C}. D^2u,D^2 u\rangle_{\mathbb{R}^{3\times3\times3}}, \qquad \widehat{W}_{\rm curv}(\Curl(\sym \nabla u))=\langle \mathbb{L}. \Curl(\sym \nabla u),\Curl(\sym \nabla u)\rangle_{\mathbb{R}^{3\times3}},
\end{align*}
where
 \begin{align*}
 \mathbb{C}=(\mathbb{C}_{ijklmn}):\mathbb{R}^{3\times3\times 3}\rightarrow\mathbb{R}^{3\times3\times 3} \ \ \text{ and }\ \ \mathbb{L}=(\mathbb{L}_{ijkl}):\mathbb{R}^{3\times3}\rightarrow\mathbb{R}^{3\times3}.
 \end{align*}
Let us also consider the tensors
 \begin{align*}
 \mathfrak{m}:=D_{D^2u} W_{\rm curv}(D^2u)\ \ \text{ and }\ \ \widehat{m}:=D_{\Curl(\sym \nabla u)} \widehat{W}_{\rm curv}(\Curl(\sym \nabla u)),
 \end{align*}
 which for our anisotropic case are
 \begin{align*}
 \mathfrak{m}:=\mathbb{C}. D^2u\ \ \text{ and }\ \ \widehat{m}:=\mathbb{L}.\Curl(\sym \nabla u).
 \end{align*}

 Since
 \begin{align*}
D^2 u={\rm Lin}(\nabla\, (\sym   \nabla   u))=\mathbb{A}.(\nabla\, (\sym   \nabla   u)),\qquad  u_{k,ij}=\varepsilon_{ik,j}-\varepsilon_{jk,i}-\varepsilon_{ij,k},
 \end{align*}
 where $\varepsilon=\sym \nabla u$, we obtain
 \begin{align*}
 \mathfrak{m}:=\mathbb{C}\mathbb{A}.(\nabla\, (\sym   \nabla   u))=\mathbb{B}.(\nabla\, (\sym   \nabla   u)).
 \end{align*}
 The next problem is to find  particular form of the tensor $\mathbb{C}$, for which we have
 \begin{align*}
 \langle\mathfrak{m}, D^2u\rangle_{\mathbb{R}^{3\times 3\times3}}=\langle\widetilde{m}, \Curl (\sym \nabla  u)\rangle_{\mathbb{R}^{3\times 3}},
 \end{align*}
 or equivalently
 \begin{align*}
 \langle \mathbb{C}\mathbb{A}.(\nabla\, (\sym   \nabla   u)), \mathbb{A}.(\nabla\, (\sym   \nabla   u))\rangle_{\mathbb{R}^{3\times 3\times3}}=\langle \mathbb{L}. \Curl (\sym \nabla  u), \Curl (\sym \nabla  u)\rangle_{\mathbb{R}^{3\times 3}}.
 \end{align*}
 This is equivalent to
  \begin{align*}
 \langle \mathbb{A}^{T}\mathbb{C}\mathbb{A}.(\nabla\, (\sym   \nabla   u)), (\nabla\, (\sym   \nabla   u))\rangle_{\mathbb{R}^{3\times 3\times3}}=\langle \mathbb{L}. \Curl (\sym \nabla  u), \Curl (\sym \nabla  u)\rangle_{\mathbb{R}^{3\times 3}}.
 \end{align*}
 Let us denote in the following
 \begin{align*}
 \mathbb{B}:=\mathbb{A}^{T}\mathbb{C}\,\mathbb{A}.
 \end{align*}
We consider a specific form of the tensor $\mathbb{B}$ in terms of  another tensor $\mathbb{L}:\mathbb{R}^{3\times 3}\rightarrow \mathbb{R}^{3\times 3}$ such that
\begin{align*}
 \langle \mathbb{B}.(\nabla\, (\sym   \nabla   u)), (\nabla\, (\sym   \nabla   u))\rangle_{\mathbb{R}^{3\times 3\times3}}=\langle \mathbb{L}. \Curl (\sym \nabla  u), \Curl (\sym \nabla  u)\rangle_{\mathbb{R}^{3\times 3}}.
 \end{align*}
Let us show how to obtain the tensor $\mathbb{B}$ if $\mathbb{L}$ is given, such that the last identity holds true.   We first remark that a tensor $\mathbb{B}:\mathbb{R}^{3\times 3\times 3}\rightarrow \mathbb{R}^{3\times 3\times 3}$ is uniquely  defined by the fourth order tensors
\begin{align}\label{lt1}
\mathbb{B}_{im}=({\mathbb{B}}_{ijkmnp})_{im}, \qquad {\mathbb{B}}_{im}:\mathbb{R}^{3\times3}\rightarrow\mathbb{R}^{3\times3},
\end{align}
and
\begin{align*}
\langle {\mathbb{B}}.\, \nabla (\sym \nabla u),\nabla (\sym \nabla u)\rangle_{\mathbb{R}^{3\times 3\times 3}}&:=\langle {\mathbb{B}}_{im}.\nabla (\sym \nabla u)_m,\nabla (\sym \nabla u)_i\rangle_{\mathbb{R}^{3\times 3}},
\end{align*}
where Einstein's summation rule is used.
Let  ${\mathbb{L}}:\mathbb{R}^{3\times 3}\rightarrow\mathbb{R}^{3\times 3}$ be  a given  fourth order tensor. We may write this tensor in the form
\begin{align}\label{lt4}
{\mathbb{L}}=( \widehat{\mathbb{L}}^1, \widehat{\mathbb{L}}^2, \widehat{\mathbb{L}}^3), \qquad \widehat{\mathbb{L}}^i:\mathbb{R}^{3}\rightarrow\mathbb{R}^{3}.
\end{align}
Let us define the tensor ${\mathbb{B}}$ by \eqref{lt1} where
\begin{align}\label{lt2}
{\mathbb{B}}_{im}=\left\{\begin{array}{lll}2\,
\skew \anti \widehat{\mathbb{L}}^i\axl \skew & \text{for}& i=m\vspace{1.5mm}\\
0& \text{for}& i\neq m.
\end{array}\right.\,
\end{align}
For the particular form \eqref{lt2} of ${\mathbb{B}}_{im}$, using  the formula $2\,\axl \skew \nabla (\sym \nabla u)_i=\curl (\sym \nabla u)_i$,  we obtain
 \begin{align*}
\langle {\mathbb{B}}_{im}.\nabla (\sym \nabla u)_m,\nabla (\sym \nabla u)_i\rangle_{\mathbb{R}^{3\times 3}}&=2\,\langle \skew \anti \widehat{\mathbb{L}}^i. \axl \skew \nabla (\sym \nabla u)_i,\nabla (\sym \nabla u)_i\rangle_{\mathbb{R}^{3\times 3}}\\&=2\,\langle \anti \widehat{\mathbb{L}}^i.\axl \skew \nabla (\sym \nabla u)_i,\skew\nabla (\sym \nabla u)_i\rangle_{\mathbb{R}^{3\times 3}}\notag\\
&=4\,\langle \widehat{\mathbb{L}}^i.\axl \skew \nabla (\sym \nabla u)_i,\axl\skew\nabla (\sym \nabla u)_i\rangle_{\mathbb{R}^{3}}\notag\\&=\langle \widehat{\mathbb{L}}^i.\curl (\sym \nabla u)_i,\curl (\sym \nabla u)_i\rangle_{\mathbb{R}^{3}}.\notag
\end{align*}
Now, we conclude that for a given  fourth order tensor ${\mathbb{L}}:\mathbb{R}^{3\times 3}\rightarrow\mathbb{R}^{3\times 3}$ the
 tensor $\widetilde{\mathbb{L}}$  defined by \eqref{lt1} and \eqref{lt2} is such that
\begin{align}\label{lt3}
\langle {\mathbb{B}}.\, \nabla (\sym \nabla u),\nabla (\sym \nabla u)\rangle_{\mathbb{R}^{3\times 3\times 3}}&=\langle {\mathbb{B}}_{im}.\nabla (\sym \nabla u)_m,\nabla (\sym \nabla u)_i\rangle_{\mathbb{R}^{3\times 3}}\\&=\langle \widehat{\mathbb{L}}^i.\curl (\sym \nabla u)_i,\curl (\sym \nabla u)_i\rangle_{\mathbb{R}^{3}}=\langle {\mathbb{L}}.\Curl (\sym \nabla u),\Curl (\sym \nabla u)\rangle_{\mathbb{R}^{3\times3}}.\notag
\end{align}

In conclusion, we have found a tensor $\mathbb{C}$ given by
 \begin{align*}
 \mathbb{C}:=\mathbb{A}\mathbb{B}\,\mathbb{A}^{T},
 \end{align*}
 where $\mathbb{B}$ is given by \eqref{lt1} and \eqref{lt2}, such that
\begin{align*}
 \langle \mathbb{C}.D^2u, D^2 u\rangle_{\mathbb{R}^{3\times 3\times 3}}=\langle \mathbb{L}. \Curl (\sym \nabla  u), \Curl (\sym \nabla  u)\rangle_{\mathbb{R}^{3\times 3}},
 \end{align*}
 or equivalently
 \begin{align*}
 \langle \mathfrak{m}, D^2 u\rangle_{\mathbb{R}^{3\times 3\times 3}}=\langle \widehat{m}, \Curl (\sym \nabla  u)\rangle_{\mathbb{R}^{3\times 3}},
 \end{align*}
  or equivalently
   \begin{align*}
W_{\rm curv}(D^2 u)=\widehat{W}_{\rm curv}(\Curl (\sym \nabla  u)).
 \end{align*}

 \section{The name  of the indeterminate couple stress model}\label{nicm}\setcounter{equation}{0}

Regarding the name of the indeterminate couple stress model, Paria \cite[p. 1]{paria} writes: ``...it has led to the difficulties that the anti-symmetric part of the stress dyadic as well as the isotropic part of the couple-stress dyadic remain indeterminate. These indeterminacies are perhaps due to the fact that the rotation vector, defined above, is not independent but depends on the displacement vector".  The theory has a variety of names, such as ``Cosserat theory with constrained rotations" (Toupin, 1964), ``Couple stress theory" (Koiter, 1964), ``Indeterminate couple stress theory" (Eringen, 1968), ``Cosserat pseudo-continuum" (Nowacki, 1968). Eringen writes \cite{Eringen99}: ``At this time [in the 1960, our addition] also popular was a theory of indeterminate couple stress which is mostly  abandoned now [1998]. In this theory, the axisymmetric [skew-symmetric] part of the stress tensor is redundant and it remains indeterminate". Sch\"afer \cite{schaefer1967cosserat} called ``indeterminate couple stress model" as  pseudo-Cosserat-continuum of the tri\'edre caches (see also \cite{Nowacki86}).

If the microrotations $\overline{A}\in \so(3)$ are constrained to be equal to the macrorotations $\skew \nabla u$, the Cosserat model reduces to the couple stress theory. This corresponds to the case $\mu_c\rightarrow\infty$, for which the antisymmetric part of the strain tensor $\skew(\nabla u-\overline{A})$ and the spherical part of the curvature tensor $\tr(\nabla\axl(\skew \overline{A}))$ tend to zero. Consequently, by energetic duality the antisymmetric part of the Cauchy stress tensor $\skew(\sigma)$ in the Cosserat model, and the first invariant of the couple stress, namely $\tr(\widetilde{m})$ do not appear in the formulation of the virtual work principle as well   as in the constitutive equations. The first invariant of the couple stress remains ``indeterminate" and it is taken to be equal to zero \cite{Koiter64}. Now, the skew-symmetric part of the total force stress tensor now is not constitutively determined, but can be obtained from balance of momentum.

\end{footnotesize}
\end{document}